\documentclass[11pt]{article}
\usepackage[utf8]{inputenc}

\usepackage[margin=1.25in]{geometry}
\usepackage{fancyhdr, titling}

\usepackage{times}
\newcommand{\arxivversion}{}

\usepackage{arxiv_macros}
\usepackage{common}

\usepackage[hyperpageref]{backref}

\setcitestyle{numeric}
\bibpunct{(}{)}{;}{a}{,}{,}

\title{More powerful multiple testing under dependence\\ via randomization}

\author{
Ziyu Xu\thanks{Department of Statistics and Data Science, Carnegie Mellon University, USA. Email: \texttt{xzy@cmu.edu}.}
\and
Aaditya Ramdas\thanks{Departments of Statistics and Data Science, and Machine Learning, Carnegie Mellon University, USA. Email: \texttt{aramdas@cmu.edu}.}
}

\date{\today}

\begin{document}
\maketitle
\begin{abstract}
    We develop a technique to improve the power of any e-value by a simple randomization involving one independent uniform random variable.
Using this framework, we show that two procedures for false discovery rate (FDR) control --- the Benjamini-Yekutieli procedure for dependent p-values, and the e-Benjamini-Hochberg procedure for dependent e-values --- can be improved through randomization.
We also improve the Hommel test under dependence, and post-selection inference procedures for confidence intervals with false coverage rate (FCR) that allow for arbitrary selection rules and dependence.
Importantly, our randomized improvements are never worse than the originals and are typically strictly more powerful, with marked improvements in simulations.
 \end{abstract}
\tableofcontents

\section{Introduction}
In the multiple testing problem, we wish to test $K$ hypotheses, out of which an unknown subset $\mathcal N \subseteq [K] \coloneqq \{1, \dots, K
\}$ contains those where the null hypothesis is true. Let $\mathcal D$ denote the set of rejections (``discoveries'') made by a multiple testing procedure. Then, $\mathcal N \cap \mathcal D$ are the ``false discoveries''. The false discovery rate (FDR) is defined as follows:
\[
\FDR \coloneqq \expect\left[ \frac{|\mathcal D \cap \mathcal N|}{|\mathcal D| \vee 1}\right],
\]
where $a \vee b := \max(a,b)$. The goal of an FDR controlling multiple testing method is to always ensure $\FDR \leq \alpha$ for a predefined level $\alpha \in [0, 1]$.

One of the seminal results in multiple testing under arbitrary dependence was developed by \cite{benjamini_control_false_2001}. They showed that if the $K$ input p-values are arbitrarily dependent, then the Benjamini-Hochberg (BH) procedure \citep{benjamini_controlling_false_1995} is only guaranteed to control the false discovery rate (FDR) if it is run at the target level $\alpha$ divided by an approximately $\log K$ factor. We refer to this more conservative version of the BH procedure --- which still ensures FDR control under arbitrary dependence --- as the \emph{BY procedure}. The BY procedure is known to be unimprovable in the sense that there is an ``extreme'' example for which it achieves FDR equal to $\alpha$ \citep{guo_control_false_2008}. Here, we show the surprising fact that outside of this extreme example, its power can be (usually strictly) improved using independent, external randomization.

Despite the above procedure using only p-values and having nothing to do with ``e-values'', the key idea for its improvement stems from the latter concept. In particular, our technique involves a simple stochastic rounding of e-values.
A random variable $X$ is said to be an \emph{e-value} if $X$ is nonnegative and $\expect[X] \leq 1$ under the null hypothesis. E-values have been recognized as a powerful tool for hypothesis testing in the last few years \citep{wasserman_universal_inference_2020,grunwald_safe_testing_2020,shafer_testing_betting_2021,vovk_e-values_calibration_2021,howard_time-uniform_chernoff_2020}.
For example, the universal inference statistic provides an e-value for testing any composite null without regularity conditions~\citep{wasserman_universal_inference_2020}; as a result, one can now construct tests using e-values in problems where no prior test was known. One such example is testing the log-concavity of a distribution; the only known valid tests are based on e-values \citep{dunn_universal_inference_2022}. Many other examples exist \citep{ramdas_gametheoretic_statistics_2022}.

E-values have an intrinsic connection to anytime-valid inference --- ``anytime-valid'' refers to the sequential nature of sampling data, where one can choose when to stop and make an inference based on the collected data. Admissible anytime-valid tests (or confidence intervals) must be derived from nonnegative martingales, which are e-values \citep{ramdas_admissible_anytime-valid_2020}.  

Recently, \cite{wang_false_discovery_2022} showed that the BY procedure is intimately connected to the \emph{e-BH procedure}, which is the e-value analog of the BH procedure for FDR control. Our paper also improves the e-BH procedure using a similar randomization technique.
To recap the e-BH procedure, let $X_1,\dots, X_K$ be $K$ e-values (testing $K$ different null hypotheses), and
let $X_{[i]}$ refer to the $i$-th order statistic of the e-values when sorted in \emph{decreasing} order (so $X_{[i]}$ is the $i$-th largest e-value). Define
\begin{align}
k^* \coloneqq \max\left\{i \in [K]: X_{[i]} \geq \frac{K}{i \alpha}\right\},\label{eq:ebh-kstar}
\end{align}
where the max of an empty set is defined as 0. The e-BH procedure rejects the hypotheses corresponding to the $k^*$ largest e-values, and we denote this set of discoveries as $\Dcal_{\rmEBH}$.  \cite{wang_false_discovery_2022} showed that the e-BH procedure controls FDR at $\alpha$ under any arbitrary dependence between the e-values (a fact that is not true for the BH procedure based on p-values, which is why the corrected BY procedure is needed).

It turns out that one can view the e-BH procedure as selecting a data-dependent threshold
\begin{align}
    \alphath^* \coloneqq \alpha(k^* + 1) / K
\end{align}
and rejecting (i.e., including in $\Dcal_{\rmEBH}$) the $i$th hypothesis if and only if $X_i \geq 1 / \alphath^*$ (indeed, if $k^*=K$, the claim is trivially true, and if $k^*<K$, then there can only be $k^*$ hypotheses with e-values larger than $K/(\alpha(k^*+1))$, otherwise we violate the maximality of $k^*$ in its definition). Now, observe that $1 / \alphath^*$ can only take values from the grid $\{K / (k\alpha): k \in [K]\}$.
If we ``rounded'' down each $X_i$ to the closest value in the grid that is less than $X_i$ (or 0 if $X_i$ is smaller than any value in the grid), the discovery set that is output by e-BH would be identical to the one where no rounding had occurred. But if we could round \emph{up} the e-value, we would potentially gain power; however, this would inflate its expectation, and it would no longer be an e-value. Our key insight is that if we \emph{randomly} rounded every e-value up or down --- appropriately so that its expectation is unchanged --- then we could increase power with positive probability.

The fact that most e-values are continuous and will typically lie between grid points provides a broad opportunity to significantly increase the power.
Rounding to the grid was previously considered (called ``e-value boosting'') in \cite{wang_false_discovery_2022}, but such methods required knowledge of the exact distribution of the e-value $X_i$ when the null hypothesis is true; that is often infeasible, particularly when we are testing composite null hypotheses. In what follows, we develop methods that use independent external randomness to \emph{stochastically round} e-values and increase the number of discoveries made by e-BH. This idea can, in turn, be used to directly improve the power of the seminal BY procedure, which is based on p-values and may at first glance appear to have nothing to do with rounding e-values. Further, we show that stochastic rounding induces power improvements whenever e-values are utilized.

\boldparagraph{A note on the use of randomization}
While the set of randomized multiple testing procedures contains all deterministic ones, it is far from clear whether the most powerful randomized procedure is \emph{strictly} more powerful (in at least some situations) than the most powerful deterministic one.

It could be the case that randomization confers no additional benefit in every situation. In order to find and exploit the gap between deterministic and randomized procedures, one must have a clear algorithmic idea: just throwing randomness at the problem will not suddenly confer more power. To the best of our knowledge, our paper is the first to recognize that it is even possible for the BY procedure to be improved via randomization.
As an important counterpoint, we do not know of any way for randomization to improve the power of other procedures like the BH procedure under positive dependence, and indeed, we conjecture that it does not. Both the BH and BY procedures have their FDR upper bounds being achieved with equality in a particular setting; yet it appears as if the power of BY can be essentially improved in almost all other situations, while we conjecture that (without further knowledge) the BH procedure cannot be strictly improved in any situation without worsening its power in other situations. Thus, it is not the case that randomization should indeed always or ``naturally'' help in multiple testing.

Most existing methods that use randomization in multiple testing are for single hypothesis tests, e.g., discrete p-values \citep{habiger_randomised_p-values_2011,dickhaus_how_analyze_2012,dohler_unified_class_2023}, Neyman-Pearson \citep{cvitanic_generalized_neyman-pearson_2001}, permutation p-values, conformal p-values \citep{vovk_machine-learning_applications_1999,vovk_algorithmic_learning_2005a}, etc.
However, our methods are the first to consider randomization that applies uniquely in the more general multiple testing setting. All of the aforementioned techniques are orthogonal to our approach --- after randomizing individual p-values or e-values, our randomized methods would still confer additional power.
Our randomized procedures strictly improve over the corresponding deterministic procedure (e.g., reject a superset of the deterministic procedure's discovery set, produce a smaller or equal p-value, etc.). Among randomization techniques in statistics (such as the bootstrap or sample splitting), this type of uniform improvement is rare.

We also note that extrinsic randomization is widely seen across many different statistical methods. In fact, it is often crucial for enabling certain guarantees and ensuring improved precision. In addition to the examples we have already described, multiple recent works in conformal prediction rely on randomization to derive either valid or more powerful covariate conditional guarantees \citep{hore_conformal_prediction_2024,gibbs_conformal_prediction_2024}. In the post-selection inference literature, a line of work developed by \citet{tian_selective_inference_2018}, \citet{rasines_splitting_strategies_2023}, \citet{leiner_data_fission_2024}, \citet{neufeld_data_thinning_2024}, \citet{neufeld_inference_latent_2024}, and \citet{dharamshi_generalized_data_2024} involves utilizing external randomization in the selection procedure or for the generation of synthetic data so that one can ensure valid selective inference. In addition, \citet{leiner_graph_fission_2024}, \citet{dharamshi_decomposing_gaussians_2024}, and \citet{liu_cross-validation_antithetic_2024} utilize similar external randomization techniques to generate synthetic data that are then used to derive a more efficient form of cross-validation. In addition, \citet{ignatiadis_empirical_bayes_2024} note that the method for empirical Bayes estimation from \citet{ignatiadis_empirical_bayes_2023a}, which originally assumed multiple samples are collected per unit, can be utilized when synthetic copies of data are generated using external randomization from a single sample for each unit.

When applying randomization, one natural concern that may arise is whether the results are reproducible. This concern can be addressed by constructing a superuniform random variable through sources of randomness derived directly from the data itself (e.g., the rank of the first sample in the data divided by 1 more than the total number of samples when the data is exchangeable), as noted in \citet[Section 10.6]{ramdas_randomized_exchangeable_2023}. Regardless of whether the above or other solutions are satisfactory, the earlier paragraph shows that the community has embraced external randomization in a huge variety of contexts, even just in recent years. It is certainly of interest to examine the curious phenomenon that a simple uniform random variable can exploit power that is left on the table by the original deterministic procedure.

\subsection{Summary of contributions}
Having explained our new ideas abstractly above, we now concisely highlight the key improvements in multiple testing methodology that we obtain through randomization. In the sequel, let $U$ be a uniform random variable over $[0, 1]$ independent of any e-values or p-values.
\boldparagraph{More powerful e-values via stochastic rounding} Our first contribution is formulating \emph{stochastic rounding}, a method for randomizing a e-value that will only increase its power over the original e-value for any hypothesis test. We generalize this to apply to situations where the threshold for rejection is not fixed, but data-adaptive.
Previous work on randomizing e-values \citep{ramdas_randomized_exchangeable_2023,ignatiadis_e-values_unnormalized_2022} ultimately produce a randomized p-value. However, in many applications, maintaining an e-value is critical to ensuring important robustness properties. Our primary example of interest is in the use of e-values for multiple testing with FDR control --- the e-BH procedure that uses e-values is robust to arbitrary dependence, while the BY procedure for p-values under arbitrary dependence pays a test level penalty that is proportional to $\log K$. Thus, \emph{we develop the first randomized e-value that is never less (and usually more) powerful than the original e-value}. Further, our randomized e-value is equally as powerful as randomized p-values developed in previous work \citep{ramdas_randomized_exchangeable_2023} for a single hypothesis test.
\boldparagraph{More powerful multiple testing and post-selection inference under arbitrary dependence} Each of the randomized procedures satisfies two main properties: (1) the procedure will never be worse (where ``worse'' is defined based on the problem, e.g., fewer discoveries in multiple testing) than the deterministic procedure which it is derived from and (2) under no conditions or some weak regularity conditions on the distribution, the randomized procedure is better with positive probability (e.g., more discoveries in multiple testing, while still having the same error rate guarantee). Thus, the power increase from randomization does not result in any kind of tradeoff or cost, and is a ``strict'' improvement in this sense.

For all of our multiple testing procedures, randomization is a way to increase discoveries in a way that is proportionate to the accumulated evidence --- larger e-values and smaller p-values that were not rejected by e-BH and BY, respectively, have larger probabilities of being rejected by the randomized procedures.

\begin{enumerate}
    \item \textit{Improving e-BH through stochastic rounding:} One of our core contributions is a new method, \UEBH, that has the following simple characterization:
    \begin{gather}
        \text{Apply BH to }(U / X_1, \dots, U / X_K). \tag*{(\UEBH)}
    \end{gather} \UEBH\ dominates the standard e-BH procedure:
    \begin{gather}
        \text{Apply BH to }(1 / X_1, \dots, 1 / X_K). \tag*{(e-BH)}
    \end{gather} Both procedures retain valid FDR control under arbitrary dependence, which we prove in \Cref{thm:jrebh}. Notably, \UEBH\ will \emph{only enlarge} the discovery set that is rejected by e-BH, i.e., rejections made by e-BH will always be a subset of the rejections produced by \UEBH. We prove its FDR control by developing a notion called ``stochastic rounding'' and we use this to show that many other randomized methods (generalizing \UEBH) have FDR control and will also only enlarge the e-BH discovery set.

    \item \textit{Randomized version of the BY procedure for p-values:} We can improve the BY procedure for multiple testing with FDR control under arbitrary dependence through randomization.
Let $(P_1, \dots, P_K)$ be $K$ arbitrarily dependent p-values and $P_{(i)}$ denote the $i$th smallest p-value. The BY procedure finds $k^*_{\BY}$ discoveries through a step-up procedure, and produces a discovery set $\Dcal_{\BY}$ that contains the hypotheses with the $k^*_{\BY}$ smallest p-values. These quantities are formulated as follows:
\begin{align}
    k^*_{\text{\BY}} \coloneqq \max \left\{i \in [K]: P_{(i)} \leq \frac{\alpha i}{K\ell_K}\right\},\
    \Dcal_{\text{BY}} \coloneqq \left\{i \in [K]: P_i \leq \frac{\alpha k^*_{\BY}}{K\ell_K}\right\}. \tag*{(\text{BY})}
\end{align}
When we have access to an random variable $U$ that is uniformly distributed on $[0, 1]$ and independent of $(P_1, \dots, P_K)$, we define a novel procedure, the \RBY\ procedure, which makes $k^*_{\text{\RBY}}$ discoveries, and produces a discovery set $\Dcal_{\text{\RBY}}$ contains the hypotheses with the $k^*_{\text{\RBY}}$ smallest p-values:
\begin{align}
    k^*_{\text{\RBY}} \coloneqq \max \left\{i \in [K]: P_{(i)} \leq \frac{\alpha\lfloor i / U\rfloor}{K\ell_K}\right\},\  \Dcal_{\text{\RBY}} &\coloneqq \left\{i \in [K]: P_i \leq \frac{\alpha k^*_{\text{\RBY}}}{K\ell_K}\right\}. \tag*{(\text{\RBY})}
\end{align}Since $U\leq 1$, $\lfloor k / U \rfloor \geq k$ for every $k \in [K]$, and the \RBY\ procedure never produces fewer discoveries than the BY procedure and will typically make more discoveries in most settings.
We prove FDR control of the \RBY\ procedure and provide sufficient conditions for strictly greater power in \Cref{thm:rounded-by}.
The FDR control is a result of applying p-to-e calibration and subsequently utilizing the FDR control of our randomized e-value procedures. We also prove a stronger version of the superuniformity lemma \citep{blanchard_two_simple_2008,ramdas_unified_treatment_2019} in \Cref{lemma:rand-superuniform} for arbitrarily dependent p-values. As a result, we can derive randomized versions of multiple testing procedure that uses reshaping functions with provable FDR control in \Cref{thm:reshaped-rand-by}.

    \item \textit{Randomized Hommel for global null testing:} Our randomization techniques are also useful for deriving randomized p-value merging (p-merging) procedures (i.e., for global null testing). We can randomize the Hommel procedure \citep{simes_improved_bonferroni_1986} to formulate a randomized U-Hommel p-value that is always upper bounded by (i.e., at least as powerful as) the typical Hommel p-value. In fact, we can use stochastic rounding to improve any admissible e-merging or homogeneous p-merging procedure, as characterized by \cite{wang_only_admissible_2025} and \cite{vovk_admissible_ways_2022}, respectively, in \Cref{sec:admissible}.
Our randomized U-Hommel procedure and its deterministic analog, the Hommel procedure, are formulated as follows:
\begin{align}
     P_{\UHommel} \coloneqq \min_{i \in [K]} \frac{P_{(i)}K \ell_K}{(\lfloor i / U\rfloor \wedge K)}, \qquad P_{\Hommel} \coloneqq \min_{i \in [K]} \frac{P_{(i)}K \ell_K}{i}.
\end{align} We show that the U-Hommel p-value also possesses Type I error control in \Cref{thm:uhommel} while never being larger and often smaller than the Hommel p-value, also develop an improvement in its corresponding closed testing procedure.
\item \textit{Rounding for post-selection inference with false coverage rate (FCR) control.} In addition to our results for multiple hypothesis testing, we extend the results of our rounding techniques to post-selection inference with false coverage rate (FCR) control, where FCR is the expected proportion of CIs that are selected (based on the data), which do not cover the true parameter. This includes the e-BY procedure \citep{xu_post-selection_inference_2022}, which solely uses e-CIs, as well as the classic post-selection procedure of \citet{benjamini_false_discovery_2005}, which we refer to as the CI-BY procedure. Our rounded versions of these two procedures hold under arbitrary dependence among the CIs and any (perhaps unknown) choice of selection rule. Here, the rounding results in narrower CIs as a result.
\end{enumerate}

\begin{figure}[h]
    \centering
    \includegraphics[width=0.9\textwidth]{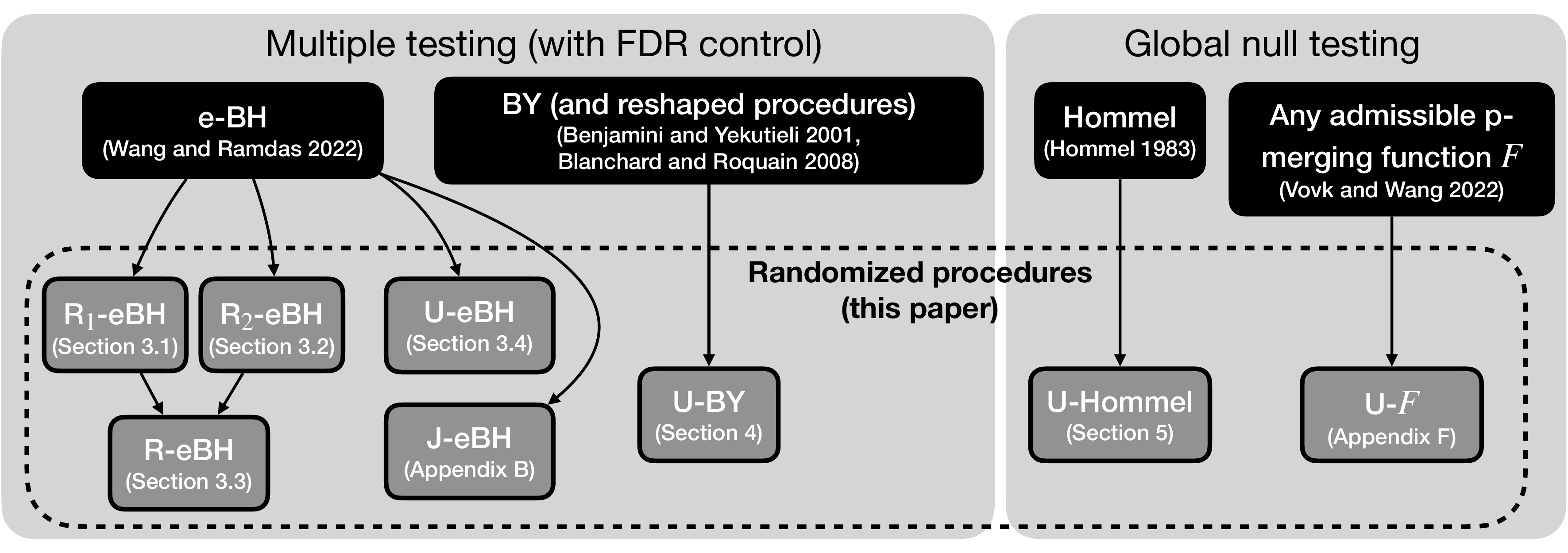}
    \caption{Summary of procedures we improve through randomization in this paper. All randomized procedures retain the validity of their deterministic counterparts, i.e., FDR control for multiple testing and Type I error control for global null testing.}
    \label{fig:contributions}
\end{figure}

\Cref{fig:contributions} summarizes our results for randomized multiple testing procedures. The remainder of the paper is organized as follows. In \Cref{sec:stoch-round}, we introduce the notion of stochastic rounding. In \Cref{sec:rand-ebh}, we use the notion of stochastic rounding and describe different ways in which it can be combined with the e-BH procedure to improve power. We formulate the \RBY\ procedure---the randomized version of the BY procedure---in \Cref{sec:rand-by}. \Cref{sec:rand-simes} contains our derivation of the U-Hommel procedure that randomizes the Hommel procedure.
We show that our methods have more power than their deterministic counterparts in numerical simulations in \Cref{sec:Simulations}, and if the sole goal of a practitioner was to maximize power, one should seriously consider using randomization to make more discoveries.

\subsection{Related work}

Prior work for boosting e-values to improve the power of e-BH in \citet{wang_false_discovery_2022} required precise knowledge of the distribution of e-values or their dependence structures (i.e., no longer assuming arbitrary dependence) to increase power. Similarly, \citet{fithian_conditional_calibration_2022} derive a method for improving the power of the BH procedure when the conditional distribution of all p-values given a sufficient statistic is known. Hence, both of these methods require additional information --- the weakest requirement \citep{wang_false_discovery_2022} is for the marginal distribution of each e-value, and the strongest \citep{fithian_conditional_calibration_2022} requires knowing the conditional distribution of all p-values. Subsequent work that builds on conditional calibration \citep{luo_improving_knockoffs_2024,lee_boosting_e-bh_2024a} inherits this limitation. In contrast, our improvements are \emph{assumption-free} --- we require no additional knowledge about the distribution of the e-values or p-values, nor about their dependence structure. In addition, our development of the stochastic rounding concept for e-values allows us to not only improve e-BH in every situation but is also method agnostic in the sense that it improves any procedure that utilizes e-values (e.g., the follow-up work of \citet{xu_online_multiple_2023} directly uses stochastic rounding to improve power in the online setting).

\section{Stochastic rounding of e-values}\label{sec:stoch-round}

Define $\reals_+$ as the nonnegative reals, and $\preals \coloneqq \reals_+ \cup \{+\infty\}$, and note the latter is closed set.
Consider any closed subset $\Gcal \subseteq \preals$ (where $\Gcal$ stands for ``grid'' since $\Gcal$ will often be countable or finite below).  Denote $g_* \coloneqq \inf\{x: x\in \Gcal\}$ and $g^* \coloneqq \sup\{x: x \in \Gcal\}$, and note that $g_*,g^* \in \Gcal$ since $\Gcal$ is closed. Further, for any $x \in [g_*, g^*]$, denote
\begin{align}
x^+ \coloneqq \inf\{y \in \Gcal: y \geq x\}, \qquad x_- \coloneqq \sup\{y \in \Gcal: y \leq x\},
\end{align} and note that $x^+, x_- \in \Gcal$ with $x^+ \geq x$ and $x_- \leq x$, and if $x \notin \Gcal$, then $x_- < x < x^+$.

Now define the stochastic rounding of $x \in \preals$ onto $\Gcal$, denote $S_\Gcal(x)$, as follows. If $x < g_*$, $x > g^*$, $x \in \Gcal$, or $x^+ = \infty$ then define $S_\Gcal(x)=x$.
Otherwise, define
\begin{align}
    S_\Gcal(x) =
\begin{cases}
 x_- \text{ with probability } \frac{x^+ - x}{x^+ - x_-}, \\
 x^+ \text{ with probability } \frac{x- x_-}{x^+ - x_-}.
\end{cases}
\label{eq:StochasticRounding}
\end{align}
Note that $S_\Gcal(x)$ need not lie in $\Gcal$, because if $x$ lies outside the range of $\Gcal$ then it is left unchanged. Also note that when $\Gcal = \preals$, we have $S_\Gcal(x)=x$ for all $x\in\preals$. \added{More generally, for any grid $\Gcal$, if $x \in \Gcal$, then $S_\Gcal(x) = x$.}

\begin{proposition}\label{prop:stoch-round-e}
If $X$ is an e-value, then $S_\Gcal(X)$ is also an e-value for any $\Gcal$. Further, $\expect[X] = \expect[S_\Gcal(X)]$.
\end{proposition}
The proof is simple: by design, $\expect[S_\Gcal(x)] = x$ for any real $x$, and thus when applied to any random variable, it leaves the expectation unchanged.
The important property for us will be that $S_\Gcal(X)$ can be larger than $X$ with (usually) positive probability, since it can get rounded up to $X^+$, and is at least $X_-$, even when rounded down.

\subsection{More powerful hypothesis testing via stochastic rounding}
In typical hypothesis testing with e-values, one ensures type I error control at level $    \alpha \in [0, 1]$ by checking if an e-value $E$ is greater than a fixed test level threshold $\alpha^{-1}$. This follows immediately from Markov's inequality and the definition of an e-value. Stochastic rounding allows us to define an e-value $S_{\alpha}(X)$ (which is how we denote $S_{\Gcal}(X)$ with $\Gcal =\{0, \alpha^{-1}, \infty\}$) that uniformly improves the power (i.e., rejection probability), and can quantify exactly the power improvement gained by $S_\alpha(X)$.
\begin{proposition}\label{prop:stoch-round-power}
    $X \geq \alpha^{-1}$ implies that $S_\alpha(X) \geq \alpha^{-1}$. Further, under any distribution (i.e., either null or alternative), the following is true.
\begin{align}
    \prob{S_\alpha(X) \geq \alpha^{-1}}  - \prob{X \geq \alpha^{-1}} = \alpha \cdot \expect[X \cdot \ind{X < \alpha^{-1}}].
\end{align}
\end{proposition}
\begin{proof}
    The implication simply follows from $S_\alpha(X) = X$ for all $X \geq \alpha^{-1}$.

    Let $B_X \mid X \sim \text{Bern}(\alpha X \wedge 1)$. Then, we can rewrite the difference in power of $S_\alpha(X)$ and $X$ as the following:
    \begin{align}
        \prob{S_\alpha(X) \geq \alpha^{-1}} - \prob{X \geq \alpha^{-1}} &= \prob{X < \alpha^{-1} \text{ and } B_X = 1}
                                                                        =\expect[B_X \cdot \ind{X < \alpha^{-1}}]\\
                                                                        &=\expect[\expect[B_X \mid X] \cdot \ind{X < \alpha^{-1}}]
                                                                        =\alpha \cdot \expect[X \cdot \ind{X < \alpha^{-1}}].
    \end{align}
    The third equality is by iterated expectation, and the last is by definition of $B_X$.
\end{proof}
One can view this as the e-value analog of the p-value derived uniformly from the randomized Markov's inequality in \citet{ramdas_randomized_exchangeable_2023}. However, we are able to achieve our power improvements by directly deriving a more powerful e-value --- this allows us to directly plug in these stochastically rounded e-values into any procedures that utilize e-values. This is an advantage in the multiple testing and post-selection inference setting since p-value procedures often have stringent requirements about the dependence between test statistics.

\subsection{Stochastic rounding to data-dependent thresholds}
In the stochastic rounding procedure we described previously, $\Gcal$ is a fixed set of values that is determined without any reference to the data. In the vast majority of multiple testing and selective inference procedures, however, the test level at which we would reject a hypothesis is data-dependent, e.g., the e-BH procedure's rejection level is dependent on all e-values, including the e-value that is being tested. Thus, an \emph{adaptive stochastic rounding} procedure stochastically rounds any input e-value $X$ to meet a data-dependent threshold $\widehat{\alpha}^{-1}$ for some $\widehat\alpha \in [0, 1]$.

Let $S_{\widehat{\alpha}}$ be shorthand for $S_{\Gcal(\widehat\alpha)}$ where $\Gcal(\widehat\alpha) = \{0, \widehat{\alpha}^{-1},\infty\}$. Note that $\widehat{\alpha}$ (and the $\Gcal(\widehat\alpha)$ we just defined) are random variables that can be arbitrarily dependent with an e-value $X$.

\begin{proposition}
\label{prop:dynamic-round-e}
     For any e-value $X$, and random variable $\widehat{\alpha} \in [0, 1]$ that possibly depends on $X$, $S_{\widehat{\alpha}}(X)$ is also an e-value and satisfies $\expect[S_{\widehat{\alpha}}(X)] = \expect[X]$.
\end{proposition}

To show \Cref{prop:dynamic-round-e} is true, notice that another way of writing $S_{\widehat{\alpha}}$ is as follows:
\begin{equation}\label{eq:rand-e}
    S_{\widehat{\alpha}}(X) = (X \cdot \ind{X \geq \widehat{\alpha}^{-1}}) \vee (\widehat\alpha^{-1} \cdot \ind{X \geq U\widehat\alpha^{-1}}),
\end{equation}
where $U$ is a uniform random variable over $[0, 1]$ that is independent of $X$ and $\widehat\alpha$, and we treat $0 / 0 = 0$. Indeed, we can check that $\expect[S_{\widehat{\alpha}}(X)] = \expect[X]$ by first taking expectation with respect to $U$ (while conditioning on $X$) and then with respect to $X$. 

\begin{proof}
We rewrite $S_{\widehat\alpha}(X)$ into the following equivalent form:
\begin{align}
    S_{\widehat\alpha}(X) = X \cdot \ind{X \geq \widehat{\alpha}^{-1}} + \underset{A}{\underbrace{\widehat{\alpha}^{-1} \cdot \ind{\widehat{\alpha}^{-1}
    > X \geq U\widehat{\alpha}^{-1}}}}, \label{eq:RandomERewrite}
\ifarxiv{}{\vspace{-5pt}}
\end{align}
where $U \sim \text{Uniform}[0, 1]$ is independent of $X$ and $\widehat{\alpha}$.
First, we derive the following equality for the expectation of $A$ as indicated in \eqref{eq:RandomERewrite}:
\begin{align}
    \expect[A]
    &=\expect[\expect[\widehat{\alpha}^{-1} \cdot \ind{\widehat{\alpha}^{-1} > X \geq U\widehat{\alpha}^{-1}} \mid X, \widehat{\alpha}]]
    = \expect[\widehat{\alpha}^{-1} \cdot \ind{X < \widehat{\alpha}^{-1}} \prob{U \leq \widehat{\alpha} X \mid X, \widehat{\alpha}}]\\
    &\labelrel{=}{rel:UniformInd} \expect[\widehat{\alpha}^{-1} \cdot \ind{X < \widehat{\alpha}^{-1}} \cdot (\widehat{\alpha} X \wedge 1)] = \expect[X \cdot \ind{X < \widehat{\alpha}^{-1}}]
    \label{eq:AUpperBound},
\end{align} where \eqref{rel:UniformInd} is by uniform distribution and independence from $X$ and $\widehat{\alpha}$ of $U$, and the last equality is by $\ind{X < \widehat{\alpha}^{-1}} = \ind{\widehat{\alpha} X < 1}$.

Under the null, we can simply upper bound the expectation of $S_{\widehat{\alpha}}(X)$ as follows.
\begin{align}
    \expect[S_{\widehat\alpha}(X)] = \expect[X\cdot\ind{X \geq \widehat{\alpha}^{-1}}] + \expect[X \cdot \ind{X < \widehat{\alpha}^{-1}}] = \expect[X] \leq 1, \label{eq:add-parts}
\end{align} where the first equality is by application of \eqref{eq:RandomERewrite} and \eqref{eq:AUpperBound}. Hence, we have shown our desired result. Note that if $U$ were to be a superuniform (stochastically larger than uniform) random variable, the last equality in \eqref{eq:add-parts} would be an inequality.
\end{proof}
We can derive a power statement for this e-value as well.
\begin{proposition}\label{prop:adaptive-round-power}
    $X \geq \widehat\alpha^{-1}$ implies that $S_{\widehat\alpha}(X) \geq \widehat\alpha^{-1}$. Under any distribution (i.e., either null or alternative), the following is true:
\begin{align}
    \prob{S_{\widehat\alpha}(X) \geq \widehat\alpha^{-1}}  - \prob{X \geq \widehat\alpha^{-1}} = \expect[\widehat\alpha X \cdot \ind{X < \widehat\alpha^{-1}}].
\end{align}
\end{proposition}
The proof follows from a similar argument to the one for \Cref{prop:stoch-round-power}.
\added{
This result quantifies precisely how much power is gained by stochastic rounding even when the test threshold is data-dependent. The power improvement is the expected product of the test level and the e-value when the e-value is below the test threshold. Notably, no power improvement will result from stochastic rounding when $X \geq \hat\alpha^{-1}$. 
}

Note that the boosting method considered in \citet{wang_false_discovery_2022} requires precise knowledge of the distribution of an e-value and the dependence structure between e-values to improve its power. In contrast, stochastic rounding requires no knowledge of the underlying distribution and hence can be applied in every situation where e-values are used.
\begin{remark}[Rounding trades off e-power for power]
    Stochastic rounding trades off between the notion of ``e-power'' that has been used to characterize the power of an e-value, and the classic notion of power in statistical inference, i.e., the probability of rejection. Recall that e-power of an e-value $X$ is $\expect[\log X]$, \added{where the expectation is taken under the alternative}. For any e-value, the following is true from Jensen's inequality \added{(for any underlying distribution)} and the fact that $\expect[S_{\hat\alpha}(X)] \leq \expect[X]$.
    \begin{align}
        \expect[\log S_{\hat\alpha}(X)] = \expect[\expect[\log S_{\hat\alpha}(X) \mid X]]  \leq \expect[\log \expect[S_{\hat\alpha}(X) \mid X]] \leq \expect[\log X]
    \end{align} Thus, \added{under the alternative,} the e-power of a stochastically rounded e-value is always at most that of the original e-value, while its power (at level $\alpha$) is at least that of the original. In short, stochastic rounding trades off e-power for power.
\end{remark}
\added{
    One can choose to perform stochastic rounding in a post hoc fashion. That is, one can choose to view all the data (i.e., the original e-value $X$ and the threshold $\hat\alpha$) before deciding whether to apply stochastic rounding or not. Formally, define a binary random variable $T \in \{0, 1\}$ that is a function of $X$ and $\hat\alpha$, and define the post hoc stochastic rounding procedure as:
    \begin{align}
        \widetilde{S}_{T, \hat\alpha}(X) = T \cdot S_{\hat\alpha}(X) + (1 - T) \cdot X.
    \end{align}
    The same arguments for validity as seen in \Cref{prop:dynamic-round-e} still apply since $T$ is measurable w.r.t.\ $X$ and $\hat\alpha$. Of course, one cannot decide whether to round or not after peeking at the external randomness in $S_{\hat\alpha}(X)$. 
}
\added{
\begin{remark}[Rounding helps decision making, but not evidence accumulation]
    In a sequential setting where we have a e-process (a nonnegative process that is an e-value at any stopping time), one could choose to round after deciding to stop, and this would increase power if we stopped before reaching a rejection threshold.  
    However, when e-values are being aggregated by averaging or multiplication (say within the construction of the e-process), it may not be beneficial to stochastically round; for example, if the grid contains zero, the product could become zero.
    Nevertheless, if the aggregated e-value must be converted into a rejection decision, we can still use the stochastically rounded e-value to gain power. Thus, one would typically only perform the rounding when data collection is complete in order to improve the power of the subsequent decision-making procedure, but would not apply it to e-values that are used or produced in the midst of an ongoing analysis. In short, rounding is beneficial for decision-making, but not for aggregating evidence.  
\end{remark}
}
\section{Randomization improves the e-BH procedure}
\label{sec:rand-ebh}
Below, we show how to (usually strictly) improve the power of the e-BH procedure through stochastic rounding of e-values.
This notion lies at the core of our randomized procedures.
In e-BH (among other testing procedures), this means that the rejection behavior would not change even if $X$ was rounded down. Let
\begin{gather}
    \alpha_i \coloneqq \alpha i / K, \qquad \Kcal \coloneqq \{\alpha_i^{-1}: i \in [K]\} \cup \{0, \infty\}
\end{gather}
denote the set of possible levels that e-BH may reject e-values at and in addition to $0$ and $\infty$. If $X \geq K / (\alpha k)$ for some $i \in [K]$, then $S_\Kcal(X) \geq K / (\alpha k)$ as well. Thus, a stochastically rounded e-value can only improve power when used in conjunction with e-BH.

\subsection{Warm-up: the stochastically rounded \RoneEBH\  procedure}

The \emph{\RoneEBH\  procedure} simply applies the e-BH procedure to the set of e-values $\{S_{\Kcal}(X_{k})\}_{i \in [K]}$.
Let $\Dcal_1$  be the set of rejections made by \RoneEBH. Recall that $k^*$ is the number of discoveries made by e-BH as defined in \eqref{eq:ebh-kstar}.
We can prove the following guarantee for \RoneEBH.
\begin{theorem}
\label{thm:grid-ebh}
For any arbitrarily dependent e-values $(X_1, \dots, X_K)$, the \RoneEBH\ procedure ensures $\FDR \leq \alpha$ and $\Dcal_1\supseteq \Dcal_{\rmEBH}$. Further, $\prob{\Dcal_1 \supset \Dcal_{\rmEBH}} > 0$, i.e., the probability that \RoneEBH\ makes extra discoveries over e-BH is positive, if and only if
\begin{align}
    \prob{\exists k \in [K - k^*]: X_{[k^* + j]} > K/(\alpha(k^* + k + 1)) \text{ for all }j \in [k]} > 0. \label{eq:kcal-cond}
\end{align}
\end{theorem}
The proof follows from the fact that if $X_i$ ever takes on a value that is between levels in $\Kcal$, the e-BH procedure will reject at the same level (and make the same rejections) as if $(X_i)_{-}$ were substituted in its place. Stochastic rounding guarantees that $S_{\Kcal}(X_i) \geq (X_i)_{-}$ almost surely, so it can only increase the number of rejections. Further, when $X_i$ is between two levels in $\Kcal$, then $S_{\Kcal}(X_i) =  (X_i)_{+} > X_i$ with positive probability, which leads to rejecting hypotheses that e-BH did not reject. This intuition leads us to the condition in \eqref{eq:kcal-cond} for which \RoneEBH\ has strictly more power than e-BH. We defer the full proof to \Cref{sec:ebh-improve}.

It is worth remarking that \eqref{eq:kcal-cond} is an extremely weak condition that would be very frequently satisfied. For example, if the e-values are independent and continuously distributed over $[0, K\alpha^{-1}]$ (or a larger interval), then \eqref{eq:kcal-cond} will hold. As an explicit example, if the data $Z_i$ for testing the $i$-th hypothesis are Gaussian with variance $\sigma^2$, and we are testing whether the mean of $Z_i$ is nonpositive (or equal to zero) against the alternative that the mean is positive, \citet{ramdas_admissible_anytime-valid_2020} show that all admissible e-values will be mixtures of likelihood ratios between a positive mean Gaussian and a zero mean Gaussian: these likelihood ratio e-values take the form $\exp(\lambda Z_i - \lambda^2\sigma^2/2)$ for $\lambda > 0$, which are clearly continuous and unbounded, as are many other e-values for testing parametric and nonparametric hypotheses. The independence mentioned above is far from necessary, but it is sufficient to ensure that the probability in \eqref{eq:kcal-cond} is not pathologically equal to zero due to some awkward worst-case dependence structure.
\begin{remark}
    In \Cref{sec:general-stoch-round}, we propose a generalized version of rounding, where one stochastically rounds an input, $x$, by sampling from a mixture distribution over all values in the grid $\Gcal$ that are larger than $x$, instead of only sampling from the point mass at $x^+$.
\end{remark}

\subsection{Data-dependent rounding for the \RtwoEBH\ procedure}
We define the \emph{\RtwoEBH\ procedure} as the procedure that rejects the $i$th hypothesis when $S_{\alphath^*}(X_i) \geq 1 / \alphath^*$ (recall that $1 / \alphath^*$ was the threshold of rejection for e-BH). Let $\Dcal_2$ be the resulting discovery set and $k^*_2 \coloneqq |\Dcal_2|$.
\begin{theorem}
     \label{thm:rand-ebh}
    For any arbitrarily dependent e-values $(X_1, \dots, X_K)$, the \RtwoEBH\ procedure ensures $\FDR \leq \alpha$ and $\Dcal_2 \supseteq \Dcal_{\rmEBH}$. Further, $\prob{\Dcal_2 \supset \Dcal_{\rmEBH}
    }> 0$ if and only if \begin{align}
    \prob{\exists i \in [K]: X_i \in (0, (\widehat\alpha^*)^{-1})} > 0. \label{eq:improve-cond}
\end{align}
\end{theorem}
\eqref{eq:improve-cond} requires some probability of e-value realizations where (1) e-BH does not reject all hypotheses and (2) the e-values of the unrejected hypotheses are not all 0. In other words, the only cases where we \emph{do not} get a strictly larger power are when either e-BH rejects all hypotheses or when all unrejected e-values equal zero. Evidently, in practice, we should essentially always see a strict power increase.

\begin{proof}
    FDR control follows from the fact that $S_{\widehat\alpha^*}(X_i)$ is an e-value for each $i \in [K]$, and that e-BH controls FDR under arbitrary dependence as shown in Theorem 2 of \citet{wang_false_discovery_2022}.

Now, we will show the result about a strict improvement. Note that $\Dcal_2 \supseteq \Dcal_{\rmEBH}$ almost surely since $S_{\alphath^*}(X_i) = X_i$ for each $i \in \Dcal_{\rmEBH}$. Define the event
\begin{align}
\mathcal{E}_\varepsilon \coloneqq \{\exists i \in [K]: X_i \in (\varepsilon, (\widehat\alpha^*)^{-1})\}.
\end{align} $\mathcal{E}_0$ occurring is necessary for $\Dcal_2 \supset \Dcal_{\rmEBH}$ to be true, so $\eqref{eq:improve-cond}$ is necessary for $\prob{\Dcal_2 \supset \Dcal_{\rmEBH}} > 0$. Note that $\prob{\mathcal{E}_\varepsilon} > 0$ for some fixed $\varepsilon > 0$ if \eqref{eq:improve-cond} holds. Since $\prob{\Dcal_2 \supset \Dcal_{\rmEBH} \mid \mathcal{E}_\varepsilon} > p$ for some fixed $p > 0$, we have $\prob{\Dcal_2 \supset \Dcal_{\rmEBH}} > 0$, thus proving our desired result.
\end{proof}

We can also quantify the power improvement over base e-BH.
\begin{proposition}\label{prop:r2ebh-power}
    The expected number of discoveries \RtwoEBH\ makes over e-BH is
    \begin{align}
    \expect[k_2^*] - \expect[k_{\rmEBH}^*] = \sum_{i = 1}^K \expect[\widehat\alpha^* X_i \cdot \ind{X_i < (\widehat\alpha^*)^{-1}}].
\end{align}
\end{proposition}
\begin{proof}
    We can make the following derivation of the difference in expected discoveries.
    \begin{align}
        \expect[k_2^* - k_{\rmEBH}^*] &=
        \sum\limits_{i = 1}^K \expect[\ind{S_{\widehat\alpha^*}(X_i) \geq (\widehat\alpha^*)^{-1}} - \ind{X \geq (\widehat\alpha^*)^{-1}}]\\
         &=\sum\limits_{i = 1}^K \prob{S_{\widehat\alpha^*}(X_i) \geq (\widehat\alpha^*)^{-1}} - \prob{X \geq (\widehat\alpha^*)^{-1}}]
         =\sum\limits_{i = 1}^K \expect[\widehat\alpha^* X_i \cdot \ind{X_i < (\widehat\alpha^*)^{-1}}.
    \end{align} The last equality is by \Cref{prop:adaptive-round-power}. Thus, we have shown our desired result.
\end{proof}
\begin{remark} For any e-value $X$, $S_{\widehat{\alpha}}(X)$ allows us to tighten some of the looseness in the $\ind{X \geq \widehat{\alpha}} \leq \widehat \alpha X$ inequality, which is the core inequality in the proof of FDR control for e-BH when $\ind{X \geq \widehat{\alpha}} = 0$. Normally, $\ind{X \geq \widehat{\alpha}} = 0$ does not imply $\alpha X = 0$. With adaptively rounded e-values, however, $\ind{S_{\widehat{\alpha}}(X) \geq \widehat{\alpha}}  = 0$ if and only if $\alpha S_{\widehat{\alpha}}(X) = 0$. Hence, \RtwoEBH\ significantly tightens the gap between the theoretical bound and the true FDR.
\end{remark}

\subsection{Combining fixed and data-dependent rounding: \RbothEBH}We can combine fixed and data-dependent rounding of \RoneEBH\ and \RtwoEBH\ to derive the \RbothEBH\ procedure in \Cref{alg:rboth} and the following guarantees about it.
\begin{algorithm}[h]
    \caption{\RbothEBH\ combines fixed and data-dependent rounding of input e-values $(X_1, \dots, X_K)$.\label{alg:rboth}}
    \SetAlgoLined
    Apply the e-BH procedure to $(S_\Kcal(X_1), \dots, S_\Kcal(X_K))$.\\
    Let $k^*_1$ denote the number of rejections.\\
    Reject the $i$th hypothesis if $S_{\widehat\alpha^*_1}(S_\Kcal(X_i)) \geq \alpha_{k^*_1 + 1}^{-1}$. In other words, reject if $S_\Kcal(X_i) \geq U_i\alpha_{k_1^* + 1}^{-1}$ where $U_i \sim \text{Uniform}[0, 1]$ and independent of $(S_{\Kcal}(X_1), \dots, S_{\Kcal}(X_K))$.\\
    Let $\Dcal_{\R}$ denote the discovery set and let $k^*_{\R} \coloneqq |\Dcal_{\R}|$.
\end{algorithm}
\begin{theorem}
    For any arbitrarily dependent e-values, $(X_1, \dots, X_K)$, the \RbothEBH\ procedure controls $\FDR \leq \alpha$. Further, \RbothEBH\ makes no fewer discoveries than either of \RoneEBH\  and e-BH, i.e., $\Dcal_{\R} \supseteq \Dcal_1 \supseteq \Dcal_{\rmEBH}$. Further, $\prob{\Dcal_{\R} \supset \Dcal_{\rmEBH}} > 0$ if and only if \eqref{eq:improve-cond} holds.
    \label{thm:rboth-ebh}
\end{theorem}
This is simply a result of combining \Cref{thm:grid-ebh} and \Cref{thm:rand-ebh}. \added{As we will later observe in simulations in \Cref{sec:Simulations}, \RbothEBH\ is typically more powerful than \RtwoEBH, but neither strictly dominates the other in all cases.}

For general problems, we should expect $k^*_{\R} > k_1^* > k^*$, although it is possible to derive ``corner cases'' where randomization does not help.
\subsection{The jointly rounded \UEBH}
So far, we have treated every stochastic rounding operation as being independent of every other. However, the validity of stochastic rounding does not need independence in the algorithmic randomness across e-values.

First, define the \emph{\UEBH\ procedure} as applying the e-BH procedure to $(X_1 / U, \dots, X_K / U)$ for a single uniform random variable $U$. $X / U$ is generally not an e-value. In fact, $\expect[X / U] = \infty$ unless $X = 0$ almost surely, since $\expect[U^{-1}] = \infty$.  However, we can still derive an e-value procedure that has equivalent power.
Let $\Dcal_\U$ be the resulting discovery set and let $k^*_{\U}$ be its cardinality.

\boldparagraph{Joint stochastic rounding} To show that \UEBH\ has valid FDR control, we first define a joint stochastic rounding procedure. Let $(S_{\Gcal_1}(X_1), \dots,  S_{\Gcal_K}(X_K))$ denote the rounded e-values of a joint stochastic rounding procedure, with $K$ different grids $\Gcal_1, \dots, \Gcal_K$. Let $U$ be a uniform random variable on $[0, 1]$ that is independent of $(X_1, \dots, X_K)$. We can define the jointly rounded e-values as follows, where we replace \eqref{eq:StochasticRounding} with the following construction:
\begin{align}
    S_{\Gcal_i}(x) \coloneqq \begin{cases}
        x_- \text{ if }U > \frac{x - x_-}{x^+ - x_-}\\
        x^+ \text{ if }U \leq \frac{x - x_-}{x^+ - x_-}.
    \end{cases}
    \label{eq:joint-rounding}
\end{align} Here, $x_-, x^+$ are defined w.r.t.\ to the grid $\Gcal_i$. Obviously, as before, jointly rounded e-values are still (dependent) e-values.
\begin{proposition}
    If $(X_1, \dots, X_K)$ are e-values, then $(S_{\Gcal_1}(X_1), \dots, S_{\Gcal_K}(X_K))$ are also e-values. Further, $\expect[S_{\Gcal_i}(X_i)] = \expect[X_i]$ for each $i \in [K]$.
    \label{prop:joint-round-evalue}
\end{proposition}
Notice that the for each $i \in [K]$, $S_{\Gcal_i}(X_i)$ is exactly the stochastic rounding of $X_i$ onto $\Gcal_i$.
Therefore, \Cref{prop:joint-round-evalue} follows immediately from its definition in \eqref{eq:joint-rounding}.

\begin{theorem}
    \label{thm:jrebh}
    For any arbitrarily dependent e-values, $(X_1, \dots, X_K)$, the \UEBH\ procedure ensures that $\FDR \leq \alpha$ and $ \Dcal_{\U}\supseteq\Dcal_{\rmEBH} $. Further, $\prob{\Dcal_{\U} \supset \Dcal_\rmEBH} > 0$ iff \eqref{eq:improve-cond} holds.
\end{theorem}
To prepare for defining the grids we round to, we introduce the following quantities:
\begin{align}
    \rank[i] \coloneqq \sum_{j \in [K]} \ind{X_j \geq X_i}, \qquad \ell(i) \coloneqq \underset{j \geq \rank[i]}{\argmax}\ j X_{[j]},
\end{align} for each $i \in [K]$. If there are multiple indices that satisfy the argmax, we take the largest index.

\begin{lemma}
    The \UEBH\ procedure is equivalent to the e-BH procedure applied to the jointly rounded e-values $(S_{\alpha_{\ell(1)}}(X_1), \dots,S_{\alpha_{\ell(K)}}(X_K))$.
    \label{prop:uebh-as-bh}
\end{lemma}
\begin{proof}
Let $k^*_{\rmEBH}$ represent the number of rejections made by e-BH. Both the e-BH and the \UEBH\ formulation will reject the hypotheses corresponding to the $k^*_{\rmEBH}$ and $k^*_\U$ largest e-values. The number of rejections made by e-BH is as follows:
\begin{align}
    k^*_{\rmEBH} = \max \left\{i \in [K]: \frac{X_{[i]}}{U} \leq \frac{ K}{\alpha i}\right\} = \max \left\{i \in [K]: U \leq \frac{\alpha iX_{[i]}}{K}\right\}
\end{align}
Now, we simply want to show that
\begin{align}
    k^*_\U = \ell^* \coloneqq \max\left\{\ell(i): U \leq \alpha X_i \ell(i) / K, i \in [K]\right\} = k^*_{\rmEBH}.
    \label{eq:u-bh-equal}
\end{align}
We will first argue that the left equality is true. Stochastic rounding rejects the $i$th hypothesis if and only if $U \leq \alpha X_i \ell(i) / K$ by expanding the definition of $\alpha_{\ell(i)}$ and applying \Cref{prop:adaptive-round-power}. $U \leq \alpha X_i \ell(i) / K$ if and only if $\ell(i) \leq \ell^*$, and exactly $\ell^*$ hypotheses satisfy that.

The right equality is true because if $k^*_{\rmEBH} > \ell^*$, then the definition of $\ell$ is violated, as $k^*_{\rmEBH}X_{[k^*_\rmEBH]} > \ell(i) X_{[\ell(i)]}$ for some $i$ where $\ell(i) \leq \ell^*$. If $\ell^* > k^*_\rmEBH$, this violates the definition $k^*_\rmEBH$, as $\ell^* > k^*_\rmEBH$ and $U \leq \alpha k X_{[\ell^*]} / K$.
Since \eqref{eq:u-bh-equal} is true, our desired result follows.
\end{proof}
When viewed through the lens of being an application of the BH procedure, it is not apparent that the \UEBH\ procedure can control FDR under arbitrary dependence. However, FDR control does follow directly from \UEBH\ being a stochastic rounding procedure.
Our proof of strict power improvement in \Cref{thm:jrebh} utilizes the following result.
\begin{proposition}
\label{prop:uebh-better-than-r2ebh}
   \UEBH\ is at least as powerful as \RtwoEBH, i.e.,  $\expect[k_{\U}^*] \geq \expect[k_2^*]$.
\end{proposition}
We defer the proof of \Cref{sec:uebh-better-proof}. Now that we have set up the appropriate notation, we complete the proof of Theorem 4.
\begin{proof}[Proof of \Cref{thm:jrebh}]
The FDR control is simply because we are running e-BH on e-values, and FDR control is achieved by Theorem 2 of \cite{wang_false_discovery_2022},

We know that \UEBH\ rejects all discoveries made by e-BH by noting that \UEBH\ is equivalent to applying e-BH on $(X_1 / U, \dots, X_K / U)$ by \Cref{prop:uebh-as-bh}. Since $X_i / U \geq X_i$ almost surely for each $i \in [K]$, we know that \UEBH\ rejects all $i \in \Dcal$.

Now, we note that \eqref{eq:improve-cond} implies that $\prob{\Dcal_2 \supset \Dcal_{\rmEBH}} > 0$, together with \Cref{prop:uebh-better-than-r2ebh}, justifies that $\prob{\Dcal_{\U} \supset \Dcal_{\rmEBH}}> 0$. Now, because of the alternative formulation of \UEBH\ in \Cref{prop:uebh-as-bh}, we know that if \eqref{eq:improve-cond} does not hold, \UEBH\ simply outputs $\Dcal_{\rmEBH}$ almost surely. This is because, almost surely, $X_i / U = 0$ if $i \not\in\Dcal_{\rmEBH}$ and $X_i / U \geq X_i$ if $i \in \Dcal_{\rmEBH}$. Thus, we have shown our desired result.
\end{proof}

We now make a few remarks about the randomized e-BH procedures we have introduced.
\begin{remark}
    While the theorems we have shown so far rely on $(X_1, \dots, X_K)$ being e-values, we actually only require them to be compound e-values \citep{ignatiadis_asymptotic_compound_2025}, i.e., they satisfy
    $\sum_{i \in \Ncal} \expect[X_i] \leq K.$ This is sufficient for FDR control, as \cite{wang_false_discovery_2022} noted for standard e-BH.
\end{remark}

\begin{remark}\label{remark:derandomization}
    To derandomize our randomized e-BH procedures, one could take the stochastically rounded e-values and average them over multiple runs (i.e., draws of external randomness). Let $S_{\Gcal_i}^{(n)}(X_i)$ denote the $i$th stochastically rounded e-value for the $n$th run. By construction, all of the stochastically rounded e-values that we used in the randomized procedures satisfy $\expect[S_{\Gcal_i}^{(n)}(X_i) \mid X_i] = X_i$. Hence, by law of large numbers, $\frac{1}{N}\sum_{n = 1}^N S_{\Gcal_i}^{(n)}(X_i) \rightarrow X_i$ as the number of runs, $N$, approaches infinity. Thus, a derandomized version of any of our randomized e-BH procedures recovers exactly e-BH on the original e-values $(X_1, \dots, X_K)$.

	    \added{
	        \citet{ignatiadis_asymptotic_compound_2025} suggest an alternative technique for derandomization. For any level $\alpha$ FDR controlling discovery set $\rejset$, we have that $X_i = \ind{i \in R} \cdot (\alpha_{|R|})^{-1}$ for each $i \in [K]$ are compound e-values. Thus, one could construct compound e-values from the discovery set of many runs of any randomized e-BH procedure and then average the e-values for each $i \in [K]$ across runs. We empirically compare the two methods of derandomization in \Cref{sec:derandomization-simulations} and conclude that our method of averaging stochastically rounded e-values is more powerful and stable than averaging the compound e-values of \citet{ignatiadis_asymptotic_compound_2025}.
	    }
\end{remark}
\begin{remark}
We note that while the \UEBH\ procedure is randomized, it always respects the `moral' property of all step-up procedures in that a hypothesis is rejected only if all hypotheses with more evidence (i.e., larger e-values) are also rejected. Formally, this means that if $j \in \Dcal_{\U}$ then $\left\{i \in [K]: X_i \geq X_j \right\} \subseteq \Dcal_{\U}$.
\end{remark}
\ifarxiv{}{\vspace{-20pt}}

\section{Randomization for multiple testing with p-values}
\label{sec:rand-by}

Our randomization techniques are also applicable to the Benjamini-Yekutieli procedure \citep{benjamini_control_false_2001}, which provides FDR control for arbitrarily dependent p-values.
Recall that $P$ is a \emph{p-value} if it is \emph{superuniform}, i.e., $\prob{P \leq s} \leq s$ for all $s \in [0, 1]$, when the null hypothesis is true. Let $(P_1, \dots, P_K)$ be $K$ arbitrarily dependent p-values and $U$ be a superuniform random variable that is independent of the p-values.
The BY procedure makes the following number of discoveries:
\begin{align}
    k^*_{\rmBY} \coloneqq \max \left\{i \in [K]: P_{(i)} \leq \frac{\alpha i}{K \ell_K}\right\},
\end{align} where $\ell_K \coloneqq \sum_{i = 1}^K 1 / i$ is the $K$th harmonic number. The BY procedure rejects the $k^*_\rmBY$ smallest p-values, which is equivalent to rejecting p-values where $P_i \leq \alpha k^*_\rmBY / (K\ell_K)$. Let $\Dcal_{\rmBY}$ be the resulting discovery set. Now, define a randomized version of $k^*_{\rmBY}$:
\begin{align}
k^*_{\text{\RBY}} \coloneqq \max \left\{i \in [K]: P_{(i)} \leq \frac{\alpha(\lfloor i / U\rfloor \wedge K)}{K\ell_K}\right\},
\end{align}
The \emph{\RBY\ procedure} rejects the $i$th hypothesis if
${P_i \leq \alpha (\lfloor k^*_{\text{\RBY}} / U \rfloor \wedge K)/ (K\ell_K)},$
and denote this discovery set as $\Dcal_{\rmRBY}$.
This leads us to our guarantee for \RBY.

\begin{theorem}\label{thm:rounded-by}
    For any arbitrarily dependent p-values $(P_1, \dots, P_K)$, the \RBY\ procedure ensures $\FDR \leq \alpha$ and \RBY\ rejects all hypotheses rejected by \BY, i.e., $ \Dcal_{\rmRBY} \supseteq \Dcal_{\rmBY}$. Further, $\prob{\Dcal_{\RBY} \supset \Dcal_{\BY}} > 0$ is true if and only if
    \begin{align}
    \prob{\exists i \in [K]: \alpha_{k^*_{\BY} + 1} / \ell_K < P_i \leq \alpha  / \ell_K} > 0. \label{eq:p-value-improv}
\end{align}
\end{theorem}

\eqref{eq:p-value-improv} is an extremely weak condition, and indeed one should clearly expect it to hold in almost all situations where the p-values are continuous (as long as their dependence structure is not too adversarially constructed).

\begin{proof}
We will show that the \RBY\ procedure controls FDR and can only improve power by seeing that it is the result of applying \UEBH\ to calibrated p-values. We defer proof of strict power improvement under \eqref{eq:p-value-improv} to \Cref{sec:by-improve}.

A \emph{calibrator} $f: [0, 1] \mapsto [0, \infty)$ is an upper semicontinuous function that satisfies $\int_0^1 f(x)\ dx = 1$ which implies that $f(P)$ is an e-value for any p-value $P$. Utilizing this connection, we can then apply \UEBH\ to the calibrated p-values.
The BY calibrator \citep{xu_post-selection_inference_2022} is defined as follows:
\begin{align}
    f^{\rmBY(\alpha, K)}(x) \coloneqq K\alpha^{-1}
    \left(\left\lceil x(\alpha/ (K\ell_K))^{-1} \right\rceil \vee 1 \right)^{-1}\ind{x \leq \alpha^{-1}\ell_K}. \label{eq:by-calibrator}
\end{align}
Hence, we can produce arbitrarily dependent e-values
${(f^{\rmBY(\alpha, K)}(P_1), \dots, f^{\rmBY(\alpha, K)}(P_K))}$ --- applying the e-BH procedure to e-values constructed using the BY calibrator is equivalent to running the BY procedure on the original p-values.
Now note the following for any $i \in [K]$:
\begin{align}
    &\ind{U \leq \alpha_i f^{\rmBY(\alpha, K)}(P_{(i)})}
    = \dynind{U \leq i\left(\left\lceil P_{(i)}(\alpha/ (K\ell_K))^{-1} \right\rceil \vee 1 \right)^{-1}\ind{P_{(i)} \leq \alpha / \ell_K}}\\
    &=\dynind{\left(\left\lceil P_{(i)}(\alpha/ (K\ell_K))^{-1} \right\rceil \vee 1 \right)\leq (i / U) \wedge K}
    =\ind{P_{(i)} \leq \alpha(\lfloor i / U \rfloor \wedge K)\rfloor / (K\ell_K)},
\label{eq:threshold-rej}
\end{align}
where the last line is because $i \in [K]$, and $\lceil x \rceil \leq y \Leftrightarrow x \leq \lfloor y \rfloor$ for any $x, y \geq 0$.
Thus, we have shown \RBY\ has the same rejection behavior as \UEBH\, i.e., the exact conditions under which the $i$th smallest p-value is rejected by \RBY\ are the same conditions under which the $i$th largest calibrated e-value is rejected by \UEBH\, as characterized in \eqref{eq:uni-swap}. Consequently, \RBY\ inherits its properties from \Cref{thm:jrebh}. This justifies our desired results.
\end{proof}

\begin{remark}
    \citet{guo_control_false_2008} present an instance where the FDR control of BY is tight, that is, exactly $\alpha$. Notably, the instance does not satisfy \eqref{eq:p-value-improv} by explicitly constructing a degenerate distribution where the p-values only take on values $\alpha i / (K\ell_K)$ for $i \in [K]$. We note in \Cref{sec:tight-fdr} that \RBY\ behaves identically to the BY procedure in such a setting and empirically verify that this is indeed the case.
\end{remark}
\ifarxiv{}{\vspace{-10pt}}
\boldparagraph{FDR control of any randomized reshaped procedure} The above proof of \Cref{thm:rounded-by} uses the FDR control of the p-to-e calibration guarantee. Alternatively, we can also prove a randomized version of the superuniformity lemma of \cite{blanchard_two_simple_2008} and \citet{ramdas_unified_treatment_2019}, which is used to justify FDR control of multiple testing procedures that utilize a reshaping function to handle arbitrary dependence. Define a \emph{reshaping function} $\beta: [0, \infty) \rightarrow [0, \infty)$ to be a (nondecreasing) function that satisfies $\beta(r) = \int_0^r x\ d \nu(x)$ for probability measure $\nu$ over $[0, \infty)$.
\begin{lemma}[Randomized superuniformity lemma]
\label{lemma:rand-superuniform}
Let $P$ be a superuniform random variable that can be arbitrarily dependent with a positive random variable $R$, and $U$ be a superuniform random variable that is independent of both $P$ and $R$. Let $c$ be a nonnegative constant and $\beta$ be a reshaping function. Then, the following holds:
\begin{align}
    \expect\left[\frac{\ind{P \leq c\beta(R / U)}}{R}\right] \leq c.
\end{align}
Note that this recovers the original superuniformity lemmas (\citealt[Lemma 3.2]{blanchard_two_simple_2008} and \citealt[Lemma 1]{ramdas_unified_treatment_2019}) when $U = 1$.
\end{lemma}

We include the proof of \Cref{lemma:rand-superuniform} in \Cref{sec:rand-superuniform-proof}. As a consequence, we can show FDR control and an improvement from randomization for any reshaped variant of the BY procedure. Let the \emph{$\beta$-reshaped BY\ procedure} make the following number of discoveries:
\begin{align}
    k^*_\beta \coloneqq \max \left\{i \in [K]: P_{(i)} \leq \frac{\alpha \beta(i)}{K}\right\}     ,
\end{align}
and rejects the $k^*_\beta$ smallest p-values to form the discovery set $\Dcal_{\beta}$. Now, define
\begin{align}
        k^*_{\beta,\U} \coloneqq \max \left\{i \in [K]: P_{(i)} \leq \frac{\alpha \beta(i / U)}{K}\right\}.
\end{align}
The $\beta$-reshaped \RBY\ rejects the $k^*_{\beta,\U}$ smallest p-values and outputs the discovery set $\Dcal_{\beta,\U}$.

\begin{theorem}\label{thm:reshaped-rand-by}
    For any reshaping function $\beta$, the $\beta$-reshaped \RBY\ procedure ensures $\FDR \leq \alpha$ and rejects a superset of $\beta$-reshaped \BY, i.e., $\Dcal_{\beta,\U} \supseteq \Dcal_\beta$. Further, $\prob{\Dcal_{\beta, \U} \supset \Dcal_\beta}> 0$ iff \begin{align}
    \dynprob{\exists i \in [K]: \frac{\alpha\beta(k^*_{\beta} + 1)}{K} < P_i \leq \max_{r > k^*_{\beta}}\ \frac{\alpha\beta(r)}{K}} > 0. \label{eq:reshap-improv}
\end{align}
\end{theorem}

As before, \eqref{eq:reshap-improv} is a very weak condition, and should be easily satisfied in almost all cases where the p-values are continuous (as long as their dependence structure is not too adversarially constructed) and for reasonable $\beta$-reshaping functions with a measure $\nu$ that allocates probability on all of $[0, K]$.

We defer the proof to \Cref{sec:reshape-proof}.
Note that \Cref{thm:reshaped-rand-by} implies \Cref{thm:rounded-by} since the BY procedure is a special case of $\beta$-reshaped BY with $\beta(r) = (\lfloor r \rfloor \wedge K) / \ell_K$, so FDR control of the \RBY\ procedure can be proven through \Cref{lemma:rand-superuniform} as well. There may be some direct relationship that can be shown between calibration functions and reshaping functions that would unify all these results, and we leave that to future work.
\added{
\begin{remark}
    One might naturally wonder whether the BH procedure \citep{benjamini_controlling_false_1995} can also be improved using randomization, as \citet{wang_false_discovery_2022} showed that it is a special case of the \emph{boosted} e-BH procedure, and~\cite{ignatiadis_asymptotic_compound_2025} point out that every FDR procedure is an instance of the (non-boosted) e-BH procedure with compound e-values.  \citet{li_note_e-values_2025a} defined e-values for which the application of e-BH recovers the BH procedure as follows:
    \begin{align}
        X_i = \frac{\ind{P_i \leq \alpha_{k^*_{\rmBH}}}}{\alpha_{k^*_{\rmBH}}},
    \end{align} where $k^*_{\rmBH}$ is the number of discoveries made by BH. However, these e-values cannot be improved, i.e., $X_i = 1 / \alpha_{k^*_{\rmBH}}$ iff $i$ is rejected by BH and $X_i = 0$ otherwise. Thus, improving the BH procedure using randomization would require further developments in the e-value formulation of BH, and it is a topic for future work.
\end{remark}
}

\section{Randomized Hommel procedure for the global null}
\label{sec:rand-simes}
In addition to FDR control in multiple testing, we can also apply our randomization ideas to testing the global null using the procedure of \citet{hommel_tests_overall_1983} (i.e., the variant of \citet{simes_improved_bonferroni_1986} that is valid under arbitrary dependence). The global null hypothesis is the hypothesis that every individual null hypothesis is true. For arbitrarily dependent p-values, the Hommel procedure simply rejects the global null if the BY procedure has made at least a single discovery. Testing the global null can be formulated as the following hypothesis test:
\begin{align}
    H_0^\gl: \Ncal = [K]\text{ vs. }H_1^\gl: \Ncal \subset [K].
\end{align}

To see how the Hommel procedure controls type I error, we note the following fact: for any procedure $\mathcal{A}$ which rejects the global null if and only if a multiple testing procedure $\mathcal{B}$ rejects any hypothesis, the following relationship will exist between the type I error of $\mathcal{A}$ and FDR of the discovery set of $\Dcal$ that is output by $\mathcal{B}$ when $H_0^\gl$ is true:
\begin{align}
    \prob{\mathcal{A} \text{ rejects }H_0^\gl} = \prob{\Dcal > 0} = \expect\left[\ind{\Dcal > 0}\right] = \FDR. \label{eq:type-I-fdr}
\end{align} Note that the last equality is because the FDP is 1 if any rejection is made under $H_0^\gl$. Thus, if $\mathcal{B}$ ensures $\FDR \leq \alpha$, the corresponding procedure $\mathcal{A}$ has type I error controlled under $\alpha$ as well. Since the BY procedure controls FDR, the Hommel procedure also controls type I error. Thus, we derive a randomized improvement of the Hommel procedure by defining the \emph{U-Hommel procedure} to reject the global null when \RBY\ makes a single discovery.

Another way of viewing the Hommel and U-Hommel procedures is as p-merging (p-value merging) procedures, i.e., a function that takes $K$ p-values as input and outputs another p-value that tests the intersection hypothesis of the input p-values; the Hommel and U-Hommel p-values are defined as follows:
\begin{align}
    P_{\Hommel} \coloneqq \min_{i \in [K]} \frac{P_{(i)}K \ell_K}{i}, \qquad P_{\UHommel} \coloneqq \min_{i \in [K]} \frac{P_{(i)}K \ell_K}{(\lfloor i / U\rfloor \wedge K)}.
\end{align} The Hommel (U-Hommel) procedure rejects at level $\alpha$ if and only if $P_{\Hommel} \leq \alpha$ ($P_{\UHommel} \leq \alpha)$.
Now,  \Cref{thm:rounded-by} and \eqref{eq:type-I-fdr} imply the following.
\begin{theorem}
    \label{thm:uhommel}
    For arbitrarily dependent p-values $(P_1, \dots, P_K)$, the U-Hommel procedure has type I error at most $\alpha$, i.e., $P_{\UHommel}$ is a p-value under $H_0^\gl$. In addition, $P_{\UHommel} \leq P_{\Hommel}$ almost surely and $\prob{P_{\UHommel} < P_{\Hommel}} > 0$ iff $\prob{\exists i \in [K - 1]: P_{i} > 0} > 0$.
\end{theorem}
Note that, conditioned on the input p-values, $P_{\UHommel}$ will be less than $P_{\Hommel}$ iff one of the $K - 1$ smallest p-values is positive and $U$ is sufficiently small --- hence, there is always a positive probability that $P_{\UHommel} < P_{\Hommel}$ unless there is no probability that one of the smallest $K - 1$ p-values is positive.

\boldparagraph{Closed testing for FWER control} We can utilize the \UHommel\ p-merging function to develop a randomized closed testing procedure that controls the \emph{family-wise error (FWER)}, which is defined as the probability of making a single false discovery, i.e.,
\begin{align}
    \textnormal{FWER}\coloneqq \prob{\Dcal \cap \Ncal \neq \emptyset}.
\end{align}

For each $I\subseteq [K]$, let $H_I \coloneqq \bigcap\limits_{i \in I} H_i$ be the intersection hypothesis for $I$, where $H_i$ is the $i$th hypothesis. We call $\varphi^I \in \{0, 1\}$ an \emph{$\alpha$-level test} for $H_I$
if $\prob{\varphi^I = 1} \leq \alpha$ when $H_I$ is true ($H_i$ is true for each $i \in I$) for some fixed $\alpha \in [0, 1]$.
A \emph{closed testing procedure} using a family of $\alpha$-level tests $(\varphi^I)_{I\subseteq [K]}$ produces the following discoveries:
\begin{align}
    \Dcal_{\textnormal{closed}} \coloneqq \{i \in [K]: \varphi^I = 1 \text{ for all }I\subseteq[K]\text{ s.t.\ }i \in I\}.\label{eq:closed-discovery-set}
\end{align}

\begin{fact}[Section 2 of \citealt{marcus_closed_testing_1976}]\label{fact:fwer-closed}
    Let $(\varphi^I)_{I\subseteq [K]}$ be a family of $\alpha$-level tests for a fixed $\alpha \in [0, 1]$. Then, $\Dcal_{\textnormal{closed}}$ of the corresponding closed testing procedure satisfies $\FWER \leq \alpha$.
\end{fact}
\citet{sonnemann_vollstaendigkeitssaetze_fuer_1988,sonnemann1982allgemeine,sonnemann_general_solutions_2008} showed that all admissible FWER controlling procedures are closed testing procedures.
Now, let $P_{(i, I)}$ be the $i$th smallest p-value among p-values corresponding to hypotheses in $I$.
Based on the Hommel p-merging function, define the following $\alpha$-level test for $H_I$:
\begin{align}
P_{\Hommel}^I \coloneqq \min_{i \in |I|}\ \frac{P_{(i, I)}|I|\ell_{|I|}}{i}, \qquad \varphi^I_{\alpha\textnormal{-}\Hommel} \coloneqq \ind{P_{\Hommel}^I \leq \alpha},
\end{align}
\citet{meijer_hommels_procedure_2019} develop a more efficient way of computing $\Dcal_{\textnormal{closed}}$ of the closed testing procedure for $(\varphi^I_{\alpha\textnormal{-}\Hommel})_{I\subseteq [K]}$ using the following formula.
\begin{gather}
h(\alpha) \coloneqq \max\{i \in [K]: P_{(K - i + j)} > \alpha j / (i\ell_i) \text{ for all }j \in [i]\},\\
\Dcal_{\Hommel} \coloneqq \{i \in [K]: P_i \leq \alpha / (h(\alpha)\ell_{h(\alpha)})\}.     \label{eq:closed-hommel-discovery-set}
\end{gather}
\begin{fact}[Lemma 1 of \citealt{meijer_hommels_procedure_2019}]
    $\Dcal_{\textnormal{closed}}$ output by the closed testing procedure that uses $(\varphi^I_{\alpha\textnormal{-}\Hommel})_{I\subseteq [K]}$ is equivalent to $\Dcal_{\Hommel}$ and ensures $\FWER \leq \alpha$.
\end{fact}
For our randomized version, define the local tests as follows:
\begin{align}
P_{\UHommel}^I \coloneqq \min_{i \in |I|}\ \frac{P_{(i, I)}|I|\ell_{|I|}}{(\lfloor i / U \rfloor \wedge |I|)}, \qquad \varphi^I_{\alpha\textnormal{-}\UHommel} \coloneqq \ind{P_{\UHommel}^I \leq \alpha},
\end{align}
Define the following corresponding discovery set:
\begin{gather}
h_{\U}(\alpha) \coloneqq \max\{i \in [K]: P_{(K - i + j)} > \alpha (\lfloor j / U \rfloor \wedge i) / (i\ell_i) \text{ for all }j \in [i]\},\\
\Dcal_{\UHommel} \coloneqq \{i \in [K]: P_i \leq \alpha \lfloor U^{-1} \wedge h_{\U}(\alpha)  \rfloor/ (h_{\U}(\alpha)\ell_{h_{\U}(\alpha)})\}.     \label{eq:closed-uhommel-discovery-set}
\end{gather}
\begin{theorem}
    \label{thm:closed-uhommel}
    $\Dcal_{\textnormal{closed}}$, produced by the closed testing procedure for $(\varphi_{\alpha\textnormal{-}\UHommel}^I)_{I \subseteq [K]}$, is equivalent to $\Dcal_{\UHommel}$ and ensures $\FWER \leq \alpha$.
    Further, $\Dcal_{\UHommel} \supseteq \Dcal_{\Hommel}$ holds almost surely.
\end{theorem}
We defer proof of this theorem to \Cref{sec:uhommel-proof}.
One can also derive a family of randomized procedures through closed testing by using the randomized admissible p-merging functions we introduce in \Cref{sec:admissible}.
\section{Randomized post-selection inference with confidence intervals}\label{sec:posi}

    We can also apply our randomization results to the post-selection inference analogs of e-BH and BY \citep{benjamini_false_discovery_2005,xu_post-selection_inference_2022}. In the post-selection inference problem, we consider $K$ parameters of initial interest, $(\vartheta_1, \dots, \vartheta_K)$, and denote their true values as
    \({\boldsymbol\theta}^* \coloneqq (\theta_1^*, \dots, \theta_K^*)\).  Formally, a parameter is defined as a functional $\vartheta: \Pcal \rightarrow \Theta$ where $\Pcal$ is the universe of distributions and $\Theta$ is the support of parameter values. Let $\pdist^* \in \Pcal$ denote the true distribution, and consequently $\theta_i^* = \vartheta_i(\pdist^*) \in \Theta$ (here we let the range of all parameters be the same for simplicity). For any arbitrary $\pdist \in \Pcal$, we let $\sprob_{\pdist}$ and $\expect_{\pdist}$ denote the probability of an event and expectation of a random variable under $\pdist$, respectively.

    We observe some data, \(\mathbf{D} = (D_1, \dots, D_K)\) that is drawn from the true distribution $\pdist^*$. For each $i \in [K]$, we assume that we can construct confidence intervals for each parameter $\vartheta_i$, i.e., we have access to $C_i : [0, 1] \rightarrow 2^\Theta$ that is constructed from $D_i$ and satisfies the following property:
\begin{align}
    \sprob_{\pdist}(\vartheta_i(\pdist) \in C_i(\alpha)) \geq 1 - \alpha \text{ for all }\pdist \in \Pcal \text{ and } \alpha \in [0, 1]. \label{eq:ci-def}
\end{align}

The data $\mathbf{D}$ is also applied to a (possibly unknown) selection rule \(\selalg: \mathbf{D} \to 2^{[K]}\) to produce a selection of parameters \(\selset \coloneqq \selalg(\mathbf{D}) \subseteq [K]\) that we wish to perform inference on.

    Our goal is to formulate a (possibly data-dependent) method for choosing $\{\alphath_i\}_{i \in \selset}$. We want the confidence intervals we produce, $(C_i(\alphath_i))_{i \in \selset}$, to be valid, i.e., most of the true parameter values of the selected parameters in $\selset$ are covered by their respective CIs. Similar to FDR control for multiple testing, we can control the \emph{false coverage rate (FCR)} to be bounded under a fixed level $\alpha \in [0, 1]$. The FCR is the expected proportion of selected parameters whose true values are not covered --- we formally define it in the following fashion.
\begin{align}
    \FCP \coloneqq \frac{\sum_{i \in \selset}\ind{\theta_i^* \not\in C_i(\alphath_i)}}{|\selset| \vee 1}, \qquad \FCR \coloneqq \expect\left[\FCP\right].
\end{align} Here, $\vee$ is the maximum operator.

To generalize our randomization results to post-selection inference, we also require the concept of an e-confidence interval \citep{vovk_confidence_discoveries_2023,xu_post-selection_inference_2022}.
\begin{definition} \label{def:e-ci}
    $C(\alpha)$ is said to be a \emph{$(1 - \alpha)$-e-confidence interval (e-CI)} for the parameter $\vartheta$ if there exists a family of e-values $\{X(\theta)_{\theta \in \Theta}\}$ associated with parameter $\vartheta$, i.e.,  $X(\theta)$ is an e-value for the set of distributions $\Pcal_\theta = \{\pdist \in \Pcal: \vartheta(\pdist) = \theta\}$ for each $\theta \in \Theta$, where $C(\alpha)$ satisfies the following identity for all $\alpha \in [0, 1]$:
    \begin{align}
        C(\alpha) = \left\{\theta \in \Theta: X(\theta) < \alpha^{-1}\right\}
    \end{align}
    Furthermore, all e-CIs are also CIs, i.e., $C(\alpha)$ is a $(1 - \alpha)$-CI for $\vartheta$ and $C$ satisfies \eqref{eq:ci-def} \citep[Prop. 1]{xu_post-selection_inference_2022}.
\end{definition}

The e-BY procedure \citep{xu_post-selection_inference_2022} requires that $C_i(\alpha)$ are e-CIs for each $i \in [K]$ and $\alpha \in [0, 1]$ and outputs the following error levels for each $i \in \selset$:
\begin{align}
    \alphath_i = \alpha |\selset| / K \text{ for each }i \in \selset, \tag*{(e-BY procedure)}
\end{align}
We can define variants of the e-BY procedure similarly to the ones for the e-BH procedure. For simplicity, we will only explicitly define and prove properties of the Ue-BY procedure, which is defined as follows:
\begin{align}
\alphath_i = \alpha |\selset| / (UK) \text{ for each }i \in \selset. \tag*{(Ue-BY procedure)}
\end{align}
As seen previously, $U$ is a uniform random variable over $[0, 1]$ independent of $\mathbf{X}$. Now, we have the following guarantee:
\begin{theorem}\label{thm:ueby}
    Assume that, for each $i \in [K]$, $C_i(\alpha)$ is an e-CI for all $\alpha \in [0, 1]$. Then, the Ue-BY procedure ensures $\FCR \leq \alpha$ for any $\alpha \in [0, 1]$ under any dependence structure among $D_1, \dots, D_K$ and selection rule $\selalg$.
\end{theorem}
\begin{proof}
    Since $C_i(\alpha)$ is an e-CI, we can denote the e-value for the true parameter value $\theta_i^*$ in the associated family of e-values as $X_i(\theta^*)$. Let $\widehat{\alpha}_{|\selset|} \coloneqq \alpha|\selset|/ K$. Now, let $S_{\widehat{\alpha}_{|\selset|}}(X_i(\theta_i^*))$ be the jointly rounded e-value derived from $X_i(\theta_i^*)$ that utilizes the external randomness $U$ from \eqref{eq:joint-rounding}. We can make the following derivations about the FCR:
\begin{align}
    &\FCR
    = \expect\left[\frac{\sum\limits_{i = 1}^K \ind{\theta_i^* \not\in C_i(\alpha|\selset| / (UK))} \cdot \ind{i \in \selset}}{|\selset| \vee 1}\right]
    \labelrel{=}{rel:family-e-ci} \expect\left[\frac{\sum\limits_{i = 1}^K \ind{X_i(\theta_i^*) \geq U / \alpha_{|\selset|}} \cdot \ind{i \in \selset}}{|\selset| \vee 1}\right]\\
    &= \expect\left[\frac{\sum\limits_{i = 1}^K \ind{\alpha_{|\selset|}X_i(\theta_i^*) \geq U } \cdot \ind{i \in \selset}}{|\selset| \vee 1}\right]
    \labelrel{=}{rel:stoch-evalue} \expect\left[\frac{\sum\limits_{i = 1}^K \ind{S_{\alpha_{|\selset|}}(X_i(\theta_i^*)) \geq \alpha_{|\selset|}} \cdot \ind{i \in \selset}}{|\selset| \vee 1}\right]\\
    &\labelrel{\leq}{eq:det-ub} \expect\left[\frac{\sum\limits_{i = 1}^K \alpha_{|\selset|}S_{\alpha_{|\selset|}}(X_i(\theta_i^*))}{|\selset| \vee 1}\right]
    = \frac{\alpha}{K}\sum\limits_{i = 1}^K\expect\left[S_{\alpha_{|\selset|}}(X_i(\theta_i^*))\right] \leq \alpha.
\end{align}
Equality \eqref{rel:family-e-ci} holds by the definition of an e-CI. Equality \eqref{rel:stoch-evalue} comes from observing the formulation of stochastic rounding in \eqref{eq:rand-e} and noting that $\ind{S_{\widehat\alpha}(X) \geq \widehat\alpha{-1}} = \ind{\widehat\alpha X \geq U}$ for any $\widehat\alpha \in [0, 1]$. Thus, we have shown our desired result.
\end{proof}
Define the CI-BY procedure (under general dependence) from \citet{benjamini_false_discovery_2005} for post-selection inference as choosing the following error levels:
\begin{align}
    \alpha_i = \alpha |\selset| / (K\ell_K) \text{ for each }i \in \selset. \tag*{(CI-BY procedure)}
\end{align}
\citet{xu_post-selection_inference_2022} showed that the CI-BY procedure can be formulated as a special case of the e-BY procedure that results from calibrating standard CIs to e-CIs, in the same way the BY procedure for multiple testing is a special case of e-BH.
\begin{fact}[Theorem 1 of \citet{xu_post-selection_inference_2022}]
    If $C(\alpha)$ is a CI that satisfies \eqref{eq:ci-def}, then
    \begin{align}
    C^{\caltext}(\alpha) \coloneqq C(f^{-1}(\alpha^{-1}))
    \end{align}is an e-CI for each $\alpha \in [0, 1]$.
\end{fact}
One can derive the FCR control of the CI-BY procedure by viewing it as an application of the e-BY procedure to calibrated CIs.
\begin{fact}[Theorem 4 of \citet{benjamini_false_discovery_2005}; Corollary 2 of \citet{xu_post-selection_inference_2022}]\label{fact:by-from-eby}
    The CI-BY procedure controls $\FCR \leq \alpha$ for any selection rule $\selalg$ and dependence structure among $D_1, \dots, D_K$.
    Further, the CI-BY procedure is equivalent to applying the e-BY procedure to the calibrated e-CIs $C^{\caltext}$ with the calibrator $f$ defined for the BY procedure in \eqref{eq:by-calibrator}.
\end{fact}

Now, we can define the UCI-BY\ procedure for post-selection inference as follows:
\begin{align}
    \alpha_i = \alpha (\lfloor |\selset|U^{-1} \rfloor \wedge K) / (K\ell_K) \text{ for each }i \in \selset, \tag*{(UCI-BY procedure)}
\end{align}
By combining \Cref{fact:by-from-eby} and \Cref{thm:ueby}, we get that \RBY\ also has FCR control.
\begin{theorem}
    The UCI-BY procedure controls $\FCR \leq \alpha$ for any selection rule $\selalg$ and dependence structure among $D_1, \dots, D_K$.
\end{theorem} As a result, we have developed randomized versions of the e-BY and CI-BY procedures for post-selection inference that produce CIs that are always tighter than their deterministic versions, and are strictly tighter (i.e., have larger $\alpha_i$) with nonzero probability.

\section{Simulations}
\label{sec:Simulations}
\begin{figure}[h]
    \centering
    \begin{subfigure}{\textwidth}
        \centering
\includegraphics[width=\textwidth]{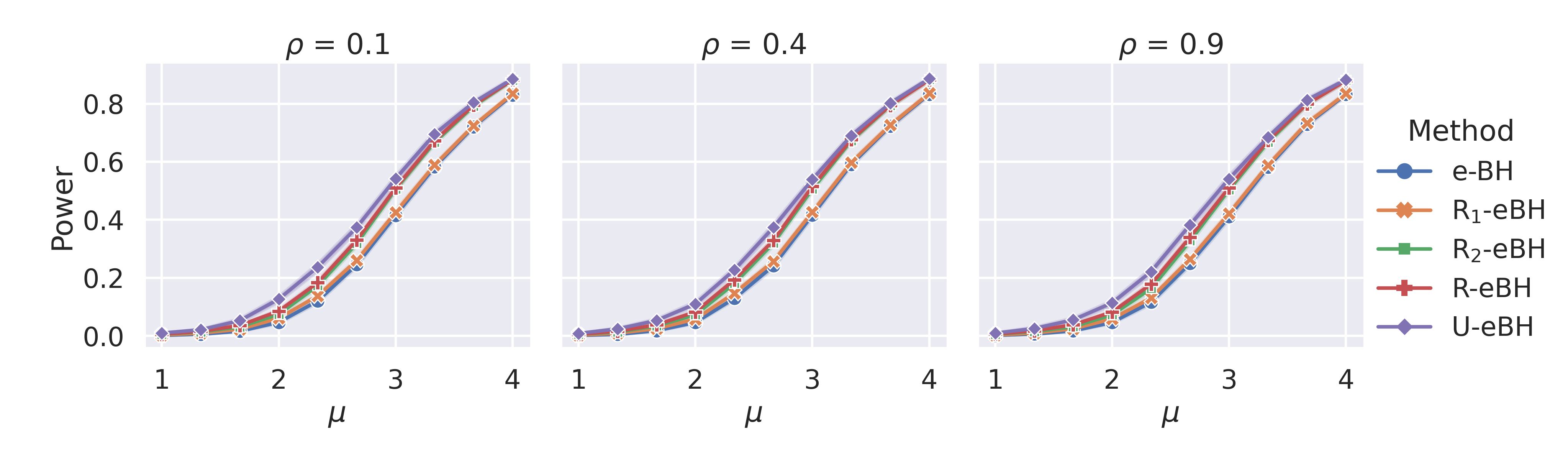} 
    \caption{Positive dependence: \RtwoEBH, \RbothEBH, and \UEBH\ all have more power than e-BH.}
    \end{subfigure}

    \begin{subfigure}{\textwidth}
        \centering
\includegraphics[width=\textwidth]{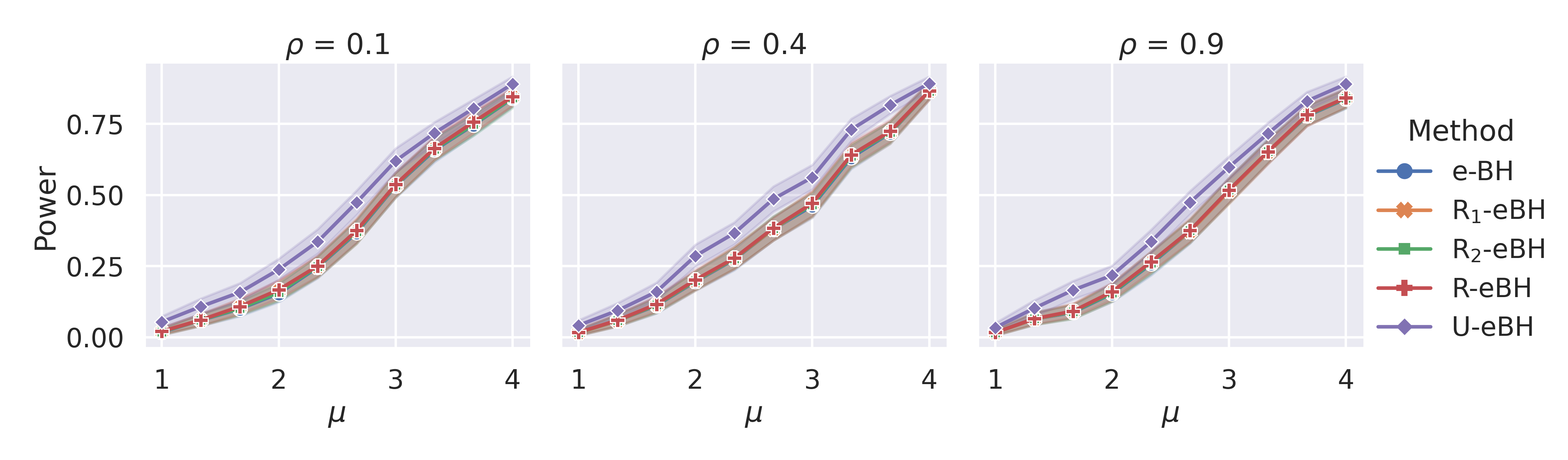}
    \caption{Negative dependence: \UEBH\ is much more powerful than the other methods.}
    \end{subfigure}
\caption{\added{Plots of the power for randomized e-BH methods and e-BH. The hypotheses are one-sided tests of the mean of a standard Gaussian with covariance parameterized by $\rho$ (larger $\rho$ means larger covariance) and non-null means of $\mu$. \UEBH\ is the most powerful in each setting, and randomized procedures uniformly improve on e-BH.}}
    
    \label{fig:ebh-heatmap}
\end{figure}
\begin{figure}[h!]
    \centering
    \begin{subfigure}{\textwidth}
\includegraphics[width=\textwidth]{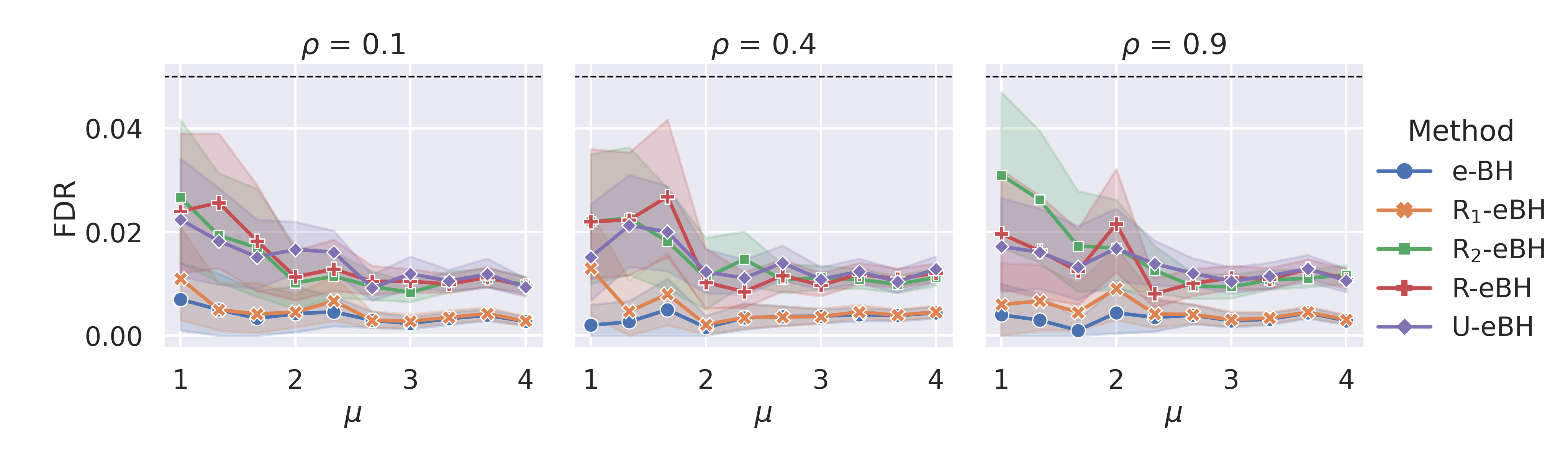} 

    \caption{Positive dependence: \RtwoEBH, \RbothEBH, and \UEBH\ all have similar FDR.}
    \end{subfigure}

    \begin{subfigure}{\textwidth}
\includegraphics[width=\textwidth]{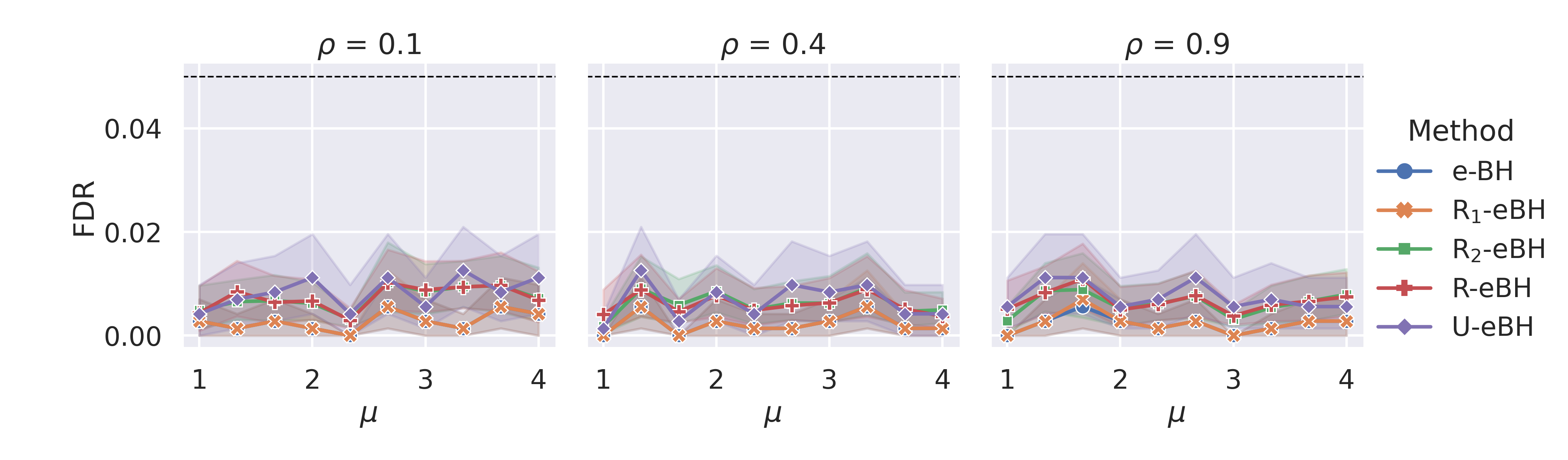} 

    \caption{Negative dependence: \UEBH\ has much larger achieved FDR across all settings.}
    \end{subfigure}
\caption{Plots of the achieved FDR of randomized e-BH methods and e-BH across different signal strengths $\mu$ and correlations $\rho$. Note that all the FDR values are under the desired level of $\alpha=0.05$. Further, the randomized e-BH methods have higher realized FDR, i.e., tighter FDR control, than e-BH (but still below $\alpha$).
    }
    \label{fig:ebh-fdr}
\end{figure}
We run simulations of the e-BH procedure and the BY procedure in the Gaussian setting, where we test $K = 100$ hypotheses. Let $\pi_0 = (K - |\Ncal|)/ K = 0.3$, i.e., the proportion of hypotheses where the null is false. For each hypothesis, we sample $X_i \sim (\mu_i, 1)$, and perform the following one-sided hypothesis test:
\begin{align}
    H_0:&\ \mu_i = 0\ \text{vs.}\ H_1:\ \mu_i \geq 0,
\end{align} for each $i \in [K]$. We consider two dependence settings:
\begin{enumerate}
    \item Positive dependence: $\mathrm{Cov}(X_i, X_j) = \rho^{|i - j|}$ for each $i, j \in [K]$, i.e., the covariance matrix is Toeplitz.
    \item Negative dependence: $\mathrm{Cov}(X_i, X_j) = -\rho / (K - 1)$, i.e., the covariance is equivalent among all pairs of samples.
\end{enumerate}
We let $\mu$ range in $[1, 4]$, and $\rho$ range in $[0, 0.9]$.
For each $i \in [K]$, where the null hypothesis is false, we let $\mu_i = \mu$, i.e., all non-null hypotheses have a single mean. Power is defined as the expected proportion of non-null hypotheses that are discovered, i.e., $\expect[|\Dcal - \Ncal| / (K - |\Ncal|)]$. We calculate the power averaged over 500 trials.
In \Cref{fig:ebh-heatmap}, we can see that \RoneEBH\ slightly increases power, while \RtwoEBH\ provides a much larger improvement. Since the two types of rounding are complementary, \RbothEBH\ has the largest increase in power over e-BH. Further, we can see in \Cref{fig:ebh-fdr} that the FDR of all the methods is controlled to be below $\alpha=0.05$. Additional analysis of these results is in \Cref{sec:add-sim}.
\begin{figure}[h]
    \centering
    \begin{subfigure}{\textwidth}
    \centering
\includegraphics[width=\textwidth]{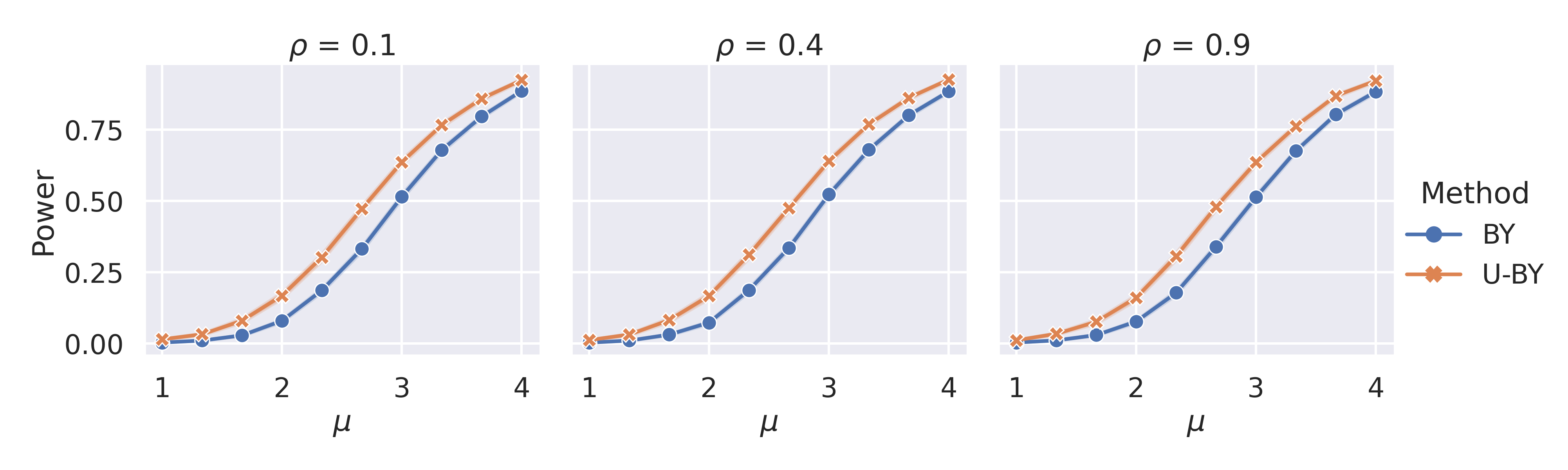}
    \caption{Positive dependence: \RBY\ has maximal power improvement when $\mu \in [2, 3]$.}
    
\end{subfigure}
\begin{subfigure}{\textwidth}
    \centering
\includegraphics[width=\textwidth]{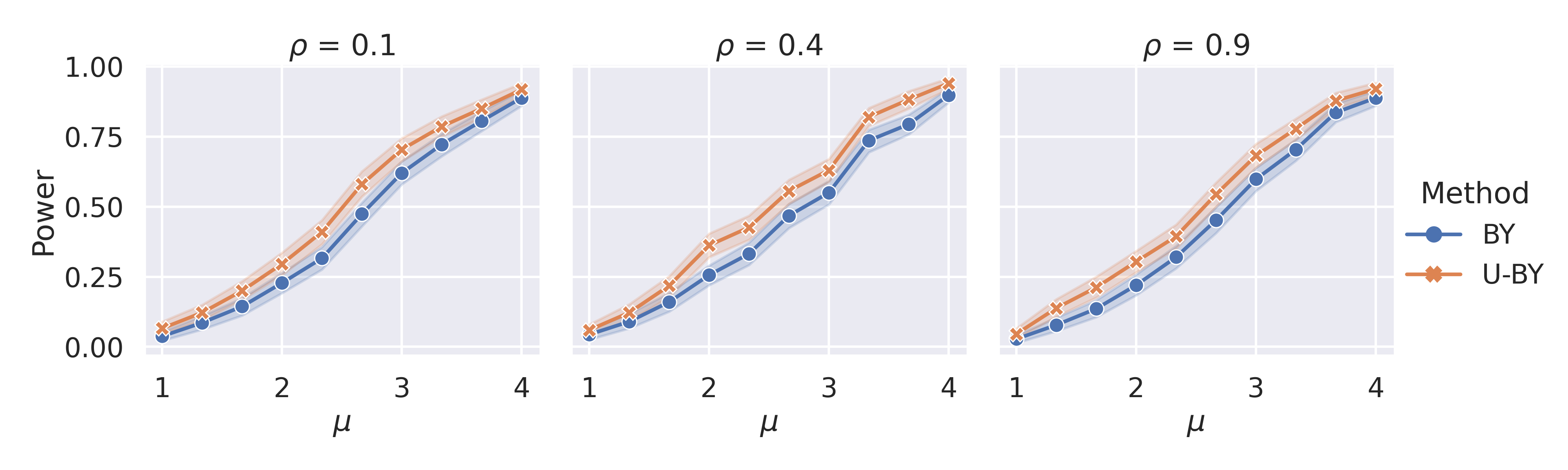}
    \caption{Negative dependence: the power improvement of \RBY\ varies across $\rho$ and $\mu$.}
    \end{subfigure}
    \caption{Plot of power of \RBY\ and BY at the corresponding value of $\mu$ and $\rho$. The hypotheses are one-sided tests of the mean of a standard Gaussian that has covariance parameterized by $\rho$ (null means are 0 and non-null means are $\mu$). \RBY\ dominates the BY procedure across all $\rho$ and $\mu$.}
    \label{fig:by-heatmap}
\end{figure}
In \Cref{fig:by-heatmap}, we compare the power of \RBY\ vs.\ BY. Similarly, \RBY\ uniformly dominates the BY procedure and shows a large amount of power improvement in most settings.
\section{Conclusion}

Stochastic rounding is a simple way to improve power for multiple testing procedures with e-values, and also procedures based on p-values that a priori appear to have nothing to do with e-values. We have primarily discussed the use of stochastic rounding to improve the e-BH and BY procedures for FDR control, and we also demonstrated that the idea works for other multiple testing procedures, such as closed testing or post-selection inference with confidence intervals.
We also illuminate some connections between stochastic rounding and using e-values as weights for p-value procedures in \Cref{sec:weighting}.

An open question that remains for stochastic rounding is: what is the rounding scheme that maximizes expected discoveries for each realization of e-values? For each realization, we can choose a different rounding scheme if necessary, since our randomization is independent of the e-values. As evidenced by \UEBH\ being more powerful than \RtwoEBH, some rounding schemes are powerful than others. We have not proven optimality of \UEBH\ or for any of our randomized e-BH procedures, so either proving their optimality or deriving a novel optimal procedure would complete the picture of what the best randomization is.

Despite hesitations that may arise in using these randomized methods in practice, we believe our results are of significant theoretical and methodological interest. At the very least, our results should be quite surprising to the reader, since even the possibility that randomization can (usually strictly) improve the power of deterministic procedures in multiple testing has not been previously considered, let alone constructively demonstrated.

\paragraph{Acknowledgments} We thank Ruodu Wang for his insightful comments and discussions, including pointing out the application of stochastic rounding to admissible p-value merging in \Cref{sec:admissible}.  
\bibliography{ref}

\appendix

\newpage

\section{Omitted proofs}
\label{sec:omitted-proofs}
\subsection{Proof for \Cref{thm:grid-ebh}}
\label{sec:ebh-improve}
We now prove the claim about strict power improvement.

Define the following event
\begin{align}
    \mathcal{E}_\varepsilon \coloneqq \{\exists k \in [K - k^*]: X_{[k^* + j]} > K/(\alpha(k^* + k + 1)) + \varepsilon \text{ for all }j \in [k]\}.
\end{align}
Note that \eqref{eq:kcal-cond} implies that $\mathcal{E}_\varepsilon$ occurs with a positive probability for some fixed $\varepsilon > 0$. When $\mathcal{E}_\varepsilon$ occurs, we note that if all e-values in $\Dcal_{\rmEBH}$ are successfully rounded (which will happen with fixed positive probability as well), there will be $k \geq 1$ additional discoveries made. Thus, we get $\prob{\Dcal_1 \supset \Dcal_{\rmEBH}} > 0$.

For $\prob{\Dcal_1 \supset \Dcal_{\rmEBH}} > 0$ to be true, we note that $\mathcal{E}_\varepsilon$ must occur with positive probability for some fixed $\varepsilon > 0$. This implies that \eqref{eq:kcal-cond} holds and we get our desired result.
\subsection{Proof of \Cref{prop:uebh-better-than-r2ebh}}
\label{sec:uebh-better-proof}
We observe the following by definition of stochastic rounding:
\begin{align}
    \ind{S_{\alpha_{\ell(i)}}(X_i) \geq 1 / \alpha_{\ell(i)}} = \ind{U \leq \alpha_{\ell(i)}X_i}.
    \label{eq:uni-rep}
\end{align}

Let $L \coloneqq \{\ell(i): i \in [K] \}$ be the set of levels that are rounded to. We define the following set of hypotheses for each $i \in L$:
\begin{align}
    \Dcal_i \coloneqq \{j \in [K]: \alpha_{\ell(j)} \leq \alpha_{i}\} = \{j \in [K]: \rank(j) \geq i\}.
\end{align} By definition of $\ell$, the following properties are true:
\begin{gather}
    |\Dcal_i| = i\text{ for each }i \in L,\text { and }\Dcal_i \subset \Dcal_{i'} \text{ for each }i, i' \in L\text{ where }i < i'. \label{eq:d-subset}
\end{gather}
Further, we can only reject a number of hypotheses that is in $L$, i.e., $k^*_\U \in L$. This is because, for any $i \not\in L$, we know that $|\{j \in [K]: S_{\alpha_{\ell(j)}}(X_j) \geq 1 / \alpha_i\}| < i$ since for all $j \in [K]$ s.t. $X_j \geq X_{[i]}$, we have that $\ell(j) < i$. Hence, define $L(i) \coloneqq \min\{j \in L: j \geq i\}$, and we get the following relationship for any $i \geq k^* + 1$:
\begin{align}
    \ind{k_\U^* \geq i} = \ind{k^*_\U \geq L(i)}.
    \label{eq:Li-i}
\end{align}
Now, we derive the following property for each $i \in L$ where $i \geq k^* + 1$:
\begin{align}
&\ind{k^*_\U \geq i}
    = \ind{S_{\alpha_{\ell(j)}}(X_j) \geq 1 / \alpha_{\ell(j)} \text{ for all }j \in \Dcal_{i}}\\
    &=\ind{U \leq \alpha_{\ell(j)}X_j \text{ for all }j \in \Dcal_{i}}
    =\ind{U \leq \alpha_{i}X_{[i]}} \label{eq:uni-swap}.
\end{align} The first line by \eqref{eq:d-subset}, and the observation that if $S_{\alpha_{\ell(j)}}(X_j) = 0$ for any $j \in \Dcal_{i}$, there will be fewer than $i$ rejections.    The equality in the second line is by \eqref{eq:uni-rep}, and the final equality is because---by definition of $\ell$---there exists $j \in \Dcal_i$ s.t. $\ell(j) = i$ and $X_j = X_{[i]}$, and $\alpha_{\ell(j)}X_j \geq \alpha_i X_{[i]}$ for all $ j\in \Dcal_i$ by the definition of $\Dcal_i$.
Thus, we have:
\begin{align}
    &\expect[k^*_\U \mid X_1, \dots, X_K]
    = k^* + \sum\limits_{i = k^* + 1}^K\expect\left[\ind{k^*_\U \geq i} \mid X_1, \dots, X_K \right]\\
    &=k^* + \sum\limits_{i =k^* + 1}^K\prob{U \leq \alpha_iX_{[i]} \mid X_1, \dots, X_K}
    = k^* + \sum\limits_{i =k^* + 1}^K \alpha_i X_{[i]}
    \geq k^* + \sum\limits_{i =k^* + 1}^K \widehat\alpha^* X_{[i]}\\
    &= k^* + \sum\limits_{i =k^* + 1}^K \prob{U_i \leq \widehat\alpha^* X_{[i]} \mid X_1, \dots, X_K}
    =\expect[k^*_2 \mid X_1, \dots, X_K].
\end{align}
The second equality is by \eqref{eq:uni-swap}, and $U_i$ in the inequality are i.i.d. uniform random variables over $[0, 1]$ for each $i \in [K]$.

\subsection{Proof for \Cref{thm:rounded-by}}
\label{sec:by-improve}

We now prove the claim about strict power improvement.
Note that $\prob{\Dcal_{\RBY} \supset \Dcal_{\BY}} > 0$  iff there is positive probability on a p-value vector $(p_1, \dots, p_K) \in [0, 1]^K$ where more discoveries can be made by the randomization in \RBY.
Note that for a realized p-value $p_i$ such that $p_i \in (\alphath_{k^*_{\BY + 1}} / \ell_K, \alpha / \ell_K]$, there is a positive probability that it will be rejected by \RBY and would not have been rejected by BY.
Further, if no $p_i$ lies within that range, then \RBY\ would not make any extra discoveries over \BY, since the p-values corresponding to the hypotheses that were not rejected by BY would be larger than $\alpha / \ell_K$, which is the maximum level of rejection that \RBY\  can achieve. Thus, $\prob{\Dcal_{\RBY} \supset \Dcal_{\BY}} > 0$ holds iff \eqref{eq:p-value-improv} holds.

\subsection{Proof of \Cref{lemma:rand-superuniform}}
\label{sec:rand-superuniform-proof}
Define $R^* \coloneqq   \inf\ \{s \in [R, \infty): P_i \leq c \beta(s)\}$.
We note that
\begin{align}
    \ind{P \leq c\beta(R / U)} = \ind{P \leq c\beta(R)} + \ind{P > c\beta(R), U \leq R / R^*}.
\end{align} Hence, we simply need to prove that
\begin{align}
     \expect\left[\frac{\ind{P \leq c\beta(R)} + \ind{P > c\beta(R), U \leq R / R^*}}{R}\right] \leq c.
     \label{eq:superuniform-reform}
\end{align}
The following proof repeatedly uses the fact that \[
1 / z = \int_0^\infty \ind{r \geq z}/z^2\ dr,
\text{ for positive $z$}.
\]
We begin with the following equalities:
\begin{align}
    &\expect\left[\frac{\ind{P \leq c\beta(R)} + \ind{P > c\beta(R), U \leq R / R^*}}{R}\right]\\
    &= \expect\left[\int\limits_0^\infty\frac{\ind{r \geq R, P \leq c\beta(R)}}{r^2}\ dr + \frac{\ind{P > c\beta(R), U \leq R / R^*}}{R}\right]\\
    &=\expect\left[\int\limits_0^\infty\frac{\ind{r \geq R, P \leq c\beta(R), P \leq c\beta(r)}}{r^2}\ dr\right]+ \expect\left[\frac{\ind{P > c\beta(R), U \leq R / R^*}}{R}\right].\label{eq:split-rand}
    \end{align}
    Now, we can get the following equalities for the right summand:
    \begin{align}
        &\expect\left[\frac{\ind{P > c\beta(R), U \leq R / R^*}}{R}\right]\\
        &=\expect\left[\frac{\ind{P > c\beta(R)}\expect[\ind{U \leq R / R^*} \mid R, P]}{R}\right]\\
        &\labelrel{=}{eq:u-ind} \expect\left[ \frac{\ind{P > c\beta(R)}}{R^*}\right]\\
        &= \expect\left[\ind{P > c\beta(R)}\int\limits_0^\infty \frac{\ind{r \geq R^*}}{r^2} dr\right]\\
        &\labelrel{=}{eq:r-star-def} \expect\left[\ind{P > c\beta(R)}\int\limits_0^\infty \frac{\ind{r \geq R, P \leq c \beta(r)}}{r^2}\ dr\right].
        \label{eq:rhs-equality}
    \end{align}
    The equality in \eqref{eq:u-ind} is by the independence and the uniform distribution of $U$. The equality in \eqref{eq:r-star-def} is by the definition of $R^*$, i.e., $r\geq R^*$ if and only if  $P \leq c\beta(r)$ and $r \geq R$. Now, we can combine \eqref{eq:superuniform-reform}, \eqref{eq:split-rand} and \eqref{eq:rhs-equality} to derive the following:
    \begin{align}
    &\expect\left[\frac{\ind{P \leq c\beta(R / U)}}{R}\right]\\
    &=\expect\left[\int\limits_0^\infty\frac{\ind{r \geq R, P \leq c\beta(R), P \leq c\beta(r)}}{r^2}\ dr\right] + \expect\left[\int\limits_0^\infty \frac{\ind{r \geq R, P > c\beta(R), P \leq c \beta(r)}}{r^2}\ dr\right]\\
    &=\expect\left[\int\limits_0^\infty \frac{\ind{r \geq R, P \leq c \beta(r)}}{r^2}\ dr \right]=\int\limits_0^\infty \frac{\prob{r \geq R, P \leq c \beta(r)}}{r^2}\ dr\\
    &\labelrel{\leq}{eq:drop-r}\int\limits_0^\infty \frac{\prob{P \leq c \beta(r)}}{r^2}\ dr \labelrel{\leq}{eq:superuniformity} c\int\limits_0^\infty \frac{\beta(r)}{r^2}\ dr \\ &=c\int\limits_0^\infty x \int\limits_0^\infty \frac{\ind{r \geq x}}{r^2}\ dr\ d\nu(x) =  c\int_0^\infty d\nu(x) = c.
\end{align}
We get \eqref{eq:drop-r} by taking the probability of a superset, and \eqref{eq:superuniformity} is by superuniformity of $P$. Consequently, we have proved our lemma.

\subsection{Proof of \Cref{thm:reshaped-rand-by}}
\label{sec:reshape-proof}
Since $\beta$ is monotonically increasing, $\beta(k^*_\beta) \leq \beta(k^*_\beta / U_i)$ for every $i \in [K]$, so $k^*_\beta \leq k^*_{\beta,\U}$. Now, define the following quantities:
\begin{align}
    m(i) \coloneqq \max\{u \in [0, 1]: P_{(i)} \leq \alpha \beta(i / u) / K\}, \qquad t(i) \coloneqq \underset{j \geq \rank(i)}{\argmax}\ m(j).
\end{align} where $\rank(i) \coloneqq \sum_{j \in [K]} \ind{P_j \leq P_i}$ refers to the rank of $P_i$ in ascending order of p-values. We claim that
\begin{align}
    \ind{P_i \leq \alpha \beta(k^*_{\beta, \U} / U) / K} = \ind{P_i \leq \alpha \beta(t(i) / U) / K}.
    \label{eq:rej-eq}
\end{align}
If the right side equals 1, $k^*_{\beta, \U} \geq t(i) \geq \rank(i)$ by construction, so the left side equals 1. On the other hand, if the right side equals 0, then we know $U > m(j)$ for all $j \geq \rank(i)$. This means that at most $\rank(i) - 1$ hypotheses are rejected, and $\rank(i) - 1 < t(i)$, so the left side is also 0.
Now we can look at the FDR of the $\beta$-reshaped \RBY\ procedure.
\begin{align}
    \FDR &=\sum\limits_{i \in \Ncal} \expect\left[\frac{\ind{P_i \leq \alpha \beta(k^*_{\beta, \U} / U) / K}}{k^*_{\beta,\U} \vee 1}\right]\\
&\labelrel{\leq}{eq:rand-bigger} \sum\limits_{i \in \Ncal} \expect\left[\frac{\ind{P_i \leq \alpha \beta(t(i) / U) / K}}{t(i)}\right] \labelrel{\leq}{eq:superuniform-apply} \sum\limits_{i \in \Ncal} \frac{\alpha}{K} \leq \alpha.
\end{align} We know \eqref{eq:rand-bigger} is true by \eqref{eq:rej-eq} and $t(i) \leq k^*_{\beta, \U}$ whenever the $i$th hypothesis is rejected. Inequality \eqref{eq:superuniform-apply} is by \Cref{lemma:rand-superuniform}. Thus, we have FDR control for the reshaping function $\beta$.

Now, we will prove the strict power improvement portion.
$\Dcal_{\beta, \U} \supset \Dcal_{\beta}$ iff there exists some $r > k^*_\beta + 1$ such that there is a $P_i$ where $i \not\in \Dcal_{\beta}$ and $\alpha\beta(r) / K \geq P_i$.
This is because there will be some critical threshold $u \geq (k^*_{\beta} + 1) / r$ where, if $U \leq u$, we will have a set $\Dcal \supset \Dcal_{\beta}$ where $P_i \leq \alpha \beta(|\Dcal| / U) / (|\Dcal| \vee 1)$ for each $i \in \Dcal$ --- such a threshold exists since we can let such $\Dcal = \Dcal_{\beta} \cup \{i\}$ if $u = (k^*_\beta + 1) / r$.
Since $\beta$ is a monotonically increasing function, there exists a $r$ iff $P_i \leq \max_{r > k^*_\beta + 1}\ \alpha\beta(r) / K$. This, combined with the fact that $i \not\in \Dcal_\beta$ implies that $\Dcal_{\beta, \U} > \Dcal_\beta$ iff there exists $i \in [K]$ s.t.\ $\alpha \beta(k^*_\beta + 1) / K < P_i \leq \max_{r > k^*_\beta  + 1}\ \alpha \beta(r) / K$. Hence, this implies our desired result.

\subsection{Proof of \Cref{thm:closed-uhommel}}
\label{sec:uhommel-proof}
First, we prove the equivalence between $\Dcal_{\textnormal{closed}}$ defined in \eqref{eq:closed-discovery-set} produced by the closed testing procedure for $(\varphi^I_\UHommel)_{I \subseteq [K]}$ and $\Dcal_{\UHommel}$ by following the proof of Lemma 1 in \citet{meijer_hommels_procedure_2019}.

Assume that $ i \in \Dcal_{\UHommel}$ and let $I \subseteq [K]$ be such a set that satisfies $i \in I$.

\begin{itemize}
    \item If $|I| > h_{\U}(\alpha)$, then there exists $j \leq |I|$ s.t.\ $P_{(k - |I| + j)} \leq \alpha (\lfloor j / U \rfloor \wedge |I|) / (|I|\ell_{|I|})$. Since
        \begin{align}
        P_{(j, I)} \leq \max_{I' \subseteq [K]: |I| = |I'|}\ P_{(j, I')} = P_{(k - |I| + j)},
    \end{align} we get that
    \begin{align}
    P_{\UHommel}^I
    \leq  P_{(k - |I| + j)}|I|\ell_{|I|}(\lfloor j / U \rfloor \wedge |I|)^{-1}
    \leq \alpha
    \end{align}
and $\varphi_\UHommel^I = 1$.

\item Otherwise, if $|I| \leq h_\U(\alpha)$, we can simply derive that
$P_\UHommel^I
\leq P_i |I|\ell_{|I|}(\lfloor U^{-1} \wedge |I| \rfloor)^{-1}
\leq \alpha.$
The first inequality is by the definition of $P_{\UHommel}^I$ and the last inequality is by the definition of $\Dcal_{\UHommel}$ and $|I| \leq h_{\U}(\alpha)$.
\end{itemize}

Now, we assume that $i \in \Dcal_{\textnormal{closed}}$. We can define $\Dcal^C \coloneqq [K] \setminus \Dcal_\UHommel$ and note it contains exactly the hypotheses with the $h_{\U}(\alpha)$ largest p-values by the definition of $h_\U(\alpha)$.
This means that $P_i < P_{(K - h_\U(\alpha) + 1)}$.
Now, let $I$ be set which contains the hypotheses with the $h_\U(\alpha) - 1$ largest p-values and $i$ and hence $|I| = h_\U(\alpha)$.

For $j \in \{2, \dots, h_\U(\alpha)\}$, we know that $P_{(j, I)} = P_{(K - |I| + j)}$, which means that $P_{(j, I)} |I|\ell_{|I|}(\lfloor j / U\rfloor \wedge |I|)^{-1} > \alpha$.
Since $i \in \Dcal_{\textnormal{closed}}$ we know that $P^I_{\UHommel} \leq \alpha$. Thus, we get that
\begin{align}
P^I_{\UHommel}
&= P_{(1, I)} |I|\ell_{|I|}(\lfloor U^{-1}\rfloor \wedge |I|)^{-1}\\
&= P_i |I|\ell_{|I|}(\lfloor U^{-1}\rfloor \wedge |I|)^{-1}\\
&=P_i h_\U(\alpha)\ell_{h_\U(\alpha)}(\lfloor U^{-1}\rfloor \wedge h_\U(\alpha))^{-1}
\leq \alpha.
\end{align}
Thus, we have shown our desired result that $\Dcal_{\textnormal{closed}} = \Dcal_{\UHommel}$.

$(\varphi^I_\UHommel)_{I \subseteq [K]}$ are valid $\alpha$-level tests because \Cref{thm:uhommel} implies $P^I_{\UHommel}$ is a p-value for $H_I$. Thus, \Cref{fact:fwer-closed} implies that the closed testing procedure using $(\varphi^I_\UHommel)_{I \subseteq [K]}$ ensures $\FWER \leq \alpha$.

Lastly, we know that $\Dcal_{\UHommel} \supseteq \Dcal_{\Hommel}$ almost surely because $\varphi^I_{\alpha\textnormal{-}\Hommel} \leq \varphi^I_{\alpha\textnormal{-}\UHommel}$ for all $I \subseteq [K]$ almost surely. Thus, we have shown all of our desired results.

\section{\added{
    The \JEBH\ procedure with i.i.d.\ uniform randomization
}}
\label{sec:jebh-short-proof}

\added{
We now introduce the \JEBH\ procedure, which allows one to use independent uniform random variables to improve the power of e-BH while still ensuring FDR control via stochastic rounding.
}

\added{
Let $(U_1, \dots, U_K)$ be i.i.d. uniform random variables, independent of $(X_1, \dots, X_K)$. Define
\begin{align}
   k^*_{\J} \coloneqq \max\left\{k \in [K]:\sum_{i \in [K]} \ind{U_i / X_i \leq \alpha_k} \geq k \right\}.
\end{align}
to be the number of discoveries made by the \JEBH\ procedure. Then, let $\Dcal_{\J} \coloneqq \{i \in [K]: U_i / X_i \leq \alpha_{k^*_\J}\}$ be the corresponding rejection set. Thus, we can see that \JEBH\ procedure can be viewed as the application of the BH procedure at level $\alpha$ applied to $(U_1 / X_1, \dots, U_K / X_K)$. We will now show that this procedure controls the FDR at level $\alpha$ since it results from e-BH applied to stochastically rounded e-values.
}

\added{
First, we note that we can define the following pseudo discovery counts for each $i \in [K]$ as follows:
\begin{align}
   k_{\J, -i} \coloneqq \max\left\{k \in [K]:1 + \sum_{j \in [K], j \neq i} \ind{U_i / X_i \leq \alpha_k} \geq k \right\}.
\end{align} This is equivalent to the number of discoveries that would be made by BH applied to the same p-values except with the $i$th p-value replaced by a 0, i.e., $(U_1 / X_1, \dots, U_{i-1} / X_{i-1}, 0,  U_{i+1} / X_{i+1}, \dots, U_K / X_K)$.
}

\added{
Now, we can show FDR control for \JEBH.
\begin{theorem}
    The J-eBH procedure applied at level $\alpha$ controls $\FDR \leq \alpha$ under arbitrary dependence between $(X_1, \dots, X_K)$.
\end{theorem}
}
\begin{proof}
    \added{
    Let $\Dcal$ be the rejection set of produced by applying the e-BH procedure at level $\alpha$ to the rounded e-values $(S_{\alpha_{k_{\J, -1}}}(X_1), \dots, S_{\alpha_{k_{\J, -K}}}(X_K))$ where we tie the stochastically rounded e-values to $(U_1, \dots, U_K)$ via the formulation
    \begin{align}
        S_{\alpha_{k_{\J, -i}}}(X_i) \coloneqq \frac{\ind{U_i \leq \alpha_{k_{\J, -i}} X_i}}{\alpha_{k_{\J, -i}}}.
    \end{align}
    }
    
    \added{
    Immediately, we notice this has valid FDR control via \Cref{prop:dynamic-round-e} and FDR control of the e-BH procedure since $\alpha_{\J, -i}$ is independent of $U_i$ for each $i \in [K]$ by construction, and hence $(S_{\alpha_{k_{\J, -1}}}(X_1), \dots, S_{\alpha_{k_{\J, -K}}}(X_K))$ are valid e-values. Now, it only remains to show that $\Dcal = \Dcal_{\J}$.
    }

    \added{
    For each $i \in \Dcal_\J$ we know that $U_i / X_i \leq \alpha_{k^*_\J} = \alpha_{k_{\J, -i}}$. 
    This implies that $S_{\alpha_{k_{\J, -i}}}(X_i) = 1 / \alpha_{k^*_\J}$, so $i \in \Dcal_\J$. 
    Conversely, for each $i \not\in \Dcal_\J$ we know that $U_i / X_i > \alpha_{k^*_\J} \geq \alpha_{k_{\J, -i}}$ where the last inequality follows from the fact that the number of discoveries made by BH is coordinatewise decreasing in each p-value.
    Thus, we have $S_{\alpha_{k_{\J, -i}}}(X_i) = 0$, so $i \not\in \Dcal_\J$.
    }

    \added{
    Since exactly $k^*_\J$ hypotheses are rejected by $\Dcal_\J$, we have that $\Dcal = \Dcal_\J$ via the definition of e-BH. Thus, we have shown that \JEBH\ controls the FDR at level $\alpha$.
    }
	\end{proof}

\added{
\boldparagraph{Simulations} In the same setting as \Cref{sec:Simulations}, we compare \JEBH\ and \UEBH. The difference in their power is shown in \Cref{fig:uebh-jebh-heatmap}. Interestingly, neither method dominates the other. \UEBH\ is more powerful when there is negative dependence or $\mu$ is small and positive dependence. On the other hand, \JEBH\ is more powerful when there is positive dependence and $\mu$ is large. This suggests one avenue for future work is to identify whether there is a single optimal joint generalized rounding scheme, and if not, in what situations are different schemes optimal.
}

\added{
\boldparagraph{Stability vs. power} We also investigate the stability of our randomized procedures compared to their power. We measure stability through the Jaccard similarity index, i.e., 
\begin{align}
    \text{JC}(\Dcal) = {n\choose 2}^{-1}\sum_{i, j \in [n], i < j }\frac{|\Dcal_i \cap \Dcal_j|}{|\Dcal_i \cup \Dcal_j|}
\end{align}
where $\Dcal_i$ is an independent run of the same procedure on the same data, and $n$ is the number of trials if we re-run the randomized procedure. We let $n = 50$. A higher JC value indicates higher stability. In \Cref{fig:stability}, we see that \UEBH\ is always more stable than \JEBH\ since it uses less randomization, in addition to being the more powerful of the two methods.
}
\begin{figure}[h!]
    \centering
    \begin{subfigure}[t]{\textwidth}
    \centering
\includegraphics[width=\textwidth]{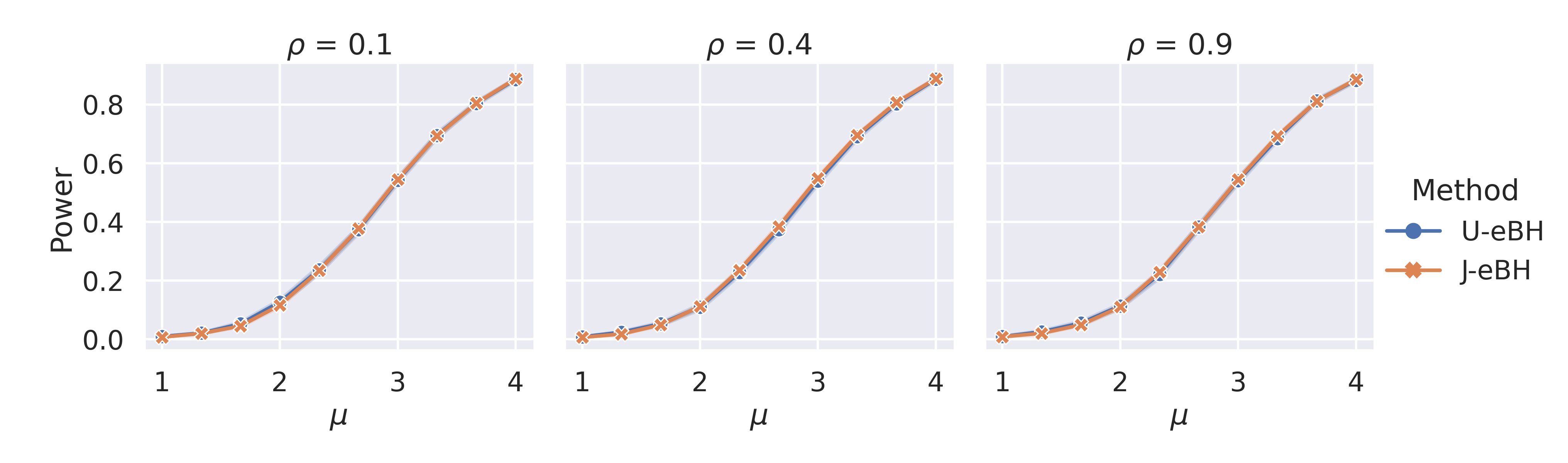}
    \caption{Positive dependence: \UEBH\ and \JEBH\ have simlar power.}
    \end{subfigure}
    \begin{subfigure}[t]{\textwidth}
    \centering
\includegraphics[width=\textwidth]{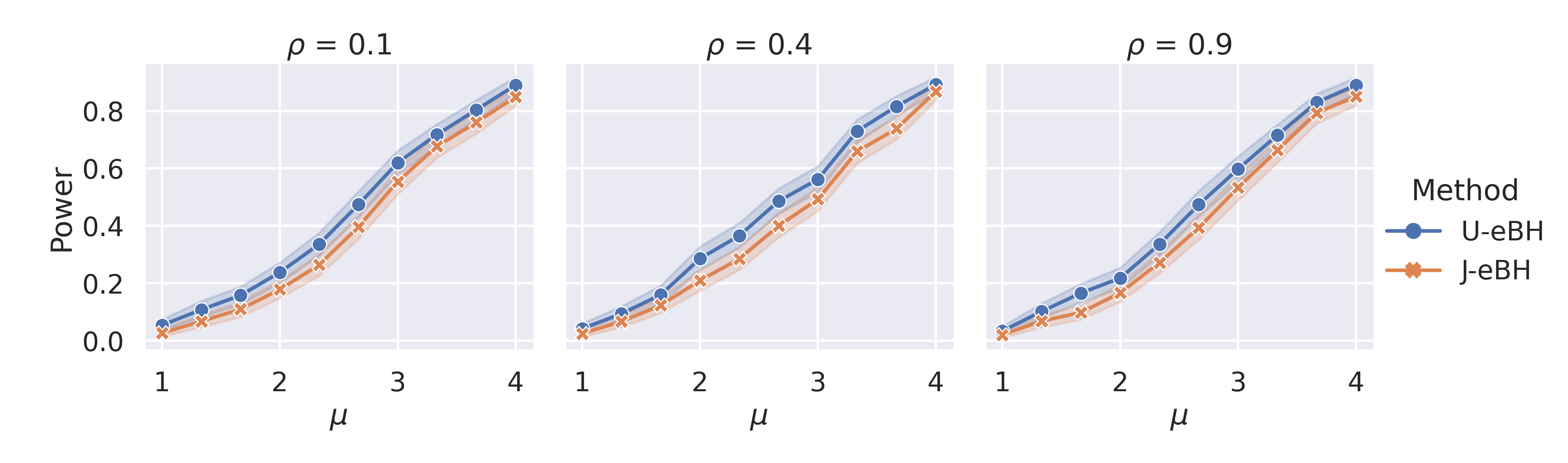}
    \caption{Negative dependence: \UEBH\ dominates \JEBH\ across the board.}
    \end{subfigure}
    \caption{Plots of of power using \UEBH\ vs \JEBH, i.e., single $U$ vs. a independent $U_i$, across different values of $\mu$ and $\rho$ in the same simulation setting as described in \Cref{sec:Simulations}. \UEBH\ is more powerful when the dependence is negative, and they have similar power when the dependence is positive.}
    \label{fig:uebh-jebh-heatmap}
\end{figure}

\begin{figure}[ht]
    \centering
    \begin{subfigure}{\textwidth}
    \centering
    \includegraphics[width=\textwidth]{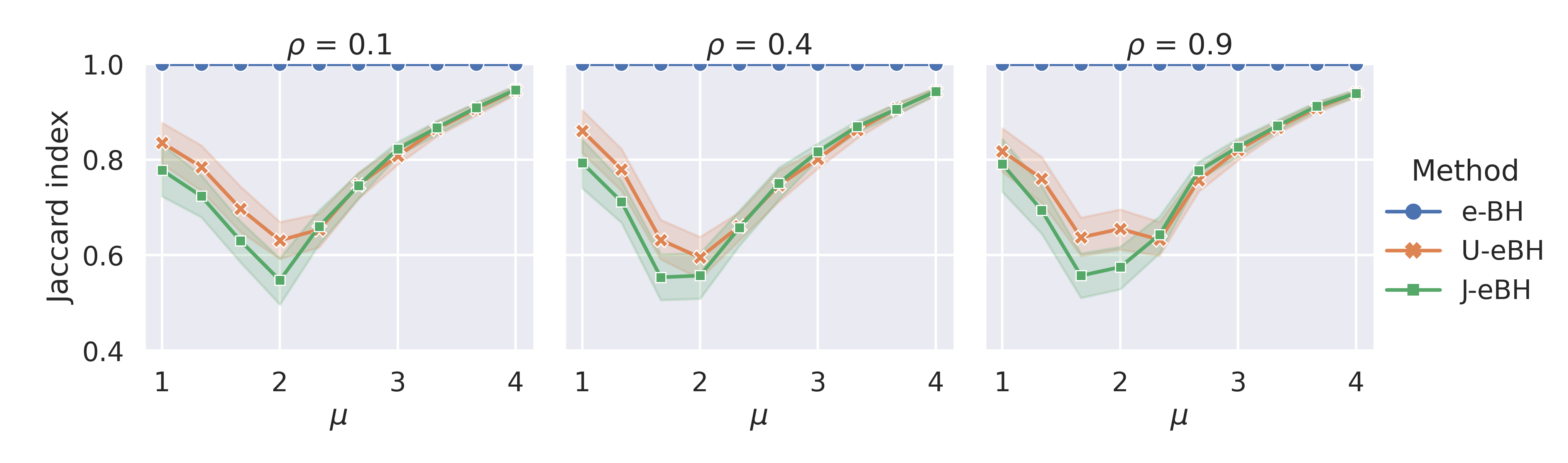}
    \caption{Positive dependence}
    \end{subfigure}
    
    \begin{subfigure}{\textwidth}
    \centering
    \includegraphics[width=\textwidth]{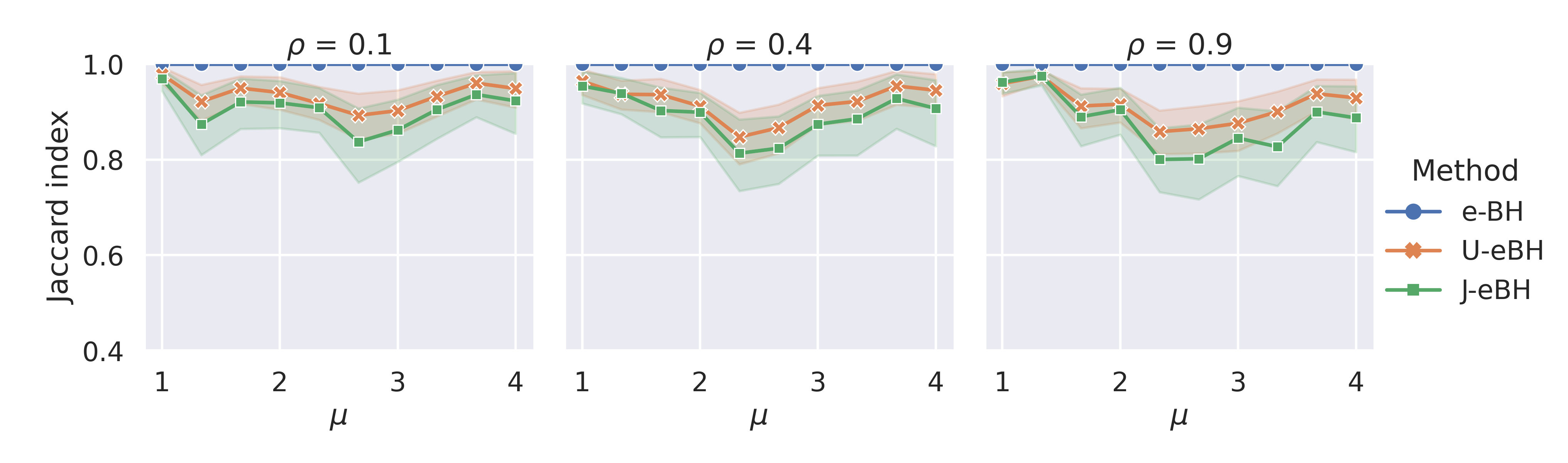}
    \caption{Negative dependence}
\end{subfigure}    
    \caption{\added{Plot of stability eBH, \UEBH, and \JEBH. We plot the stability of a procedure in terms of the Jaccard similarity index on the x-axis and power on the y-axis for different positive and negative dependence settings and values of $\mu$. We see that \UEBH\ is always more stable since it uses less randomization, and empirically seems to achieve more power.}\label{fig:stability}}
\end{figure}

\section{Generalization of stochastic rounding}
\label{sec:general-stoch-round}

One might wish to generalize stochastic rounding (to the nearest points in $\Gcal$) to instead take a probability distribution over values in $\Gcal$ that are greater than or equal to $x^+$. For e-BH, the simplest such generalization would be to assign a uniform probability to all potential e-BH rejection levels greater than the input $x$.  Let $r = x - x_-$, and define $j^+ = K / (\alpha x^+)$ . Then, we would define $S^{\mathrm{uni}}_\Gcal(X)$ a follows:
\begin{align}
    S^{\mathrm{uni}}_\Gcal(x) = \begin{cases}
        \frac{K}{\alpha j} & \text{ with probability }\frac{\alpha r}{K H_{j^+}}\\
        x_- & \text{ with probability }1 - \frac{\alpha r j^+}{K H_{j^+}}.
    \end{cases}
\end{align}
Another way of generalizing the stochastic rounding is to have equal ``expectation'' assigned to each potential e-BH rejection level. We would have $p_j \propto j$, where $p_j$ is the probability of outputting $K / (\alpha j)$ for some $j \geq j^+$.
\begin{align}
    S^{\mathrm{eq}}_\Gcal(x) = \begin{cases}
        \frac{K}{\alpha j} & \text{ with probability }\frac{\alpha j r}{K j^+}\text{ for }j \in [j^+]\\
        x_- & \text{ with probability }1 - \frac{\alpha r(j^+ + 1)}{2K}.
    \end{cases}
\end{align}

More generally, we know that we can consider $\Gcal^+(x) \coloneqq \{ y \in \Gcal: y \geq x\}$, i.e., the set of values in the grid greater than $x$. Let $F_x$ be the c.d.f. of a distribution over $\Gcal^+(x)$ for each $x$ that satisfies
\begin{align}
    p^+(x) \cdot \int\limits_{y \in \Gcal^+(x)} y\ dF_x(y) + (1 - p^+(x))x_- = x,
    \label{eq:constrained-F}
\end{align} where $p^+(x) \in [0, 1]$ is the probability $X$ is made larger by stochastic rounding.
Let $S^+_{\Gcal}(X)$ be our generalized stochastically rounded e-value i.e.,
\begin{align}
    S^+_{\Gcal}(x)\coloneqq \begin{cases}
        Y & \text{ where }Y \sim F_x\text{ ind. of }X\text{ with probability }p^+(x)\\
        x_- & \text{ with probability }(1 - p^+(x)).
    \end{cases}
\end{align}
\begin{proposition}
    If $X$ is an e-value, then $\expect[S^+_{\Gcal}(X)]$ is an e-value. Further, $\expect[X] = \expect[S^+_{\Gcal}(X)]$.
\end{proposition}
The above proposition follows from the definition of $S^+_\Gcal$, and \eqref{eq:constrained-F}. We also note that in the multiple testing setting, we can allow the $Y$ that is sampled to be dependent on all the e-values $(X_1, \dots, X_K)$ while still having a conditional distribution with expectation $x$. This allows us to adaptively choose the rounding distribution $F_x$ for all $x \in [0, \infty)$ based on the randomized e-values results so far if randomization is done sequentially over the e-values.

 We can extend this generalized rounding idea to allow for joint distribution across the rounded e-values. In essence, our rounding scheme would be, given some input e-values $(X_1, \dots, X_K)$, the rounded e-values $(S_{\Gcal_1}^+(X_1), \dots, S_{\Gcal_K}^+(X_K))$ satisfy the following:
\begin{align}
    \expect[S_{\Gcal_i}^+(X_i) \mid X_i] \leq X_i.
    \label{eq:cond-exp}
\end{align}

\boldparagraph{Joint generalized rounding} We define a single grid for all e-values: $\Gcal_i = \Kcal_{k^*} \coloneqq \{1 / \alpha_j: j > k^*, j \in [K]\} $ for each $i \in [K]$. Since $1 / \widehat\alpha^*$ is larger than the max element of $\Kcal_{k^*}$, all $X_i$ that are rejected by e-BH do not change after rounding. Now, we will define a function $D$ that maps $(X_1, \dots, X_K)$ to a distribution over $[K]\setminus \Dcal_{\rmEBH}$, which we call the \emph{rounding distribution}. The rounding distribution $D$ will be restricted to ensure each $S_{\Kcal_{k^*}}^+(X_i)$ satisfies \eqref{eq:cond-exp}, and remains a valid e-value.

 Let $\Dcal^{\mathrm{rd}}$ be drawn from $D(X_1, \dots, X_K)$. Then,  define each rounded e-value $(i \in [K])$ as follows:
\begin{align}
    S_{\Kcal_{k^*}}^+(X_i) = (X_i \cdot \ind{X_i \geq 1 / \widehat\alpha^*}) \vee \frac{\ind{i \in \Dcal^{\mathrm{rd}}}}{\alpha_{|\Dcal^{\mathrm{rd}}|}}.
    \label{eq:j-dist}
\end{align}

\added{
 To find the sampling distribution that maximizes the set of discoveries and satisfies these constraints, we can define the following optimization problem,  where for each $S \subseteq [K]\setminus \Dcal_{\rmEBH}$, $\mathrm{p}_{S}$ denotes the probability that $S \cup \Dcal_{\rmEBH}$ is rejected:
\begin{align}
    \max\limits_{\mathrm{p}_S: S \subseteq [K]\setminus \Dcal_{\rmEBH}} &\sum\limits_{S \subseteq [K]\setminus \Dcal_{\rmEBH}} |S| \cdot \mathrm{p}_S\label{eq: general round program}\\
    \text{s.t. }&\sum\limits_{S \subseteq [K]\setminus \Dcal_{\rmEBH}} \mathrm{p}_S = 1,\\
    &\sum\limits_{S \subseteq [K] \setminus \Dcal_\rmEBH:\ i \in S} \frac{\mathrm{p}_S}{\alpha_{|S \cup \Dcal_\rmEBH|}X_i}  \leq 1 \text{ for each }i \in [K]\setminus \Dcal_{\rmEBH},\\
    &\mathrm{p}_S \in [0, 1] \text{ for each } S \subseteq [K]\setminus \Dcal_{\rmEBH}.
\end{align}
}

\added{
The objective function maximizes the expected number of discoveries made after rounding. The first and third constraints ensure that the distribution over $\Dcal^\textnormal{rd}$ is a valid probability distribution, and the second constraint ensures that the conditional expectation of each rounded e-value will be at most the original e-value.
Thus, any assignment that satisfies the above constraints will result in a rounding distribution that produces valid e-values after rounding via \eqref{eq:j-dist}. 
\begin{proposition}
    If $D(X_1, \dots, X_K)$ is a rounding distribution whose probabilities are characterized by $\prob{\Dcal^{\mathrm{rd}} = S \cup \Dcal_\rmEBH} = \mathrm{p}_S$ for each $S \subseteq [K]\setminus \Dcal_{\rmEBH}$ and $(\mathrm{p}_S)_{S \subseteq [K] \setminus \Dcal_\rmEBH}$ satisfy the constraints of \eqref{eq: general round program}, then $(S_{\Kcal_{k^*}}^+(X_1), \dots, S_{\Kcal_{k^*}}^+(X_K))$ are e-values.
\end{proposition}
This proposition follows from the fact that the constraints ensure that $\expect[S_{\Kcal_{k^*}}^+(X_i) \mid X_i] \leq X_i$ for each $i \in [K]$.
}
\added{
Thus, this framework enables flexibility in the type of dependence that may be used to randomize the e-BH procedure. 
We leave the study of efficient algorithms to solve the above optimization problem to future work.
}

\section{\added{Existing multiple testing procedures as stochastic rounding}}
\label{sec:weighting}
\added{
The randomization techniques that we introduced also have interesting connections to existing multiple testing methods that use randomization. In fact, we generalize the following three methods that have appeared in prior literature.}
\begin{enumerate}
    \item 
    \added{
    The results in \citet{ignatiadis_e-values_unnormalized_2022} and our work are complementary. We discuss in-depth connections of our method with theirs in the sequel, along with a new method for multiple testing in the setting of having access to both e-values and p-values for each hypothesis. 
    }
\item
    \added{
     The randomized pruning techniques in \citet{fithian_conditional_calibration_2022} can be interpreted as a variant of our J-eBH procedure, but are specific to improving the power of the dBH procedure. They do not consider the general applicability of randomization to other multiple testing methods, nor do they study its connection to e-values. Further, they use randomized pruning to repair the mismatch in thresholds that results from using an estimator of the total number of discoveries made by dBH based on the sufficient statistic, versus the actual total number of discoveries. In contrast, the results in this paper show that randomization can be used to strictly improve the power of several multiple testing procedures, and will never fail to reject existing discoveries made by e-BH. 
    }
    \item
    \added{
     Following our work, \citet{jin_model-free_selective_2023} also use a randomized pruning in their methods that is similar in spirit to the methods of \citet{fithian_conditional_calibration_2022}, but applied to the weighted conformal p-values that they develop. Their procedures are analogs of our \JEBH\ and \UEBH\ procedures. Again, their work does not consider the general applicability of randomization to various multiple testing methods nor the connection to e-values. Further, their focus is similar to \citet{fithian_conditional_calibration_2022}, i.e., they are using the randomized pruning to also repair the mismatch in thresholds that results from using an estimator of the total number of discoveries made by WCS based on the calibration set and a single test point.
    }
\end{enumerate}

\added{
\subsection{E-value weighted multiple testing procedures of \citet{ignatiadis_e-values_unnormalized_2022}}
\cite{ignatiadis_e-values_unnormalized_2022} investigate the setting where an e-value and a p-value are available for each hypothesis. 
}
They mentioned, in passing, how to convert an e-value to a randomized p-value \citep[Remark 3.2]{ignatiadis_e-values_unnormalized_2022}, whereas the current work shows how to convert an e-value into a randomized e-value (through stochastic rounding) and explores its implications for multiple testing. The robustness of e-values to dependencies is what allows us to achieve FDR control under dependence. We now discuss what implications and relations our techniques have for this aforementioned setting.

\boldparagraph{E-value weights for p-values} As we saw in \Cref{prop:uebh-as-bh}, an intriguing formulation of \UEBH\ is the application of the BH procedure to $(U / X_1, \dots, U / X_K)$.  Recall that the BH procedure operates on p-values $(P_1, \dots, P_K)$ and make the following number of discoveries:
\begin{align}
    k^*_{\BH} \coloneqq \max\left\{i \in [K]: P_{(i)} \leq \frac{ \alpha i}{K}\right\}.
\end{align}  The BH procedures then rejects the hypotheses corresponding to the $k^*_{\BH}$ smallest p-values.
\cite{ignatiadis_e-values_unnormalized_2022} and \cite{ramdas_randomized_exchangeable_2023} both note that $U / X$ is a p-value. \added{Further, \citet{ignatiadis_e-values_unnormalized_2022} noted that this allowed the use of their ep-BH procedure with only e-values (rather than requiring both a e-value and a p-value for each hypothesis), since one can always use a uniform random variable in place of the p-value --- this is precisely the \UEBH\ procedure. Our \JEBH\ procedure also recovers the ep-BH procedure with independent uniform random variables. We note that their proof of FDR control also only requires a PRDS assumption on the uniform random variables $(U_1, \dots, U_K)$. On the other hand, our stochastic rounding framework allows for an altogether different framework for randomness dependence. In \Cref{sec:general-stoch-round}, we allow a class of dependence structures between the external randomness used for each hypothesis that includes types of dependence that are not PRDS. To consider a simple situation where stochastic rounding allows for dependence structures beyond PRDS, note that in \RtwoEBH, one can choose to use $U_j = 1 - U_i$ for two hypotheses where $i \neq j$. The exact relation between these two frameworks is an interesting avenue for future work.}

A key conceptual difference between the methods in \cite{ignatiadis_e-values_unnormalized_2022} is that the randomization for our methods originates from the novel notion of a stochastically rounded e-value (and its variants). \citet{ignatiadis_e-values_unnormalized_2022} solely consider multiple testing when both p-values and e-values were available for a hypothesis, which is different from our (more traditional) setup. This lends greater flexibility to our methods since they can be directly plugged in whenever e-values are available. Consequently, we can prove FDR control for the \RBY\ procedure for arbitrarily dependent p-values, and derive a randomized superuniformity lemma (\Cref{lemma:rand-superuniform}) which ensures FDR control for any step-up multiple testing procedure that uses a reshaping function. In \Cref{sec:posi}, we use our framework to derive randomized versions of post-selection inference procedures for false coverage rate (FCR) control as well. Subsequent work in online multiple testing \citep{xu_online_multiple_2023} can also directly leverage the stochastic rounding framework to derive more powerful procedures with online FCR and FDR control. Hence, the framework for stochastic rounding that we develop in this paper is quite flexible and enables randomization to be generalized to many settings.
\boldparagraph{Combining a p-value and an e-value} Our randomization techniques also imply a new method for multiple testing when both p-values and e-values are at hand. We observe that we can replace the uniform random variable in \RtwoEBH\ with a p-value, and derive a method for combining a p-value $P$ and any e-value $X$ into a single e-value:
\begin{align}
    M_{\widehat\alpha}(X, P) \coloneqq (X \cdot\ind{X \geq \widehat\alpha^{-1}}) \vee (\alphath^{-1} \cdot \ind{\widehat\alpha X \geq P}).
\end{align}
\begin{proposition}
    Let $X$ be an e-value w.r.t.\ to null hypothesis $H_0$, and $\widehat\alpha \in (0, 1)$ be a (random) test level.
    If $P$ is a p-value when conditioned on $X$ and $\widehat\alpha$, i.e., $\prob{P \leq s \mid X, \widehat\alpha} \leq s$ for all $s \in [0, 1]$ when $H_0$ is true, then $M_{\widehat\alpha}(X, P)$ is an e-value under $H_0$ as well.
\end{proposition}
The proof follows from the same argument as \Cref{prop:dynamic-round-e}.
Recall that $1 / \widehat\alpha^*$ is the rejection threshold for the e-BH procedure. Thus, we can define the \underline{\emph{\PEBH\ procedure}}, as rejecting the $i$th hypothesis when $M_{{\widehat\alpha}^*}(X_i, P_i) \geq 1 / \widehat\alpha^*$.

\begin{theorem}
Let $\mathbf{P} \coloneqq (P_1, \dots, P_K)$ and $\mathbf{X} \coloneqq (X_1, \dots, X_K)$ be vectors of p-values and e-values, respectively, where $P_i$ and $X_i$ are a p-value and e-value for the same null hypothesis for each $i \in [K]$. Assume $\mathbf{P}$ is independent of $\mathbf{X}$, but otherwise $\mathbf{P}$ and $\mathbf{X}$ can each have arbitrarily dependent components. Then, the \PEBH\ procedure ensures $\FDR \leq \alpha$ and rejects a superset of the hypotheses rejected by e-BH applied to $\mathbf{X}$.
\end{theorem}

 \cite{ignatiadis_e-values_unnormalized_2022} discuss procedures for FDR control under arbitrarily dependent p-values where one first uses a calibrator, $h$, to produce the weighted e-value $h(P)E$, on which the e-BH procedure is then applied. However, there is no deterministic calibrator $h$ s.t. $h(p) \geq 1$ for all values of $p \in [0, 1]$, except the trivial calibrator $h(p) = 1$. Thus, these calibrator methods are unlike \PEBH, as they cannot \emph{guarantee} improvement over e-BH. \added{Notably, \citet{ignatiadis_e-values_unnormalized_2022} noted that the ep-BH procedure is most powerful when the e-values are constructed from a smaller, initial dataset (which is usually obtained first) before collecting and constructing a p-value from a larger follow-up dataset. In this setting, the \PEBH\ procedure ensures that all discoveries that are made solely by using the e-values and eBH will \emph{always} be preserved by the \PEBH\ procedure, and the introduction of p-values will only increase the number of discoveries.}

\section{\RBY\ when the FDR control is sharp}
\label{sec:tight-fdr}

One might wonder what the behavior of \RBY\ is when applied to the p-value distribution constructed by \cite{guo_control_false_2008}, on which BY applied at level $\alpha$ has an FDR of exactly $\alpha$. We will show that in this situation \RBY\ has a behavior that is identical to BY, and also has an FDR of exactly $\alpha$.

Let $K_0 \in [K]$ be the number of null hypotheses --- let the first $K_0$ hypotheses be true nulls (i.e., $\Ncal = [K_0]$ is the set of true nulls), and the remaining hypotheses be non-nulls (i.e, $\{K_0 + 1, \dots, K\} = [K] \setminus [K_0]$). For a specific level $\alpha \in [0, 1]$, \cite{guo_control_false_2008} begin their construction by sampling $N \in [K + 1]$ according to the following distribution:
    \begin{align}
        \prob{N = n} =
        \begin{cases}
            \alpha K_0 /(nK\ell_K) & \text{ for }n \in [K_0]\\
            \alpha / (K\ell_K) & \text{ for }n \in \{K_0 + 1, \dots,  K\}\\
            1 - \alpha\left(1 + (\ell_{K_0} - 1)K_0\right) / (K\ell_K) & \text{ for }n = K + 1
        \end{cases}.
    \end{align}

Let $U_0$ and $U_1$ be independent, uniform random variables in $[0, 1]$.
\begin{enumerate}
\item If $N \leq K$, draw $I_0 \subseteq [K_0]$ uniformly randomly from the subsets of $[K_0]$ that are of size $N \wedge K_0$. If $N > K_0$, additionally draw $I_1 \subseteq [K]\setminus [K_0]$ uniformly randomly from subsets of $[K]\setminus [K_0]$ that are of size $N - K_0$ --- otherwise let $I_1 = \emptyset$. Let $I \coloneqq I_0 \cup I_1$. Now, we define our p-values as follows:
\begin{align}
    P_i =
    \begin{cases}
    \alpha(N - 1 + U_0) / (K\ell_K) & \text{ for }i \in I\\
    \alpha / \ell_K + (1 - \alpha / \ell_K)U_1 & \text{ for }i \not\in I
    \end{cases} \label{eq:guo-rao-good-p}.
\end{align}
\item Otherwise, if $N = K + 1$, the p-values are generated as follows:
\begin{align}
    P_i = \begin{cases}
        \alpha / \ell_K + (1 - \alpha / \ell_K)U_0 \text{ for }i \in [K_0]\\
        1 \text{ for }i \in [K] \setminus [K_0]
    \end{cases}.
    \label{eq:guo-rao-bad-p}
\end{align}
\end{enumerate}
\begin{fact}[Theorem 5.1 (iii) of \citet{guo_control_false_2008}]
    The BY procedure, when applied to the aforementioned p-value distribution characterized by \eqref{eq:guo-rao-good-p} and \eqref{eq:guo-rao-bad-p}, has an FDR of exactly $\alpha(K + K_0(\ell_{K_0} - 1)) / (K\ell_K)$. Note that, when $K_0 = K$, the FDR of BY is exactly $\alpha$.
    \label{fact:guo-rao-fdr}
\end{fact}
We note that in this setting, BY always rejects all hypotheses corresponding to p-values less than $\alpha / \ell_K$. Hence, \RBY\ cannot improve upon BY, as \RBY\ never rejects any hypotheses with p-values above $\alpha / \ell_K$. Thus, BY and \RBY\ should have the same FDR in this setting. We verify through simulations in \Cref{fig:guo-rao-exp}, where we set $\alpha=0.05$, $K = 10^5$, vary $\pi_0 \coloneqq K_0 / K$ in $[0, 1]$, and plot the FDR calculated from simulations against the FDR derived in \Cref{fact:guo-rao-fdr}. Clearly, BY and \RBY\ have identical empirical FDR that matches closely to what \Cref{fact:guo-rao-fdr} indicates.

\begin{figure}[h]
    \centering
\includegraphics[width=0.9\textwidth]{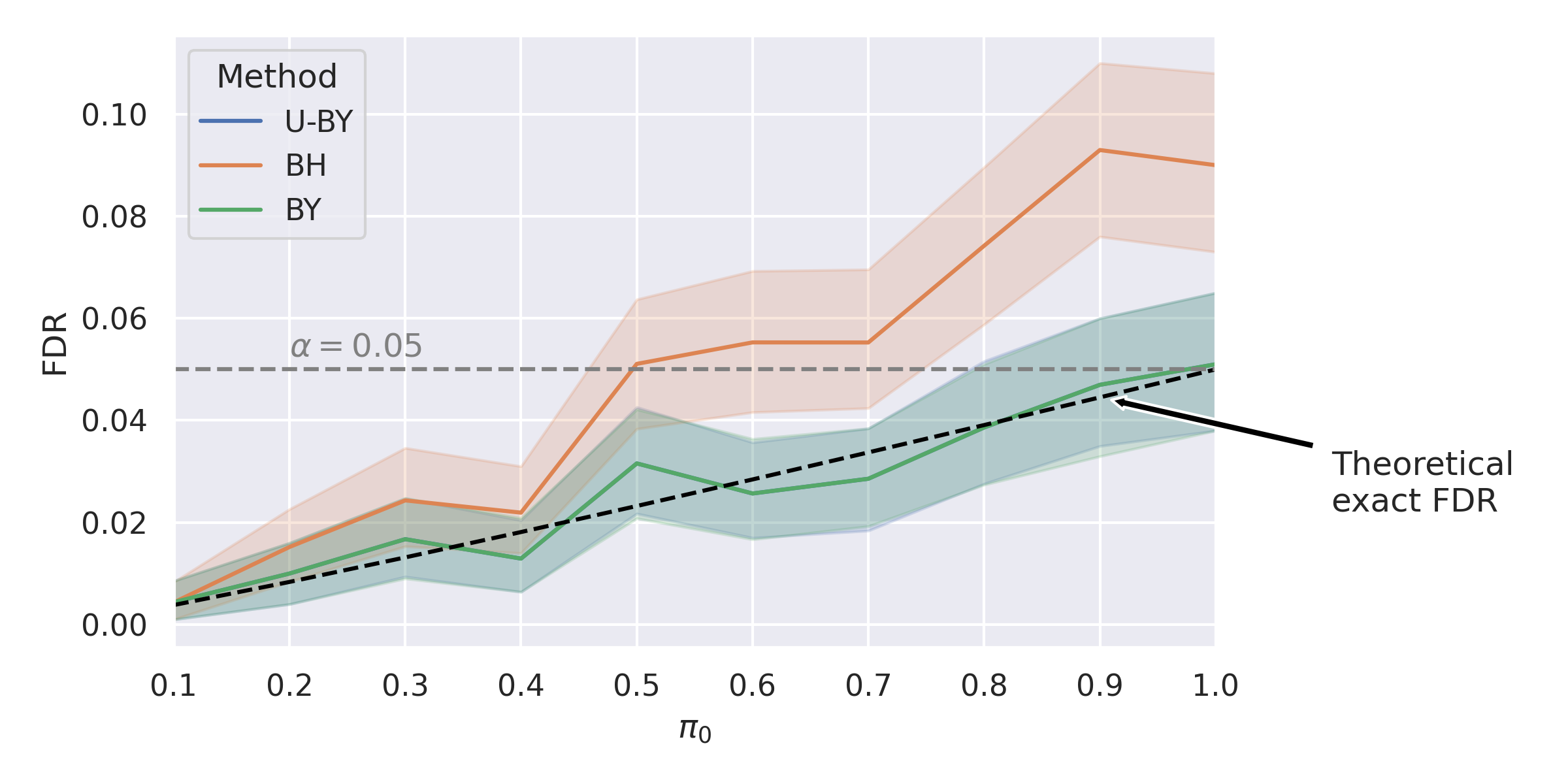}
\caption{Plot of proportion of nulls, $\pi_0$, vs.\ the FDR (estimated through simulations) of different FDR controlling procedures applied at level $\alpha=0.05$ (denoted by the dashed grey line) in the setting described in \citet{guo_control_false_2008}. The bands surrounding each line are 95\% CIs for the FDR that arise from Monte Carlo error. The dashed black line indicates the exact FDR of BY derived theoretically by \cite{guo_control_false_2008}. BH does not ensure  $\FDR \leq \alpha$. \emph{Both BY and \RBY\ have identical estimates and closely match the theoretical FDR.}}
    \label{fig:guo-rao-exp}
\end{figure}

\section{Randomizing admissible \added{e-merging} and p-merging functions}
\label{sec:admissible}

We can extend our results from \Cref{sec:rand-simes} concerning the Hommel p-merging function to randomize admissible p-merging functions in general, as well as randomizing e-merging functions. \added{We can define an \emph{e-merging function} $E: [0, \infty)^K \rightarrow [0, \infty)$ to be a function where $E(\mathbf{X})$ is an e-value if $\mathbf{X} = (X_1, \dots, X_K)$ is a vector of $K$ arbitrarily dependent e-values.} 
\added{Similarly}, we define a \emph{p-merging function} $F: [0, \infty)^K \rightarrow [0, \infty)$ to be a function where $F(\mathbf{P})$ is a p-value if $\mathbf{P} = (P_1, \dots, P_K)$ is a vector of $K$ arbitrarily dependent p-values. Note that p-values are usually restricted to the domain $[0, 1]$, but the definition of p-merging functions relaxes the domain to $[0, \infty)$ for simplicity. We also call $F$ \emph{ homogeneous} if $F(\lambda\mathbf{p}) = \lambda F(\mathbf{p})$ for any $\lambda \in (0, 1]$ and $\mathbf{p} \in [0, \infty)^K$. The class of homogeneous p-merging functions includes the analog of the Bonferroni correction, methods introduced by \citet{hommel_tests_overall_1983} and  \citet{ruger_1978} based on order statistics, and functions that use p-value averages devised by \citet{vovk_combining_p-values_2020}.

\added{
\citet{wang_only_admissible_2025} proved that the only admissible e-merging functions are convex combinations of the input e-values and 1. Let $\mathbf{x} = (x_1, \dots, x_K) \in [0, \infty)^K$ be a vector of realized e-values.
\begin{fact}[Theorem 1 of \citet{wang_only_admissible_2025}]
    Every admissible e-merging function can be formulated as
    \begin{align}
        E(\mathbf{x}) = \lambda_{K + 1} + \sum\limits_{i = 1}^{K} \lambda_i x_i,
    \end{align}
    for some $(\lambda_1, \dots, \lambda_{K + 1}) \in \Delta^{K + 1}$, where $\Delta^{K + 1}$ is the simplex on $K + 1$ dimensions.
\end{fact}
Then, we can define the stochastically rounded e-merging function for a data-dependent threshold $\hat\alpha \in (0, 1)$ as 
\begin{align}
\widetilde{E}_{\hat\alpha}(\mathbf{x}) \coloneqq S_{\hat\alpha}(E(\mathbf{x})).
\end{align}
\begin{proposition}
    $\widetilde{E}_{\hat\alpha}$ is a bona fide e-merging function for any (possibly data-dependent) $\hat\alpha \in (0, 1)$. Further, $\widetilde{E}$ is always more powerful than $E$ in the sense that $\prob{\widetilde{E}_{\hat\alpha}(\mathbf{X}) \geq \hat\alpha^{-1}} \geq \prob{E(\mathbf{X}) \geq \hat\alpha^{-1}}$.
    \label{prop:rand-e-merge}
\end{proposition}
The proof of this follows directly from $E$ being a valid e-merging function and the properties of stochastic rounding preserving e-value validity in \Cref{prop:dynamic-round-e} and having improved power in \Cref{prop:adaptive-round-power}.
}

\citet{vovk_admissible_ways_2022} prove that any admissible homogeneous p-merging function has a dual form that is formulated in terms of calibrators. We will leverage this dual form to apply our randomization techniques and produce randomized p-merging functions that are never greater than the p-values produced by the unaltered p-merging function. The key representation result that we use is the following:
\begin{fact}[\citealt{vovk_admissible_ways_2022}; Theorem 5.1]
    Let $R_{\alpha}(F) \coloneqq \{\mathbf{p} \in [0, \infty)^K: F(\mathbf{p}) \leq \alpha\}$. For any admissible homogeneous p-merging function $F$, there exist $(\lambda_1, \dots, \lambda_K) \in \Delta^K$ ($\Delta^K$ is the simplex on $K$ dimensions) and admissible calibrators $f_1, \dots, f_K$ such that:
    \begin{align}
        X(\mathbf{p}) &\coloneqq \sum\limits_{i = 1}^K \lambda_i f_i(p_i) \text{ and }\\
        R_\alpha(F) &= \left\{\alpha \cdot \mathbf{p} : \mathbf{p}\in [0, \infty)^K,  X(\mathbf{p}) \geq 1\right\} \text{ for each }\alpha \in (0, 1),
    \label{eq:pmerge-dual}
\end{align} where $p_i$ is the $i$th component of $\mathbf{p}$. Conversely, for any $(\lambda_1, \dots, \lambda_K) \in \Delta^K$ and calibrators $(f_1, \dots, f_K)$, \eqref{eq:pmerge-dual} determines a homogeneous p-merging function.
\end{fact}
Thus, we can define a randomized e-merging function as 
\begin{align}
    \tilde{E}
\end{align}

Note that $X(\mathbf{P})$ is an e-value in the definition above. Hence, we can define a randomized p-merging function by stochastically rounding $X(\mathbf{P})$. Let $\widetilde{F}$ be the randomized version of the p-merging function $F$, where $F$ can be represented by \eqref{eq:pmerge-dual}. We can define the rejection region of $\widetilde{F}$ in the following fashion:
\begin{align}
    R_\alpha(\widetilde{F}) &\coloneqq \left\{\alpha \cdot \mathbf{p} \in [0, \infty)^K: S_{1}(X(\mathbf{p})) \geq 1\right\}\\
    &= \left\{\alpha \cdot\mathbf{p} \in [0, \infty)^K: X(\mathbf{p}) \geq U\right\} \text{ for each }\alpha \in (0, 1),
    \label{eq:rand-rej-region}
    \end{align} where $U$ is a uniform random variable.
\begin{theorem}
    $\widetilde{F}$ with the randomized rejection region defined in \eqref{eq:rand-rej-region} is a bona-fide p-merging function. Further $\widetilde{F}(\mathbf{p}) \leq F(\mathbf{p})$ for any $\mathbf{p} \in [0, \infty)^K$.
    \label{thm:rand-rej-pmerge}
\end{theorem}
\begin{proof}
    First, we want to show the validity of this randomized p-merging function by proving the following:
    \begin{align}
        \prob{\mathbf{P} \in R_{\alpha}(\widetilde F)} &= \prob{S_1(X(\mathbf{P} / \alpha)) \geq 1} \leq \alpha.
        \label{eq:valid-rej-region}
    \end{align} To show this, we use the same argument as \cite{vovk_admissible_ways_2022}, and we can make the following derivation:
    \begin{align}
        \expect[X(\mathbf{P} / \alpha)] &= \sum\limits_{i = 1}^K \lambda_i \expect[f_i(P_i / \alpha)] \leq \sum\limits_{i = 1}^K \lambda_i \int\limits_0^1 f_i(x / \alpha)\ dx =\alpha \sum\limits_{i = 1}^K \lambda_i \int\limits_0^{1 / \alpha} f_i(x)\ dx = \alpha.
    \end{align} The first inequality is because each $f_i$ is a calibrator. The last equality is because $f_i(x) = 0$ for $x > 1$ since $f_i$ is a calibrator, and $\sum_{i \in [K]} \lambda_i = 1$ by definition.
    Consequently, $\expect[S_1(X(\mathbf{P} / \alpha))] \leq \alpha$ as a result of stochastic rounding, i.e., \Cref{prop:dynamic-round-e}. By Markov's inequality, we know that \eqref{eq:valid-rej-region} is true.

    We also know the following chain of implications by definition of $F$:
    \begin{align}
        F(\mathbf{p}) \leq \alpha \Leftrightarrow \sum\limits_{i = 1}^K \lambda_i f_i(p_i / \alpha) \geq 1 \Rightarrow \sum\limits_{i = 1}^K \lambda_if_i(p_i / \alpha) \geq U  \Leftrightarrow \widetilde{F}(\mathbf{p}) \leq \alpha.
    \end{align} The implication in the above chain is simply because $1 \geq U $ is always true. Hence, we have shown our desired result.
\end{proof}
    \Cref{thm:rand-rej-pmerge} now allows us to define a large family of randomized p-merging functions that improve over admissible deterministic p-merging functions. An example of such a p-merging function from \cite{vovk_admissible_ways_2022} is the \textit{grid harmonic merging function} which dominates the Hommel function and the domination is strict when $K \geq 4$. It is defined using the formulation in \eqref{eq:pmerge-dual} and the following calibrator:
\begin{align}
    f^{\textnormal{GH}}(x) \coloneqq \frac{K\ind{\ell_K x \leq 1}}{\lceil K\ell_K x\rceil}.
\end{align}
Consequently, the grid harmonic merging function, $F^{\textnormal{GH}}$, has a rejection region that is defined as follows:
\begin{align}
    R_\alpha(F^{\textnormal{GH}}) \coloneqq \left\{\alpha \cdot \mathbf{p} : \sum\limits_{i = 1}^K\frac{K\ind{\ell_K p_i \geq 1}}{\lceil K\ell_K x\rceil} \geq 1\right\}.
\end{align} Applying our randomization approach, we can define $\widetilde{F}^{\textnormal{GH}}$ as follows:
\begin{align}
    R_\alpha(\widetilde{F}^{\textnormal{GH}}) \coloneqq \left\{\alpha \cdot \mathbf{p} : \sum\limits_{i = 1}^K\frac{K\ind{\ell_K p_i \leq 1}}{\lceil K\ell_K x\rceil} \geq U\right\}.
\end{align}

\cite{gasparin_combining_exchangeable_2024} contains follow-up work that provides concrete examples of randomized p-merging functions and extends this theory to provide p-merging functions for exchangeable p-merging functions as well.

\section{Additional simulations}
\label{sec:add-sim}

\subsection{The effect of derandomization}\label{sec:derandomization-simulations}
\added{
We compare the power of \UEBH\ with its derandomized counterparts in \Cref{fig:derand-power}. The two derandomized procedures use either the average of stochastically rounded e-values (round-avg) or the compound e-values of \citet{ignatiadis_asymptotic_compound_2025} (set-avg), as we described in \Cref{remark:derandomization}, over $m$ independent runs of the \UEBH\ procedure. We can see that both derandomized procedures have lower power than \UEBH, with round-avg converging toward the power of e-BH as $m$ grows, while set-avg results in a significant drop in power even below that of e-BH and toward 0. We also see that the stability of round-avg increases with $m$, but actually decreases for set-avg. Thus, it seems that set-avg often produces e-values that fall below any thresholds that can be rejected via e-BH, and this is exacerbated as $m$ increases. Hence, derandomization via stochastic rounding seems to be the preferred method if one wants to derandomize \UEBH.
}
\begin{figure}
    \begin{subfigure}{\textwidth}
        \includegraphics[width=\textwidth]{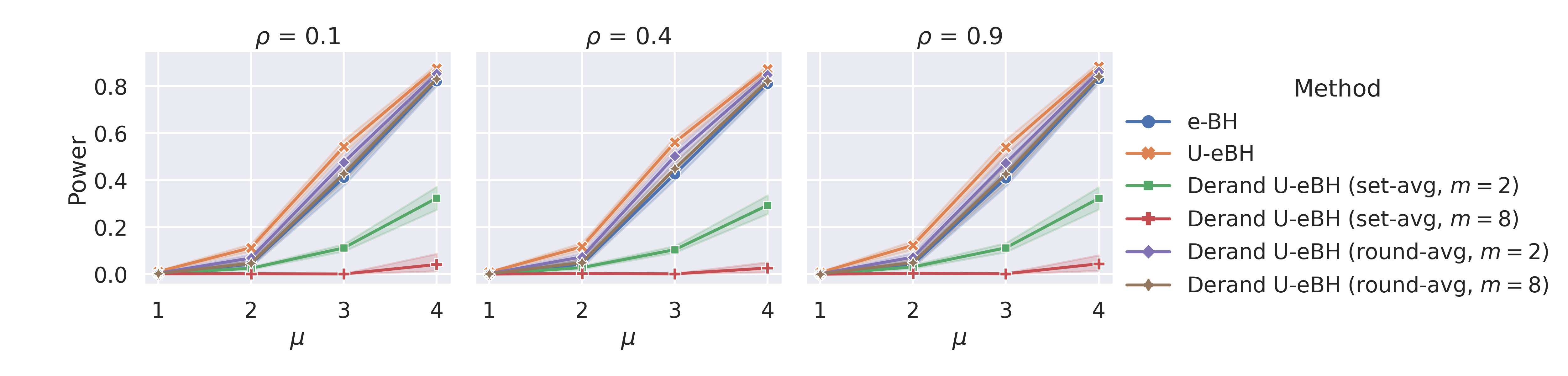}
        \caption{\added{Positive dependence}}
    \end{subfigure}
    \begin{subfigure}{\textwidth}
        \includegraphics[width=\textwidth]{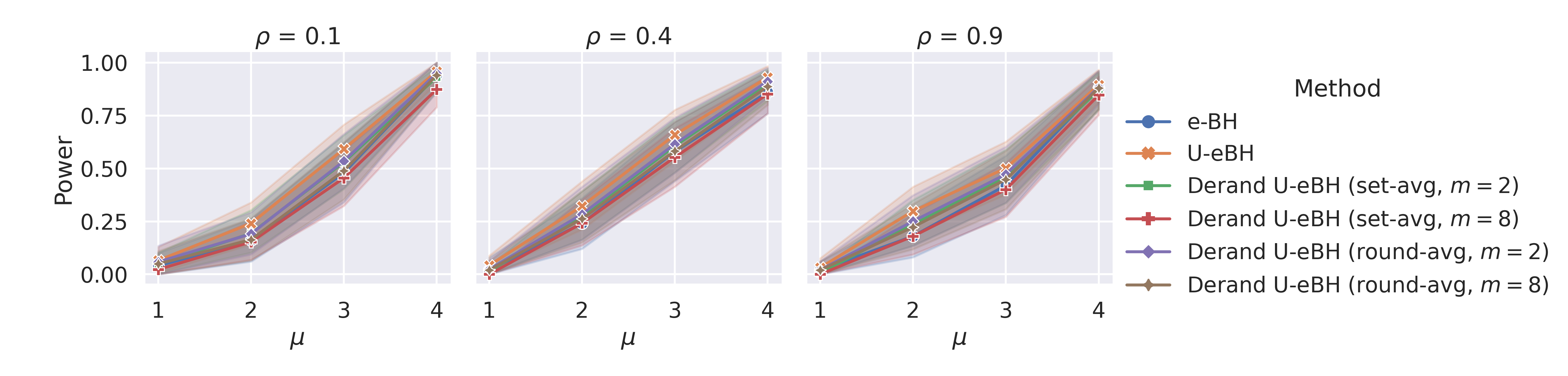}
        \caption{\added{Negative dependence}}
    \end{subfigure}
    \caption{\added{Plot of power vs.\ signal strength, $\mu$, for \UEBH\ and its derandomized counterparts using either stochastically rounded e-values (round-avg) or compound e-values \citep{ignatiadis_asymptotic_compound_2025} (set-avg) with different numbers of runs $m$. The simulation setting is the same as described in \Cref{sec:Simulations}. \UEBH\ has the highest power; the power of round-avg decreases with $m$ and approaches that of e-BH, while the power of set-avg is substantially lower and decreases towards 0 as $m$ grows.}}
    \label{fig:derand-power}
\end{figure}

\begin{figure}
    \begin{subfigure}{\textwidth}
        \includegraphics[width=\textwidth]{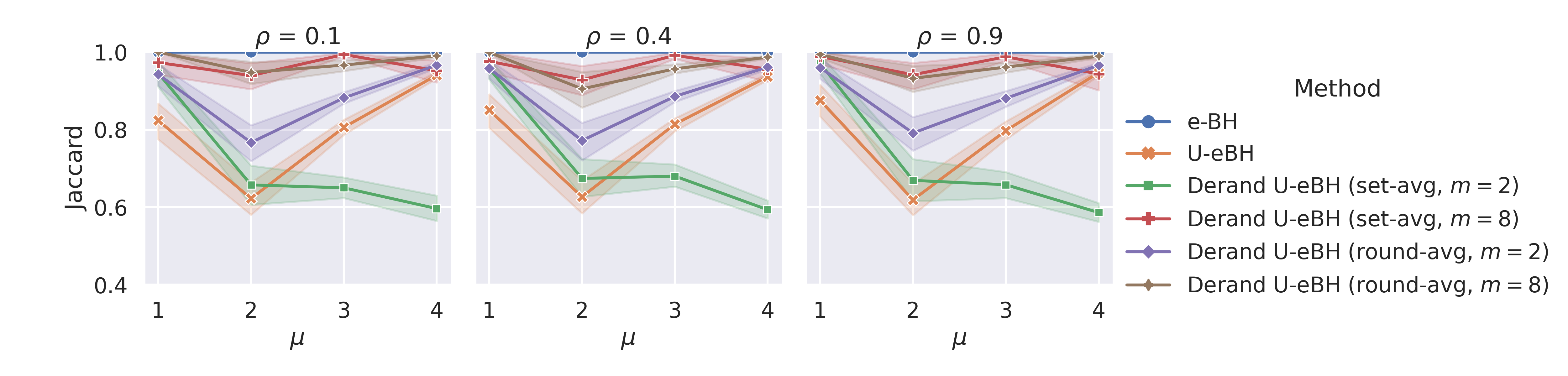}
        \caption{\added{Positive dependence}}
    \end{subfigure}
    \begin{subfigure}{\textwidth}
        \includegraphics[width=\textwidth]{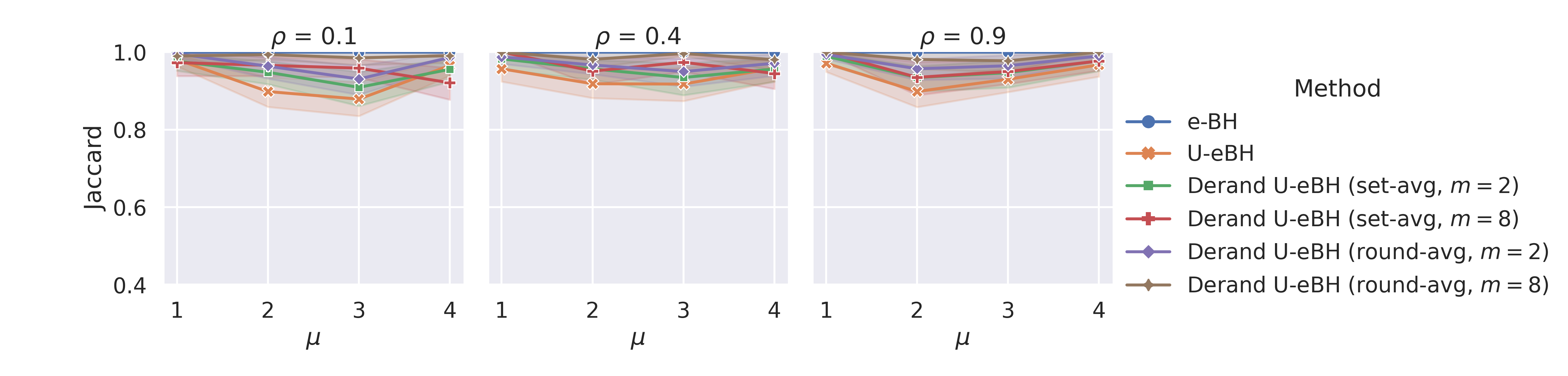}
        \caption{\added{Negative dependence}}
    \end{subfigure}
    \caption{\added{Plot of Jaccard similarity index vs.\ signal strength, $\mu$, for \UEBH\ and its derandomized counterparts. Interestingly, derandomized set-avg has lower stability as $m$ increases.}} 
    \label{fig:derand-stability}
\end{figure}

\subsection{Comparison of different randomized e-BH procedures}
We performed some additional simulations to probe the differences between methods more deeply. We can see the improvement of using \RbothEBH\ over each of the stochastic rounding approaches individually in \Cref{fig:ebh-ablation-heatmap}. Notably, \RtwoEBH\ is not uniformly dominated by \RbothEBH\ in all simulation settings, but it is beaten in the majority of them, particularly when the signal strength, $\mu$, is smaller.

\begin{figure}[h!]
    \centering
    \begin{subfigure}{0.7\textwidth}
    \centering
    \includegraphics[width=\textwidth]{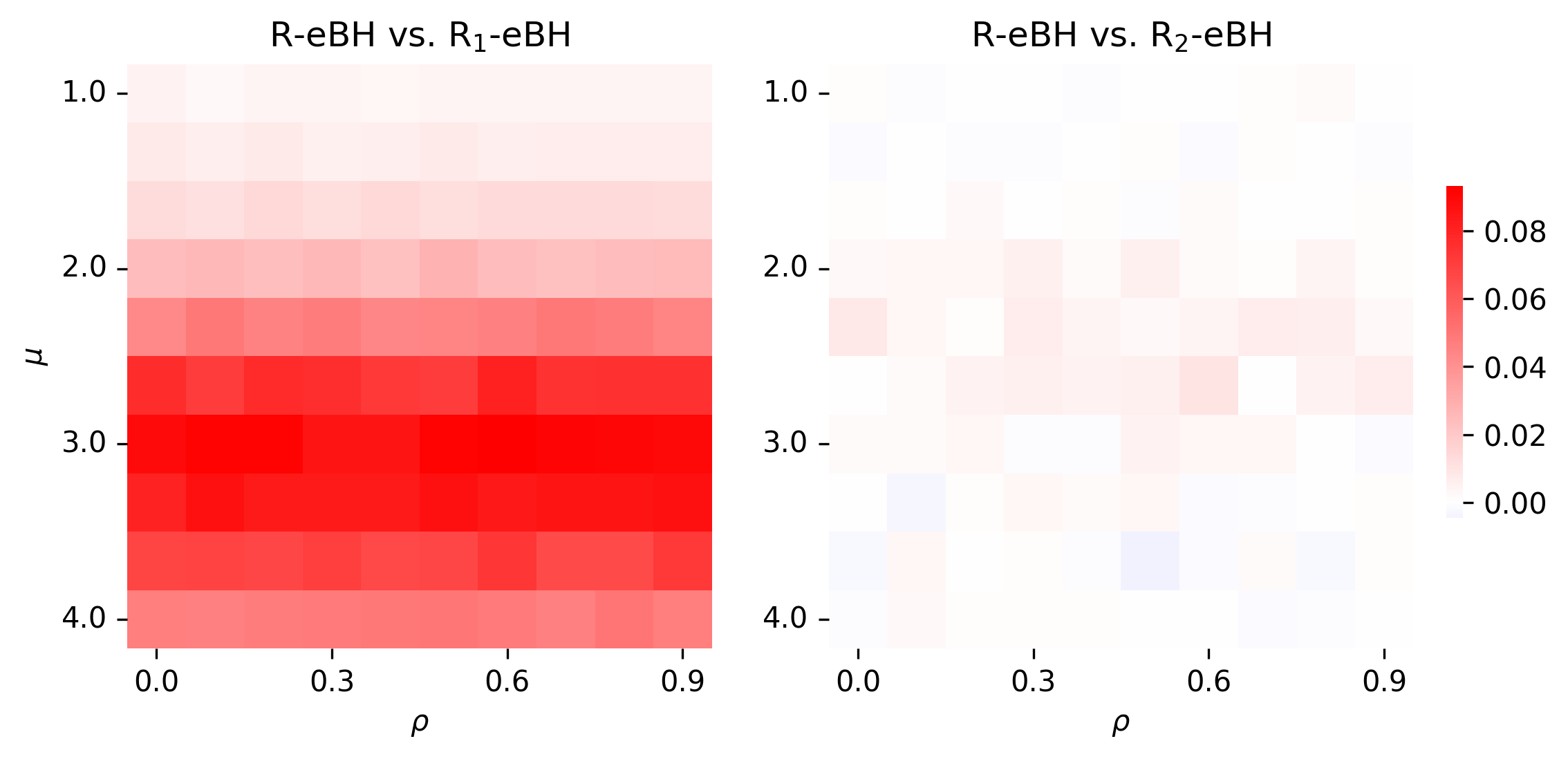}
    \caption{Positive dependence: \RbothEBH\ performs much better than \RoneEBH\ and slightly better than \RtwoEBH\ when $\mu$ is small.}
    \end{subfigure}

    \begin{subfigure}{0.7\textwidth}
    \centering
    \includegraphics[width=\textwidth]{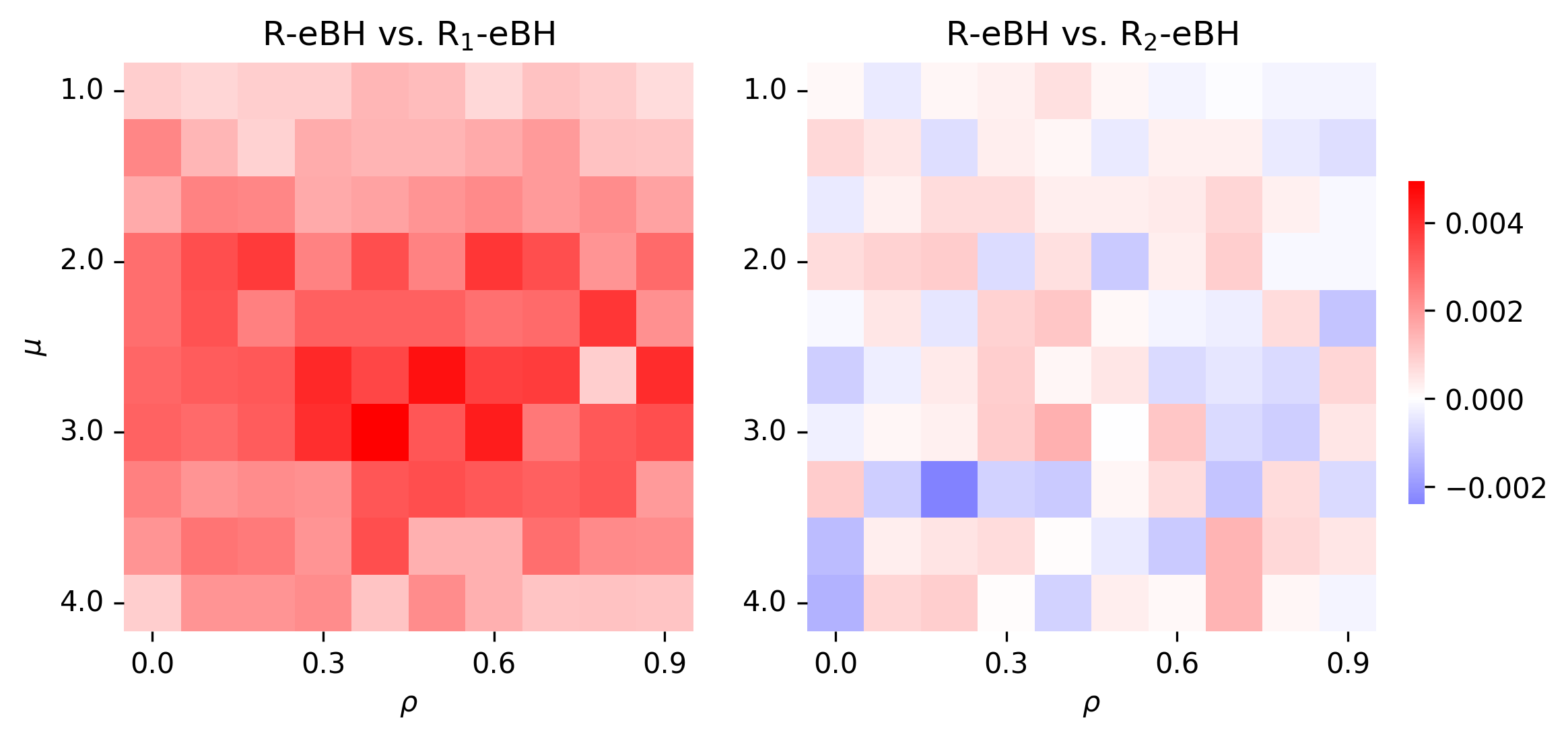}
    \caption{Negative dependence: \RbothEBH\ performs better than \RoneEBH, but neither \RbothEBH\ nor \RtwoEBH\ dominate the other.}
    \end{subfigure}
    \caption{Heatmap of the difference in power of using \RbothEBH\ over either \RoneEBH\  or \RtwoEBH across different values of $\mu$ and $\rho$ in the same simulation setting as described in \Cref{sec:Simulations}. \RbothEBH\ clearly performs better than \RoneEBH. Compared to \RtwoEBH, \RbothEBH\ performs only better in some settings.}
    \label{fig:ebh-ablation-heatmap}
\end{figure}

Another comparison to consider is the relationship between the e-BH procedure and using the BY procedure with $U_i / X_i$. We can see in \Cref{fig:ebh-ablation-trival-heatmap} that BY applied to $U_i / X_i$ has much lower power, and therefore is beaten even by baseline e-BH.

\begin{figure}[h!]
    \centering
    \begin{subfigure}{0.49\textwidth}
    \centering
    \includegraphics[width=0.8\textwidth]{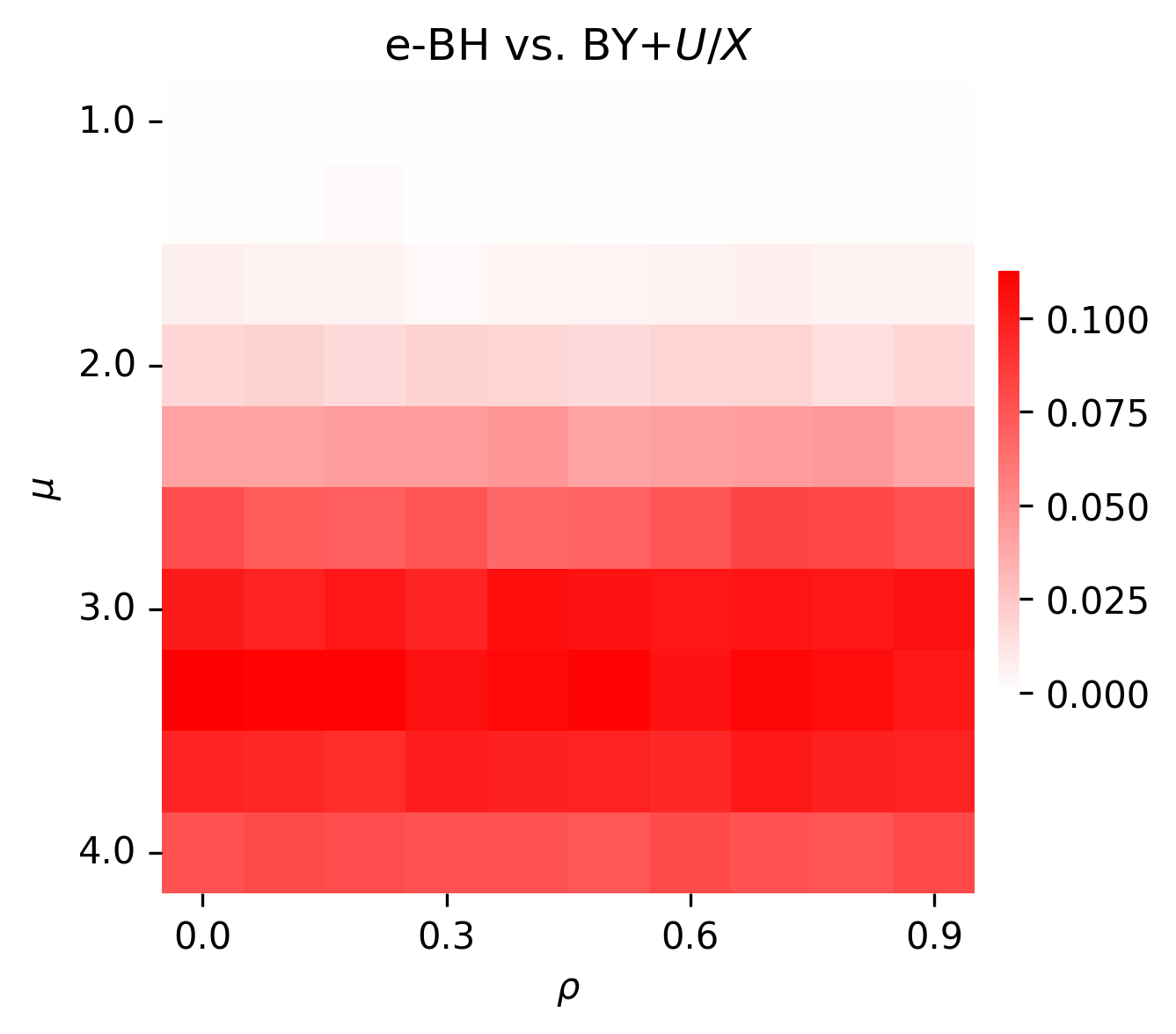}
    \caption{Positive dependence: e-BH dominates applying BY to $U / X_i$ across all $\mu$ and $\rho$.}
    \end{subfigure}\hfill\begin{subfigure}{0.49\textwidth}
    \centering
    \includegraphics[width=0.8\textwidth]{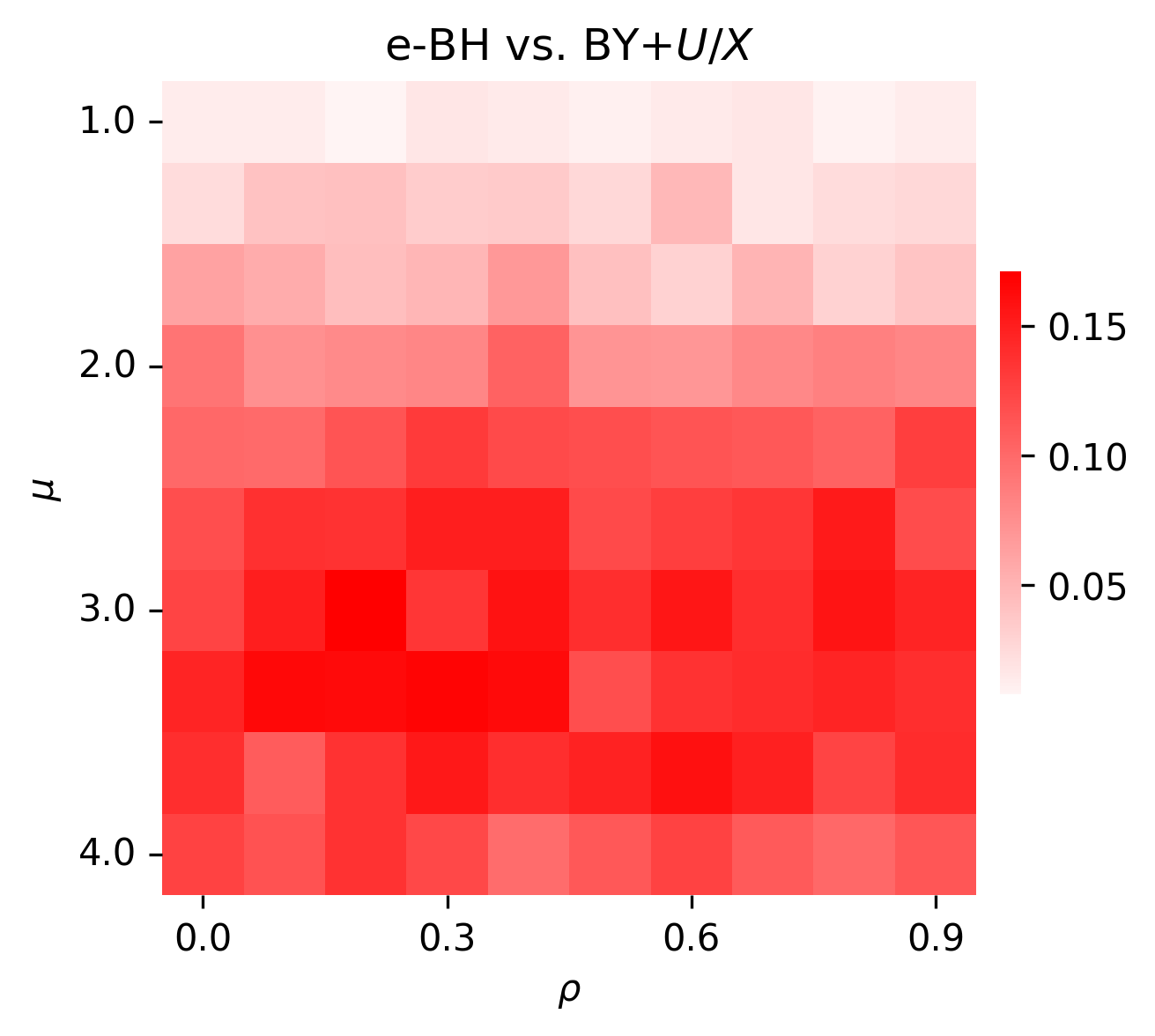}
    \caption{Negative dependence: e-BH also dominates BY with $U / X_i$ for all $\mu$ and $\rho$ here too.}
    \end{subfigure}
    \caption{Heatmap of the difference in power of using e-BH over applying BY to $U_i / X_i$ across different values of $\mu$ and $\rho$ in the same simulation setting as described in \Cref{sec:Simulations}.}
    \label{fig:ebh-ablation-trival-heatmap}
\end{figure}

We also compare how the use of a single $U_i = U$ for all uniform random variables instead of independent $U_i$ affects the performance of \RtwoEBH\ and \RbothEBH\ in \Cref{fig:u-choice}. There is no particular relationship between which one is better or worse, so it does not seem like selecting one vs.\ the other makes a significant difference in the procedure.

\begin{figure}[h!]
    \centering
        \begin{subfigure}{0.65\textwidth}
        \centering
        \includegraphics[width=\textwidth]
        {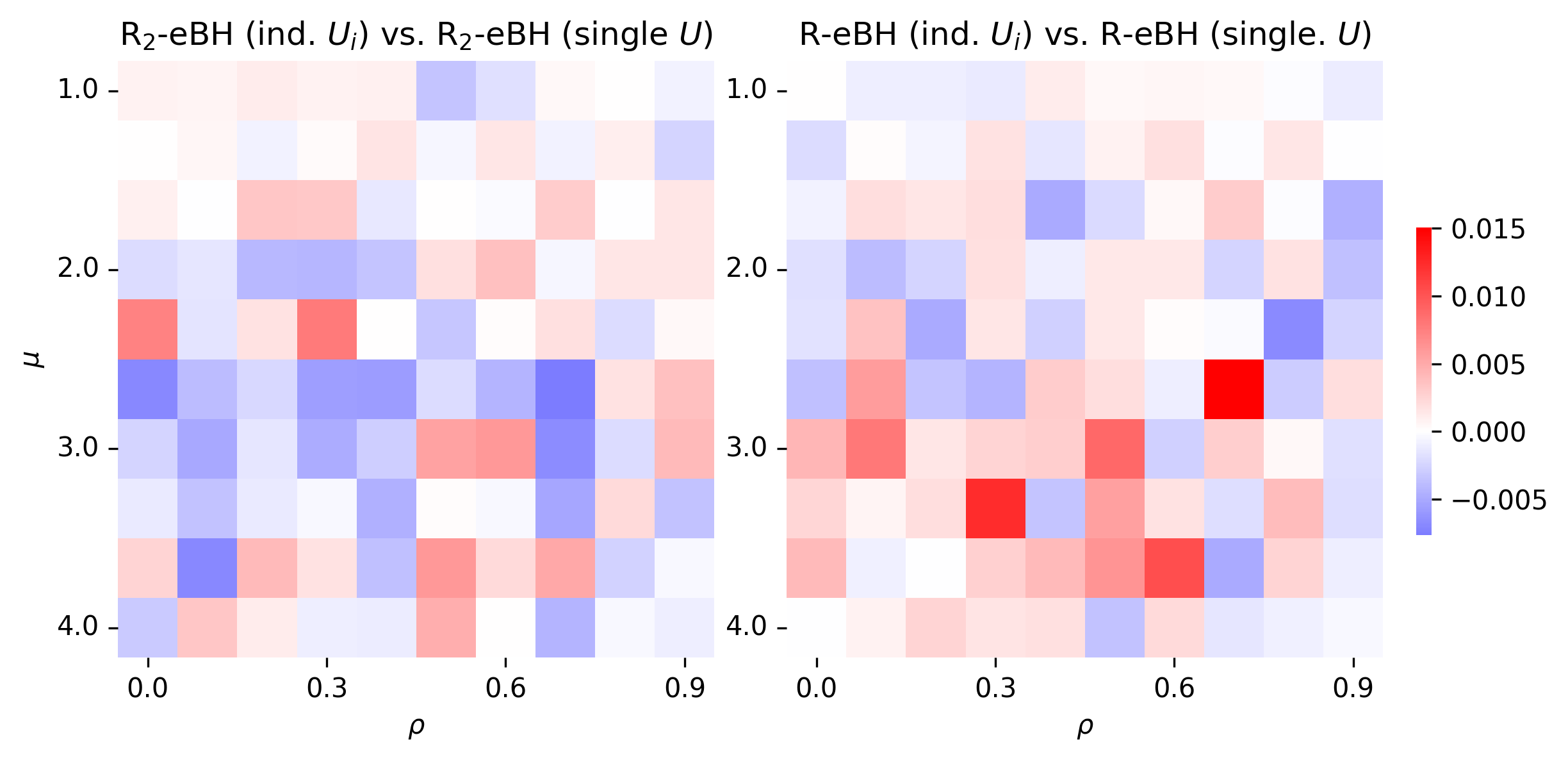}
        \caption{Positive dependence: independent $U_i$ seem to be better on average, particularly for \RbothEBH.}
        \end{subfigure}

        \begin{subfigure}{0.65\textwidth}
        \includegraphics[width=\textwidth]{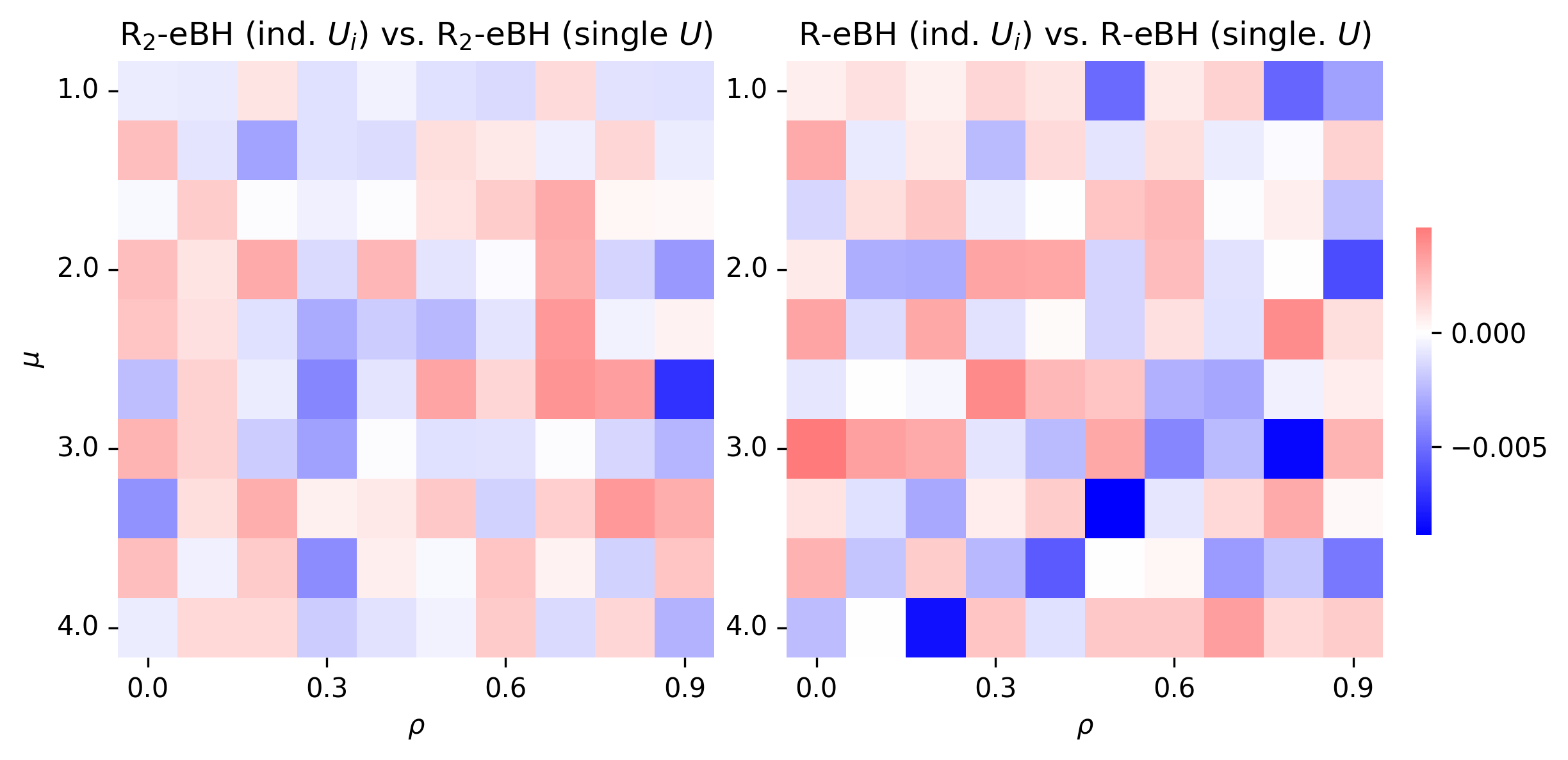}
        \caption{Negative dependence: Single $U$ seems to be better on average, particularly for \RbothEBH.}
        \end{subfigure}
        \caption{Heatmaps of the difference in power of using independent $U_i$ vs. a single $U$ (i.e., $U_i = U$ for each $i \in [K]$) across different values of $\mu$ and $\rho$ (as described in \Cref{sec:Simulations}).}
    \label{fig:u-choice}
\end{figure}
 
\end{document}